\documentclass[aps,prb,reprint,superscriptaddress,twocolumn,showkeys,amsmath,amssymb]{revtex4-2}
\usepackage[utf8]{inputenc}
\usepackage{bm}
\usepackage{graphicx}
\usepackage{amsthm}
\usepackage[english]{babel}
\usepackage{microtype}
\usepackage{mathtools}
\usepackage{dsfont}
\usepackage{mathrsfs}
\usepackage{braket}
\usepackage{hyperref}
\usepackage[shortlabels]{enumitem}
\usepackage[dvipsnames]{xcolor}
\usepackage{stmaryrd}
\usepackage{afterpage}
\usepackage{marvosym}
\usepackage{adjustbox}
\usepackage{tikz}
\usepackage{array}
\usepackage{tabularx}
\usepackage{tabularray}
\UseTblrLibrary{booktabs}

\usetikzlibrary{calc}
\usetikzlibrary{arrows.meta}

\hypersetup{
  colorlinks   = true, 
  urlcolor     = Blue, 
  linkcolor    = BrickRed, 
  citecolor   = PineGreen 
}

\newtheorem{lem}{Lemma}
\newtheoremstyle{colon}%
{}
{}
{\itshape}%
{}%
{\bfseries}%
{:}%
{ }%
{}

\theoremstyle{colon}

\newcommand{\fourvalentvertex}[5]{
  \begin{tikzpicture}[baseline=(v)]
    \coordinate (v) at (0,0);
    \coordinate (a) at (0,1);
    \coordinate (b) at (1,0);
    \coordinate (c) at (0,-1);
    \coordinate (d) at (-1,0);

    \draw[-{Latex[round,length=3mm]}, thick] (a) -- ($(a)!0.6!(v)$);
    \draw[thick] (v) -- (a) node[midway,left] {\( #1 \)};
    \draw[-{Latex[round,length=3mm]}, thick] (v) -- ($(b)!0.4!(v)$);
    \draw[thick] (v) -- (b) node[midway,above,xshift=4pt] {\( #2 \)};
    \draw[-{Latex[round,length=3mm]}, thick] (v) -- ($(c)!0.4!(v)$);
    \draw[thick] (v) -- (c) node[midway,right,yshift=-4pt] {\( #3 \)};
    \draw[-{Latex[round,length=3mm]}, thick] (v) -- ($(d)!0.4!(v)$);
    \draw[thick] (v) -- (d) node[midway,below] {\( #4 \)};

    \node at ([shift={(0.2,0.2)}]v) {\( #5 \)};
  \end{tikzpicture}
}

\newcommand{\plaquette}[6]{%
  \begin{tikzpicture}[baseline=(center)]
    \coordinate (A) at (0,0);
    \coordinate (B) at (1.5,0);
    \coordinate (C) at (1.5,1.5);
    \coordinate (D) at (0,1.5);
    \coordinate (center) at (0.75,0.75);

    \draw[thick] (A) -- (B) node[midway,below] {\( #1 \)};
    \draw[-{Latex[round,length=3mm]}] (A) -- ($(A)!0.6!(B)$);
    \draw[thick] (B) -- (C) node[midway,right] {\( #2 \)};
    \draw[-{Latex[round,length=3mm]}] (B) -- ($(B)!0.6!(C)$);
    \draw[thick] (C) -- (D) node[midway,above] {\( #3 \)};
    \draw[-{Latex[round,length=3mm]}] (C) -- ($(C)!0.6!(D)$);
    \draw[thick] (D) -- (A) node[midway,left] {\( #4 \)};
    \draw[-{Latex[round,length=3mm]}] (A) -- ($(D)!0.4!(A)$);

    \node at (center) {\( #5 \)};
    \fill (A) circle (0.07cm);
  \end{tikzpicture}%
}

\pdfpageheight\paperheight
\pdfpagewidth\paperwidth

\def\tcm{T.C.M. Group, Cavendish Laboratory, University of Cambridge, J.J. Thomson Avenue, Cambridge, CB3 0HE, UK}
\def\DAMTP{DAMTP, University of Cambridge, Wilberforce Road, Cambridge, CB3 0WA, UK}

\begin{document}

\title{Classifying one-dimensional Floquet phases through two-dimensional topological order}
\author{Campbell McLauchlan}
\affiliation{Centre for Engineered Quantum Systems, School of Physics, University of Sydney, Sydney, NSW 2006, Australia}

\author{Vedant Motamarri}
\affiliation{\tcm}

\author{Benjamin B\'eri}
\affiliation{\tcm}
\affiliation{\DAMTP}

\begin{abstract}
Floquet systems display rich phenomena, such as time crystals, with many-body localisation (MBL) protecting the phases from heating. While several types of Floquet phases have been classified, a unified picture of Floquet MBL is still emerging.
Static phases have been fruitfully studied via ``symmetry topological field theory" (SymTFT), wherein the universal features of $G$-symmetric systems are elucidated by placing them on the boundary of a topological order of one dimension higher. In this work, we provide a SymTFT approach to classifying $G$-symmetric Floquet MBL phases in 1D, for $G$ a finite Abelian group with on-site unitary action. In the SymTFT, these 1D systems correspond to the boundaries of the quantum double associated to $G$, and the classification naturally arises from considering the Lagrangian subgroups and boundary excitations of the quantum double. The classification covers all known Floquet phases while uncovering others previously unexplored, along with bulk features of phases thought to have only boundary signatures. We refer to the latter phases as ``dual" time crystals.
For static phases, we show how anyons of the quantum double and (string) order parameters provide a natural and simple interpretation of known classification schemes. By extending our framework to the boundaries of twisted quantum doubles, we uncover a new time-crystalline phase with non-onsite symmetry, which cannot be obtained through local, symmetric Hamiltonian drives. 
We numerically demonstrate evidence for the absolute stability of this phase, and observe that for open boundary conditions it has greater stability to symmetric perturbations.
We finally discuss perspectives on using programmable quantum devices to realise and probe the phases we discuss. 
Our results show that SymTFT provides a powerful approach to unifying phases and features of Floquet systems.
\end{abstract}

\maketitle

\section{\label{sec:intro} Introduction}

Periodically driven systems display rich phase structures often without analogues in zero-temperature or equilibrium settings. Such phases can display Floquet Majorana fermions~\cite{Floquet_Majoranas_PRL_2011}, spontaneous symmetry breaking (SSB)~\cite{Khemani_2016,KeyserlingkSondhi2016floquetSSB} and symmetry-protected topological (SPT) phases~\cite{potter2016prx,KeyserlingkSondhi2016floquetSPT,else2016prb}, anomalous edge states~\cite{rudner2013anomalous}, topological insulation~\cite{roy2017prb}, and conformal field theory~\cite{Floquet_CFT}, just to name a few examples. One of the striking characteristics of certain driven systems is the breaking of discrete time-translation symmetry, i.e., the emergence of ``time crystals"~\cite{Wilczek2012quantumTC,ShapereWilczek2012classicalTC,WatanabeOshikawa2015TC,Khemani_2016,else2016timecrystal,KeyserlingkSondhi2016floquetSPT,KeyserlingkSondhi2016floquetSSB,KeyserlingkKhemaniSondhi2016stability,KhemaniKeyserlingkSondhi2017RepTh,khemani2019briefhistorytimecrystals,Sacha2018,Else2020}. While driven systems are normally expected to evolve to a trivial, infinite-temperature state, many-body localisation (MBL) can allow them to avoid this thermalisation~\cite{Fleish80,Gornyi2005,basko2006metal,Huse_localization_2013,serbyn2013local,Huse_MBL_phenom_14,chandran2015constructing,ros2015integrals,Rademaker2016LIOM,Wahl2017PRX,Goihl2018,NandkishoreHuse_review,AltmanReview,Abanin2017,Alet2018,ImbrieLIOMreview2017}. Recent work has considered theoretical proposals for demonstrating Floquet phases and features in experiments~\cite{Koyama_Floquet_rotor_PRR_2023,Wahl_2024}, and demonstrated the existence of time crystals in quantum processors~\cite{mi2022time,Xiang_Top_TC_Experiment_2024}.

The classification of Floquet phases is critical for understanding their range of features. Previous work has provided a classification of $G$-symmetric Floquet SSB phases~\cite{KeyserlingkSondhi2016floquetSSB} and Floquet SPT~\cite{KeyserlingkSondhi2016floquetSPT,potter2016prx,else2016prb} phases by including discrete time-translation among the system's symmetry in the analysis. However, more work can be undertaken to provide a unified picture of all phases demonstrated by a Floquet MBL system under a symmetry group $G$. Additionally, while Floquet SPT phases have been examined from the perspective of their boundary signatures, as we show, they also possess important bulk signatures, which have not been previously explored.

Recently, we have shown that the approach of symmetry topological field theory (SymTFT)~\cite{aasen2016topological,Aasen_cat,JiWen2020categorical,Kong_holographic_entanglement_2020,lichtman2020bulk,Kaidi22,Bhardwaj_PtII,ChatterjeeWen23,TH23,Huang_Chen_criticality,WenPotter_SPTgapless,bhardwaj2023club,Lootens_1D,Lootens23,Lootens21,albert2021spin,gaiotto2021orbifold}, while usually used to study ground states, can provide a fruitful unifying picture of driven 1D systems, that can elucidate known and uncover novel phenomena~\cite{motamarri2024symtftequilibriumtimecrystals}. In such an approach, 1D symmetric systems are placed on the boundary of a fictitious 2-dimensional ``bulk" system. Even in $\mathbb{Z}_2$-symmetric systems, previously unknown features become apparent, including ``dual" time crystals, charge pumping between a system and its dynamical boundary conditions, and connections to Floquet-enriched topological order~\cite{PotterMorimoto2017FSET,po2017radical}, including  Floquet codes~\cite{Hastings2021dynamically,Vuillot21,DavydovaPRXQ23,kesselring2022anyon,Vu2024Measurement_Induced,sullivan2023floquet,aasen2023measurement,ellison2023floquet,davydova2023quantum,dua2023engineering,aitchison2024competingautomorphismsdisorderedfloquet}.

In this work, we approach the problem of classifying $G$-symmetric 1D Floquet phases via the SymTFT framework. We provide a novel classification scheme for such phases, a scheme that naturally arises when viewing the 1D system to be the boundary of a 2D bulk. Our scheme encompasses previous approaches to classifying SSB and SPT Floquet phases while also uncovering previously unexplored phases, which are partially-$G$-symmetry-broken and combine regular and dual time-crystalline signatures. 
We consider $G$ to be an Abelian symmetry group, and mostly focus on cases where it has on-site, unitary action. Our focus on Abelian $G$ is natural in the driven setting of interest, as non-Abelian symmetry groups are not compatible with symmetry-preserving MBL and hence are unstable to heating~\cite{potter2016prx,Vasseur2016}. While we only consider unitary representations of $G$, this already presents new features beyond just a consolidation of previous classification schemes.

To begin, we consider the 1D gapped (more generally MBL) boundaries of 2D topologically ordered systems, namely those of the quantum double associated to the symmetry group $G$ (or $G$-TO, for short). These boundaries are classified by the Lagrangian subgroups of the $G$-TO~\cite{Kapustin_2011,Levin_2013,Barkeshli13a}, certain subgroups of the $G$-TO's anyonic excitations. The corresponding gapped boundaries, in turn, classify the gapped phases of 1D $G$-symmetric systems~\cite{aasen2016topological,Aasen_cat,JiWen2020categorical,Kong_holographic_entanglement_2020,lichtman2020bulk,Kaidi22,Bhardwaj_PtII,ChatterjeeWen23,TH23,Huang_Chen_criticality,WenPotter_SPTgapless,bhardwaj2023club}. Here we show that this approach also allows one to classify MBL and Floquet phases of such 1D systems, and furnishes new insights even for ground-state settings. 

We discuss the possible SSB, partially symmetry-broken, and SPT phases. For static SPT phases, we provide a novel and simple interpretation of the classification scheme of 1D SPT phases via the anyons of the 2D bulk. We then consider unitary drives, which leads us to our classification of Floquet phases: we establish a classification in terms of a Lagrangian subgroup of the $G$-TO, $\mathcal{M}$, along with an ``excitation class", labelled by some anyonic excitation $b\notin \mathcal{M}$. We find many phases with spontaneous time-translation symmetry breaking (time-crystallinity) along with phases that we classify as ``dual" time crystals, wherein the time crystallinity is diagnosed by local order parameters charged under a dual symmetry. 
The duality between charges and boundary conditions~\cite{gaiotto2021orbifold,Bhardwaj_PtII,Lootens_1D} is central in understanding many features of these phases. In this way we capture all known MBL Floquet phases of 1D $G$-symmetric systems, while also finding phases that were not previously discussed.

We also consider 1D systems with non-onsite symmetry action. Non-onsite symmetries have been observed to entail rich physics in static systems~\cite{SPT_gapless_edge_Wen,JiWen2020categorical,Wen_2013_GaugeAnomalies,zhang2024NonOnsiteSSB}, but have not been explored in driven settings.  To study them in the SymTFT framework, we consider the bulk to be a twisted quantum double (TQD)~\cite{Hu_2013,Levin-Wen,Levin_gu_DS_model,JiWen2020categorical,Ellison_2022_TQDs}, a generalisation of the $G$-TO mentioned above.
The boundaries give rise to novel driven 1D systems, and to illustrate their key feature we focus on the simplest case, $G=\mathbb{Z}_2$, where the corresponding 2D bulk is the double semion model~\cite{Levin-Wen}. The system we explore can exhibit time-translation symmetry breaking, but with fewer Floquet MBL phases than with onsite symmetry and with unusual differences between the cases with closed and open boundary conditions. %

This paper is structured as follows. In Section~\ref{sec:Illustrative_example}, we illustrate the ideas of this paper using a simple example, the $\mathbb{Z}_2$-symmetric, driven transverse-field Ising model. In Section~\ref{sec:G_TO}, we then introduce quantum double models based on Abelian groups, which will be used to study 1D phases with on-site symmetries. We also highlight the ``boundary algebra" which will be placed in correspondence with the 1D system of interest. In Section~\ref{sec:Static_phase_SymTFT} we then review the classification of static phases via Lagrangian subgroups. We provide a simple interpretation of the static phase classification scheme based on anyons of the quantum double model and (string) order parameters. In Section~\ref{sec:SymTFT_Floquet_phases}, we provide the SymTFT classification of Floquet phases. We introduce a specific form of fixed-point  (i.e., zero-correlation-length) drives,
and show that the scheme we introduce captures known and previously unknown phases. We discuss boundary conditions, relate our approach to existing classification schemes, demonstrate the time-translation symmetry breaking signatures of many of the phases, and consider some illustrative examples. In Section~\ref{sec:TQD_phases}, we extend this framework by investigating Floquet phases at the boundary of the double semion model. Finally, in Section~\ref{sec:Experimental_demo}, we discuss some experimental considerations. In Section~\ref{sec:Conclusion}, we conclude and provide avenues for future work.

\section{Illustrative example: Driven Ising model}\label{sec:Illustrative_example}

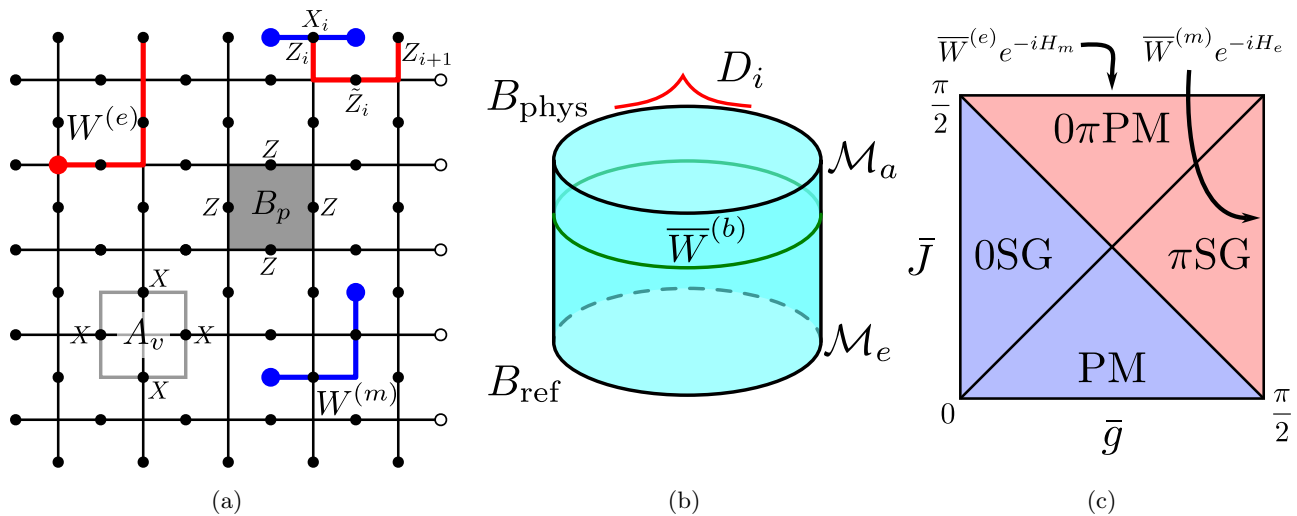
\begin{figure*}
    \centering
    \begin{tikzpicture}
\node[inner sep=0pt] at (-4,0)         {\includegraphics[width=0.95\linewidth]{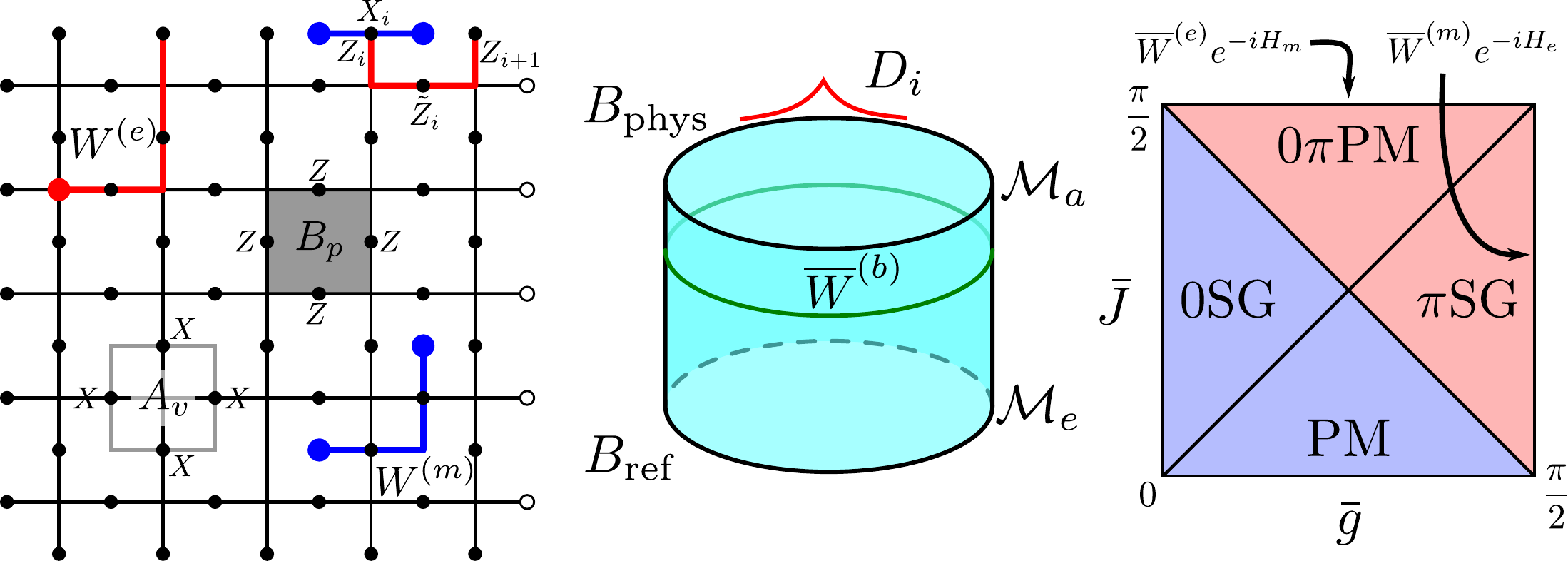}};
\node[inner sep=0pt] at (-9.63,-3.5) {(a)};
\node[inner sep=0pt] at (-3.55,-3.5)    {(b)};
\node[inner sep=0pt] at (2,-3.5)    {(c)};
\end{tikzpicture}
\caption{$\mathbb{Z}_2$-symmetric Floquet phases on the boundary of the toric code TO. (a) Toric code with left and right boundaries identified and top and bottom boundaries $e$-condensing. An example of a vertex operator ($A_v$) and a plaquette operator ($B_p$) are shown. String operators $W^{(m)}$ and $W^{(e)}$, and some members of the boundary algebra are shown. (b) Schematic of the $\mathbb{Z}_2$-TO on the cylinder with the top (bottom) boundary, $B_\text{phys}$ ($B_\text{ref}$), gapped according to Lagrangian subgroup $\mathcal{M}_a$ ($\mathcal{M}_e$), for anyon $a\in \lbrace e,m\rbrace$. Time-translation symmetry breaking (TTSB) is generated by the inclusion of string operator $\overline{W}^{(b)}$, for $b\notin \mathcal{M}_a$, in the drive. For strong disorder, local integrals of motion pick up exponentially decaying tails. This is schematically illustrated for LIOMs $D_i$, with support shown in red. (c) The phase diagram
of the driven Ising model. $\bar{J}$ and $\bar{g}$ are the means of disordered parameters $J_i$ and $g_i$ (see Eqs.~\ref{eqn:H_0_ZZ} and \ref{eqn:H_1_X}). We identify two TTSB fixed-point drives, relating them to drives from the model in (b) (see Eqs.~\ref{eqn:piSG_U} and \ref{eqn:0piPM_U}).}
\label{fig:classification_illustration}
\end{figure*}

We here explain our classification scheme using the example of a $\mathbb{Z}_2$-symmetric system. We begin by introducing $\mathbb{Z}_2$-topological order (TO) via the toric code phase. We then introduce the 1D system of interest, the driven transverse-field Ising model, before providing its SymTFT description by embedding it on the boundary of a toric code phase on a cylinder. We enumerate the fixed-point drives of the chain and explain how to move away from these fixed points to uncover the whole phase diagram, with phases stabilised by MBL. We then discuss symmetry-breaking perturbations and the absolute stability of the phases.

\subsection{\texorpdfstring{$\mathbb{Z}_2$-topological order: the toric code}{Z2-topological order: the toric code}}

The $\mathbb{Z}_2$-topologically ordered phase is modelled by Kitaev's toric code~\cite{kitaev2003fault,BravyiKitaev_SC}.
This is a system described by a local 2D Hamiltonian $H_{\text{toric}}$, defined on a square lattice with qubits placed on the edges of the lattice, see Fig.~\ref{fig:classification_illustration}(a). To define $H_{\text{toric}}$, we first define ``vertex" and ``plaquette" operators respectively:
\begin{align}
    A_v &= \prod_{j\in v} X_j\\
    B_p &= \prod_{j\in p} Z_j
\end{align}
where $v$ is the set of qubits adjacent to the vertex labeled by $v$, and $p$ is the set of qubits around the corresponding plaquette (a square in the lattice), as shown in Fig.~\ref{fig:classification_illustration}(a). The Hamiltonian is then defined as:
\begin{align}
    H_\text{toric} = -\sum_v A_v - \sum_p B_p.
\end{align}
The ground states are $+1$ eigenstates of $A_v$ and $B_p$ for all $v,p$. The system is topologically ordered: it has non-trivial ground state degeneracy determined by the boundary conditions of the manifold on which the lattice sits, non-trivial topological entanglement entropy, and anyonic excitations~\cite{kitaev2003fault,Kitaev2006,Kitaev_2006}.

We next describe these excitations. A ``string operator" $W_\gamma^{(e)} = \prod_{j\in \gamma} Z_j$ is a product of single-qubit Pauli operators along a path $\gamma$ of edges in the lattice, starting at vertex $v_1$ and terminating at $v_2$. Such an operator commutes with all $B_p$ and all $A_v$ away from the end-points of $\gamma$, but anti-commutes with $A_{v_1}$ and $A_{v_2}$. We interpret $W_\gamma^{(e)}$ as creating a pair of quasiparticles, called $e$ particles, located on $v_1$ and $v_2$. These quasiparticles are ``anyons", as we shall explain below. Now let us define
the dual $\widetilde{\mathcal{L}}$ of a lattice $\mathcal{L}$ by assigning a vertex in $\widetilde{\mathcal{L}}$ to each plaquette $p$ in $\mathcal{L}$, and connecting vertices in $\widetilde{\mathcal{L}}$ by edges if and only if the corresponding plaquettes in $\mathcal{L}$ are adjacent, which leaves the faces in $\widetilde{\mathcal{L}}$ associated with the vertices of $\mathcal{L}$. We then similarly define the ``dual string operator" $W_{\tilde{\gamma}}^{(m)} = \prod_{j\in\tilde{\gamma}} X_j$ to be the product of $X$ operators along a path $\tilde{\gamma}$ in the dual lattice. This operator creates $m$ particles on the plaquettes at the endpoints of $\tilde{\gamma}$. Examples of these quasiparticles are shown in Fig.~\ref{fig:classification_illustration}(a).

The $m$ and the $e$ particles are bosons,
since the state picks up no phase upon the exchange of two $e$ (equiv. $m$) particles. Denoting the exchange phase associated to anyon $a$ by $e^{i\theta_a}$ we thus have $e^{i\theta_e} = e^{i\theta_m} = +1$. However, an $e$ and $m$ particle braid non-trivially with one another: encircling one by the other, the state picks up a $-1$ phase due to the anti-commutation of the $m$ and $e$ string operators. We summarise this with the braiding phase $\mathcal{B}(e,m) = -1$. Along with these two quasiparticles, the toric code phase has two other anyons, the ``vacuum particle" $\mathbf{1}$, and the fermion $f \coloneqq e\times m$, which is obtained by fusing the $e$ and $m$ particles.

The lattice can be given periodic boundary conditions, i.e., defined on a torus, in which case the ground state is degenerate. String operators $\overline{W}^{(e)}_\gamma$ or $\overline{W}^{(m)}_{\tilde{\gamma}}$ that wrap around the torus handle commute with the Hamiltonian but act non-trivially on this space of ground states.
We will be considering the TO systems on a cylinder. The open boundaries at the top and bottom of the cylinder will be gapped [as achieved by using suitable boundary plaquette and vertex operators, as exemplified in Fig.~\ref{fig:classification_illustration}(a)]. (In the context of Floquet phases, we will refer to MBL, rather than gapped, boundaries, since there is no notion of a ground state in such systems.)
There are two distinct gapped boundaries of the toric code phase, which differ by the anyons that are \textit{condensed} at the boundary. 
If an $e$ ($m$) particle can be condensed at a boundary, it means we can terminate a string operator $W_\gamma^{(e)}$ ($W_{\tilde{\gamma}}^{(m)}$) on that boundary without it anti-commuting with any vertex (plaquette) operators. An example of such a $W_\gamma^{(e)}$ string operator terminating at an $e$-condensing boundary is shown in Fig.~\ref{fig:classification_illustration}(a). 
The anyons that are condensed at a particular boundary form a ``Lagrangian subgroup" and the gapped boundaries of a TO system are in one-to-one correspondence with these~\cite{Kapustin_2011,Levin_2013,Barkeshli13a}.
(We will define Lagrangian subgroups in the following section.)
For the toric code phase, there are only two Lagrangian subgroups: 
$\mathcal{M}_e = \lbrace \mathbf{1},e\rbrace$
and $\mathcal{M}_m = \lbrace \mathbf{1}, m\rbrace$.  These are both closed under the fusion operation because $e\times e = m\times m = \mathbf{1}$.

We now define the setup of interest: a toric code phase defined on a cylinder [see Fig.~\ref{fig:classification_illustration}(b)]. The bottom boundary of this system will be a ``reference boundary", $B_\text{ref}$. We will set $B_\text{ref}$ to be $e$-condensing, described by $\mathcal{M}_e$. Meanwhile the top boundary will ultimately describe the physical 1D system of interest. We will refer to that as the ``physical boundary", or $B_\text{phys}$.

\subsection{Driven Ising chain}

Here we introduce the physical $\mathbb{Z}_2$-symmetric 1D system we shall aim to describe via SymTFT, namely the driven Ising chain. This is a 1D chain of $L$ spins (or qubits) undergoing evolution by a Floquet unitary $U_F = \mathcal{T}\text{exp}(-i\int_0^T H(t)dt)$. We mostly consider two-step drives, where $H(t) = H_0/(T/2)$ for $0\leq t < T/2$, and $H(t) = H_1/(T/2)$ for $T/2\leq t < T$, with $H_0$ and $H_1$ defined below. Hence, we have:
\begin{align}
    U_F &= e^{-i H_1}e^{-i H_0},\\
    H_0 &= \sum_{j=1}^{L} J_j Z_j Z_{j+1},\label{eqn:H_0_ZZ}\\
    H_1 &= \sum_{j=1}^L g_j X_j.\label{eqn:H_1_X}
\end{align}
Here
we consider only the simplest Ising-symmetric terms in the Hamiltonians (without more complicated interactions). We also assume periodic boundary conditions (defining $Z_{L+1}\equiv Z_1$).
$U_F$ generates the periodic time evolution of the system. 
Due to this periodic nature of the driving, the system possesses discrete time-translation symmetry (TTS).
It is also symmetric under the $\mathbb{Z}_2$ symmetry generated by $P \coloneqq \prod_j X_j$. We choose the parameters $J_j$ and $g_j$ randomly about their respective means, $\bar{J}$ and $\bar{g}$, pulling them from uniform distributions with widths $\delta J$ and $\delta g$, respectively, with $\delta J \gg \delta g$ (or vice versa, depending on the phase of the system we consider), so that the eigenstates of $U_F$ are many-body localised~\cite{Fleish80,Gornyi2005,basko2006metal,Huse_localization_2013,serbyn2013local,Huse_MBL_phenom_14,chandran2015constructing,ros2015integrals,Rademaker2016LIOM,Wahl2017PRX,Goihl2018,NandkishoreHuse_review,AltmanReview,Abanin2017,Alet2018,ImbrieLIOMreview2017}.

This system has four Floquet phases, shown in Fig.~\ref{fig:classification_illustration}(c), which are exemplified by four distinct limits in phase space~\cite{Khemani_2016},
corresponding to four ``fixed-point" (i.e., zero-localisation-length) drives. Unlike zero-temperature phases, the features of these driven phases will be present in all eigenstates of $U_F$. Such eigenstates will have corresponding eigenvalues of the form $e^{i\varepsilon}$, for ``quasienergy" $\varepsilon$, since $U_F$ is unitary. 

To begin, we have the phases that are analogous to the phases of the un-driven Ising chain: these are the spin glass (SG) and paramagnet (PM). The fixed-point drive for the former is obtained by setting all $g_i = 0$; the eigenstates of $U_F$ hence become those of $H_0$: we have doubly-degenerate eigenspaces with eigenstates labeled by the values of $Z_j Z_{j+1}$ and of $P$. This state exhibits long-range order and spontaneous symmetry-breaking, as detected in correlations of the order parameter $Z_j$. The PM phase, by contrast, can be obtained by setting all $J_i = 0$. Now the eigenstates of $U_F$ become those of $H_1$, and these are paramagnetic. 

We now cover the non-trivial Floquet phases. First, let us set $g_i = \frac{\pi}{2}$
for all $i$. Then $U_F \propto P e^{-i H_0}$. While the eigenstates of this operator are the same as those of the SG phase, this phase is very distinct. Consider state $\ket{p,\, \{s_i\}}$ for $P\ket{p,\, \{s_i\}}=p\ket{p,\, \{s_i\}}$ and $Z_i Z_{i+1}\ket{p,\, \{s_i\}} = s_i \ket{p,\, \{s_i\}}$. Then $U_F \ket{p,\, \{s_i\}} = p e^{-i\sum_j J_j s_j}\ket{p,\, \{s_i\}}$ and hence the two states $\ket{p=+1,\, \{s_i\}}$ and $\ket{p=-1,\, \{s_i\}}$ 
have quasienergies that differ by $\pi$ -- this is called $\pi$-spectral pairing. This results from time-translation symmetry breaking (TTSB) exhibited by this phase~\cite{Wilczek2012quantumTC,ShapereWilczek2012classicalTC,WatanabeOshikawa2015TC,Khemani_2016,else2016timecrystal,KeyserlingkKhemaniSondhi2016stability,KhemaniKeyserlingkSondhi2017RepTh}: we note that observables display period-doubling, $Z_j(t) = (-1)^tZ_j(0)$ for $t\in \mathbb{Z}$ counting the number of Floquet periods undergone by the system. %
This phase is referred to as the $\pi$-spin glass phase ($\pi$SG). 

Finally, we can take $J_i = \frac{\pi}{2}$. So $U_F \propto \prod_j Z_j Z_{j+1} e^{-i H_1} = Z_1 Z_L e^{-iH_1}$, assuming we have open boundary conditions. Now we find TTSB on the boundary of the system: $X_1(t) = (-1)^t X_1(0)$ for $t\in \mathbb{Z}$. Hence, this phase also exhibits $\pi$-spectral pairing between eigenstates differing in the eigenvalue of $Z_1 Z_L$. We will refer to the system as a dual time crystal, for reasons that will be clear later. Meanwhile, there is degeneracy between the two states with equivalent eigenvalue of $Z_1 Z_L$. This phase is referred to as the $0\pi$PM phase. With closed boundary conditions, it seems indistinguishable from the PM phase, since the non-trivial time-periodicity and spectral pairing occurs on the boundary of the system. However, as we will show, with closed but dynamical boundary conditions this phase does strikingly differ from the PM phase.

\subsection{SymTFT description of the driven Ising chain}

Here we describe how the boundary conditions at $B_\text{phys}$ for the toric code on a cylinder, together with a suitable topological drive ingredient, correspond to the phases in the driven Ising model. We describe this correspondence in terms of the actual Hilbert spaces, and in terms of categorical symmetries and order parameters.

As described above, we let the toric code phase be defined on a cylinder, with $B_\text{ref}$ condensing $e$ anyons [see Fig.~\ref{fig:classification_illustration}(b)]. We consider the subspace of the Hilbert space in which we have no anyons in the bulk of the 2D system. Physically, we can imagine the bulk (and $B_\text{ref}$) to have an energy gap far larger than any energy scale associated with the dynamics on $B_\text{phys}$.
We will show that the four fixed point drives introduced in the previous section can be obtained for $B_\text{phys}$ by choosing for it a Lagrangian subgroup and an ``excitation class". The fact that the $B_\text{phys}$ Hilbert space can be identified with that of the Ising chain can be made explicit, by noting two fixed point limits. First, consider $B_\text{phys}$ to be static and gapped, with Lagrangian subgroup $\mathcal{M}_e$. Then the plaquette operators are modified along $B_\text{phys}$ to now be weight-3 [see Fig.~\ref{fig:classification_illustration}(a)]. If we label the sites along $B_\text{phys}$ by $i\in [1,L]$ and use $\tilde{Z}_i$ to refer to the operator acting on the qubit in the row below, in between $i$ and $i+1$, then the boundary plaquette operators become:
\begin{align}
    B_p^\text{bound} = Z_i \tilde{Z}_{i}Z_{i+1}.
\end{align}
This boundary operator can be understood as a string operator, which we denote $W^{(e)}_{(i,i+1)}$, running between nearest-neighbor sites in the lattice.

Meanwhile, if we project the boundary qubits into $X_i$ eigenstates, the vertex operators along the boundary become weight-3, since we can replace the stabiliser $A_v$ with the weight-3 term $A_v^\text{bound} = A_v X_i$, with $i$ the site immediately above $v$ [see Fig.~\ref{fig:classification_illustration}(a)]. Now the boundary condenses $m$ anyons. Hence these two gapped boundaries, described by $\mathcal{M}_e$ and $\mathcal{M}_m$ respectively, correspond exactly to the two fixed point limits of the transverse-field Ising chain, upon setting the ``link" degrees of freedom $\tilde{Z}_{i} = +1$. (We shall later also consider retaining the $\tilde{Z}_{i}$ and reinterpreting them as gauge degrees of freedom related to the boundary conditions of the 1D system.)
These $X_i$ operators also correspond to string operators, which we denote $W^{(m)}_{(i,i+1)}$. These strings are understood to run between ``link" $i$ (which sits between sites $i-1$ and $i$) and link $i+1$. 

To summarise, we have the following assignments between boundary algebra terms and string operators:
\begin{align}
    Z_i \tilde{Z}_i Z_{i+1} &\leftrightarrow W^{(e)}_{(i,i+1)}\\
    X_i &\leftrightarrow W^{(m)}_{(i,i+1)}.
\end{align}
Notice that the symmetry, $P = \prod_i X_i$ can be re-interpreted as a string operator for an $m$-anyon wrapping around the circumference of the cylinder: $P = \prod_i W^{(m)}_{(i,i+1)}\coloneqq \overline{W}^{(m)}$. Meanwhile, we can also define the operator obtained by wrapping an $e$-anyon string around the cylinder: $\prod_i Z_i\tilde{Z}_iZ_{i+1} \coloneqq \overline{W}^{(e)}$. We refer to this operator as the dual symmetry.
We also define $W^{(e)}_k$ to be an $e$-string operator running from $B_\text{phys}$, where it intersects at site $k$, to $B_\text{ref}$. We assign:
\begin{align}
    Z_k \leftrightarrow W^{(e)}_k.
\end{align}
$W_k^{(e)}$ does not create any anyons away from $B_\text{phys}$ because the $e$ anyon at
its endpoint condenses at $B_\text{ref}$. It also anti-commutes with the symmetry, $\overline{W}^{(m)}$, and so, in terms of the 1D system at $B_\text{phys}$, it corresponds to a local order parameter to detect spontaneous symmetry breaking.
Finally, we define $W^{(m)}_k$ to be an $m$-string operator running from $B_\text{phys}$, where it intersects at link $k$, to $B_\text{ref}$, after passing through a short anyon-permuting domain wall we can introduce in the bulk which effects the interchange $e\leftrightarrow m$. This is required so that we can terminate the string operator on $B_\text{ref}$ without creating any excitations there (see Ref.~\onlinecite{motamarri2024symtftequilibriumtimecrystals} for further details). We then have the assignment of link Pauli operators:
\begin{align}\label{eq:Xtilde}
    \tilde{X}_k \leftrightarrow W^{(m)}_k.
\end{align}

Above, we have shown that we can describe the static phases of the Ising chain (the SG and PM) via the Lagrangian subgroups of the toric code phase. We may also describe the two time crystal phases. Notice that $e^{-i\frac{\pi}{2}\sum_iX_i} \propto \prod_{i}X_i = \overline{W}^{(m)}$, and $e^{-i\frac{\pi}{2}\sum_i Z_i\tilde{Z}_iZ_{i+1}} \propto \prod_i \tilde{Z}_i = \overline{W}^{(e)}$. Hence the $\pi$SG and $0\pi$PM phases, at least at their fixed point limits for now, can be obtained by:
\begin{align}
    U_F^{\pi\text{SG}} &= \overline{W}^{(m)} e^{-iH_e},\label{eqn:piSG_U}\\
    U_F^{0\pi\text{PM}} &= \overline{W}^{(e)} e^{-iH_m},\label{eqn:0piPM_U}
\end{align}
where we define $H_e = \sum_i J_i Z_i \tilde{Z}_i Z_{i+1}$ as the Hamiltonian corresponding to the $\mathcal{M}_e$ Lagrangian subgroup, and similarly $H_m = \sum_i g_i X_i$ for the $\mathcal{M}_m$ Lagrangian subgroup. 
Both systems have closed boundary conditions, which we can alter subsequently. In the absence of any $W^{(m)}_k$ in the drive (or the operators that we time-evolve in correlation functions), the link degrees of freedom are not dynamical and the string of $\tilde{Z}_i$ operators in Equation~\ref{eqn:0piPM_U} can all be set to $+1$. We will see later, however, that they can describe twisted boundary conditions.

In summary, we can obtain the four fixed point drives by defining a gapped Hamiltonian $H_\mathcal{M}$ on $B_\text{phys}$, corresponding to a Lagrangian subgroup $\mathcal{M}$,
and prepending the drive $e^{-iH_\mathcal{M}}$ with a kick from an operator $\overline{W}^{(a)}$ for some anyon $a$. If the anyon $a$ is not in $\mathcal{M}$ we have the two terms in Equations~\ref{eqn:piSG_U} and \ref{eqn:0piPM_U}. Otherwise, if $a\in\mathcal{M}$, the anyon string can be absorbed into $B_\text{phys}$ -- it is a product of terms from $H_\mathcal{M}$. Hence, the drive $\overline{W}^{(e)} e^{-iH_e}$ is equivalent to $e^{-iH_e}$, which is the drive for the SG phase. Similarly, $\overline{W}^{(m)} e^{-iH_m} \equiv e^{-iH_m}$, the drive for the PM phase.

While we see that we do not produce TTSB if $a\in \mathcal{M}$, let us now
show how a non-trivial prepending string operator can generate TTSB. %
Consider the $\pi$SG phase. We can observe the characteristic period doubling by considering a string operator $W^{(e)}_i$.
Since both boundaries absorb $e$ anyons in this phase, the string operator creates no anyons anywhere. We see that $W^{(e)}_i (t) = [(U_F^{\pi \text{SG}})^\dagger]^t W^{(e)}_i [U_F^{\pi \text{SG}}]^t = (-1)^t W^{(e)}_i$, since $W^{(e)}_i$ anti-commutes with $\overline{W}^{(m)}$ and commutes with $H_e$. We also have long-range spatiotemporal
order in all eigenstates of the drive:
\begin{align}
    \bra{n}W^{(e)}_i(t) W^{(e)}_j \ket{n} &= (-1)^{t} \bra{n}W^{(e)}_i W^{(e)}_j\ket{n}\nonumber\\
    &\xrightarrow{|i-j|\rightarrow \infty} (-1)^t \cdot c(i,j,\ket{n})\neq 0,
\end{align}
for some function $c(i,j,\ket{n})=\pm 1$ which accounts for the values of the $Z_i\tilde{Z}_iZ_{i+1}$ operators in eigenstate $\ket{n}$.
From the SymTFT perspective, it is the non-trivial braiding between $e$ and $m$ anyons that creates this time crystal signature. This is because $\overline{W}^{(m)} W^{(e)}_i \overline{W}^{(m)}W_i^{(e)} = \mathcal{B}(e,m)$ results in the period-doubling of the unequal-time correlator above.
The non-trivial temporal dynamics of the phase can be understood from the perspective of the non-local degree of freedom of the toric code phase (the ``logical qubit" whose logical operators correspond to $\overline{W}^{(m)}$ and $W_i^{(e)}$). The drive periodically flips this global degree of freedom. %

Similarly, the prepending $\overline{W}^{(e)}$ operator in $U_F^{0\pi\text{PM}}$ results in TTSB and long-range spatiotemporal order as observed by the \textit{dual} order parameter $\tilde{X}_k$ (charged under the dual symmetry, $\overline{W}^{(e)}$~\cite{JiWen2020categorical}):
\begin{align}
    \bra{n}W^{(m)}_i(t) W^{(m)}_j\ket{n} &\xrightarrow{|i-j|\rightarrow\infty} (-1)^t c'(i,j,\ket{n})\neq 0.
\end{align}
The link degrees of freedom are interpreted in terms of the boundary conditions of the 1D system: setting all of these equal to $+1$ corresponds to introducing periodic boundary conditions, while setting an odd number of them to be $-1$ is interpreted as introducing anti-periodic boundary conditions. Equivalently, $\overline{W}^{(e)} = \pm 1$ for periodic/anti-periodic boundary conditions. Hence, we interpret the dual order parameter as \textit{twisting} boundary conditions (between periodic and anti-periodic).
Due to the long-range spatiotemporal order being witnessed by the dual order parameter, we refer to the $0\pi$PM phase as a dual time crystal, and will discuss its features further below.

\subsection{Away from fixed-point drives}
\label{sec:AFFPD}

The SymTFT description has utility also away from the fixed-point drives described above. Due to MBL, we can find an extensive number of local integrals of motion (LIOMs), in terms of which we can characterise and also express the drive~\cite{Khemani_2016,MBL_spin_chains_2016,Huse_MBL_transition,MBL_Heisenberg_2008,khemani2019briefhistorytimecrystals}. Let us consider the more general setup of the Floquet unitary:
\begin{align}
    U_F &= e^{-iH_1}e^{-iH_0},\\
    H_0 &= \sum_j J_j Z_j\tilde{Z}_j Z_{j+1} + \text{perturbations},\\
    H_1 &= \sum_{j}g_j X_j + \text{perturbations},
\end{align}
where the ``perturbations" can include such terms as $Z_{j-1}Z_jZ_{j+1}Z_{j+2}$, $X_iX_{i+1}$, $X_i$, %
etc., but where, for simplicity, we restrict to those perturbations in $H_0$ ($H_1$) that commute with $Z_j\tilde{Z}_jZ_{j+1}$ ($X_j$).
We will assume that these additional interactions enter with couplings much smaller than the larger of the distribution widths of $J_i$ or $g_i$, and that they commute with the $\mathbb{Z}_2$ symmetry $P$ (and the dual symmetry).

Suppose $U_F$ is close to $U_F^{\pi \text{SG}}$. Then we take $g_j = \frac{\pi}{2}+ \tilde{g}_j$, where the mean of $\tilde{g}_j$ is $0$ and its distribution width $\delta \tilde{g} \ll \delta J$.
We have that:
\begin{align}
    U_F &\propto \left(\prod_j X_j\right) e^{-i \sum_j \tilde{g}_j X_j + \text{perturbations}}e^{-i H_0} \\
    &\approx P e^{-i\mathcal{H}(D_j)},
\end{align}
where
$\mathcal{H}(D_j)$ is a time-independent Hamiltonian diagonal in a basis of quasi-local operators $D_j$, the LIOMs or ``$\ell$-bits":
\begin{align}\label{eqn:MBL_Hamiltonian}
    \mathcal{H}(D_j) = \sum_j \tilde{J}_j D_j + \sum_{j\neq k} \tilde{J}_{jk}D_j D_k + \ldots
\end{align}
where the $\tilde{J}_{jk\ldots}$ are exponentially-decaying in the maximum separation between the indices. %
The $\ell$-bits themselves are related to the fixed-point LIOMs, which in this phase are $Z_j\tilde{Z}_jZ_{j+1}$ for all $j$, via a quasi-local unitary operator $\mathcal{U}$: $D_j = \mathcal{U}^\dagger Z_j\tilde{Z}_jZ_{j+1} \mathcal{U}$. $D_j$
is then localised to link $(j,j+1)$, potentially with exponentially decaying tails of support on nearby sites, as shown in Fig.~\ref{fig:classification_illustration}(b). This relationship between LIOMs at and away from the fixed-point generically holds for MBL systems~\cite{khemani2019briefhistorytimecrystals}.

Similarly,
if $U_F$ is close to $U_F^{0\pi \text{PM}}$, we can take $J_j = \frac{\pi}{2} + \tilde{J}_j$ for $\tilde{J}_j$ weakly disordered around $0$ and we arrive at $U_F \approx \overline{W}^{(e)} e^{-i\mathcal{H}(\tau_j^x)}$. Here, the $\tau_j^x$ are the $\ell$-bits obeying the relation $\tau_j^x = \mathcal{U}^\dagger X_j \mathcal{U}$, for some quasi-local unitary operator $\mathcal{U}$ (since $X_j$ are the fixed-point LIOMs). $\mathcal{H}(\tau_j^x)$ is given by an expression similar to Equation~\ref{eqn:MBL_Hamiltonian}, with the $\tau_j^x$ replacing the $D_j$.

We only apply disorder and perturbations to the qubits on the physical boundary of the $\mathbb{Z}_2$-TO, and so the remainder of the system remains unaffected by these features. Note also that the perturbations we add commute with the operators $\overline{W}^{(a)}$, for $a=e,m$, and so these $\overline{W}^{(a)}$ operators can be moved away from the boundary $B_\text{phys}$ -- they are a topological feature of the drive. Hence, away from the fixed points, the four distinct drives, $U_F^{\text{SG}}$, $U_F^{\pi\text{SG}}$, $U_F^{0\pi\text{PM}}$, $U_F^{\text{PM}}$, have the following generic forms:
\begin{align}
    U_F^{(\pi)^k \text{SG}} &= \left(\overline{W}^{(m)}\right)^k e^{-i\mathcal{H}(D_j)}\\
    U_F^{(0\pi)^k\text{PM}} &= \left(\overline{W}^{(e)}\right)^k e^{-i\mathcal{H}(\tau^x_j)}
\end{align}
for $k\in \lbrace 0, 1\rbrace$ and where the $\ell$-bits are either domain wall operators
$D_j$, or operators $\tau^x_j$, depending on the phase. The Hamiltonian $\mathcal{H}$ is quasi-local in one of these choices of $\ell$-bit. Note that for the cases in which we have $D_j$ as $\ell$-bits, we may also refer to operators $\tau^z_j$ and $\tau^x_j$, defined using the same quasi-local unitary relating $D_j$ to the fixed-point LIOMs: $\tau^z_j\equiv \mathcal{U}^\dagger Z_j \mathcal{U}$ and $\tau^x_j \equiv \mathcal{U}^\dagger X_j \mathcal{U}$.

\subsection{Breaking the (dual) Ising symmetry}\label{sec:breaking_dual_Ising_sym}

The $\pi$SG phase is stable against $\mathbb{Z}_2$-symmetry-breaking perturbations. In such cases, we can relate the eigenstates of the symmetric unitary, $U_F^0$, with those of the perturbed
unitary, $U_F^\text{pert}$, via a quasi-local unitary $\mathcal{U}$, just as we did above. Then there exists an emergent symmetry, $\tilde{P} = \mathcal{U}^\dagger P \mathcal{U} = \prod_j \tau^x_j$, that is preserved by $U_F^\text{pert}$~\cite{KeyserlingkKhemaniSondhi2016stability}. The drive can be written in the form: $U_F^\text{pert} = \tilde{P} e^{-i \mathcal{H}(D_j)}$. In terms of the SymTFT description via the toric code, we may dress the boundary plaquette and vertex operators as $B_p \mapsto \mathcal{U}^\dagger B_p \mathcal{U}$ and $A_v \mapsto \mathcal{U}^\dagger A_v \mathcal{U}$ which maintains their commutation relations and, by the locality of $\mathcal{U}$, their locality. Then, once again, the emergent symmetry $\prod_j \tau^x_j \sim \overline{W}^{(m)}$ for some operator $\overline{W}^{(m)}$ in the bulk of the $\mathbb{Z}_2$-TO and $\sim$ refers to equivalence up to a product of stabilisers of the $\mathbb{Z}_2$-TO. We can find the same $\pi$-spectral pairing and hence time-translation symmetry breaking in the perturbed model, leading to its ``absolute stability".

Meanwhile, the $0\pi$PM phase is stable against \textit{dual} symmetry-breaking (but symmetric under $P$) perturbations, where the dual $\tilde{\mathbb{Z}}_2$-symmetry is given by~\cite{JiWen2020categorical} $\overline{W}^{(e)} = \prod_j \tilde{Z}_j$, the product of all link Pauli-$Z$ operators. %
The argument follows the case above: we can find a quasi-local $\mathcal{U}$ that relates the unperturbed eigenstates with the perturbed ones and hence there exists an emergent dual symmetry, $\prod_j \tilde{\tau}^z_j$ for $\tilde{\tau}^z_j = \mathcal{U}^\dagger \tilde{Z}_j \mathcal{U}$. We show numerical evidence for this in Ref.~\cite{motamarri2024symtftequilibriumtimecrystals}. 

While one might worry that the unitary $\mathcal{U}$ for the dual-symmetry-breaking perturbation case of the $0\pi$PM may not be local, the feature of absolute stability in both this and the $\pi$SG phases is guaranteed by Kramers-Wannier duality, which in the SymTFT is effected by a domain wall that swaps $e\leftrightarrow m$ anyons. Wrapping such a domain wall around the cylinder exchanges the Lagrangian subgroups on $B_\text{phys}$: $\mathcal{M}_e \leftrightarrow \mathcal{M}_m$. It also exchanges $\overline{W}^{(e)} \leftrightarrow \overline{W}^{(m)}$. It therefore maps the $0\pi$PM phase to the $\pi$SG phase. The $\tilde{\mathbb{Z}}_2$-symmetry breaking perturbation is mapped to a $\mathbb{Z}_2$-symmetry breaking perturbation, under which we know the $\pi$SG is absolutely stable~\cite{KeyserlingkKhemaniSondhi2016stability}. The $\pi$-spectral pairing of the $0\pi$PM phase with closed boundary conditions results from the fact that eigenstates are labelled by the values of the $\ell$-bits $\tau^x_j$ along with the value of $\overline{W}^{(e)}$. The latter is $+1$ ($-1$) if the boundary conditions are (anti-)periodic. Hence the $\pi$-pairing is between pairs of eigenstates with periodic/anti-periodic boundary conditions.
While $\tilde{\mathbb{Z}}_2$-symmetry-breaking perturbations couple these boundary conditions, since the spectral pattern is $\pi$-pairing (and not a degeneracy, as it would be for the PM), absolute stability emerges by the same mechanism as for the $\pi$SG.

\section{\texorpdfstring{$\mathbf{G}$-topological order and the 1D operator algebra for SymTFT}{G-topological order and the 1D operator algebra for SymTFT}}\label{sec:G_TO}

SymTFT, in both zero-temperature and non-equilibrium settings, can be used to provide a description of a 1D $G$-symmetric system in terms of a 2D topologically ordered system~\cite{JiWen2020categorical,Kong_holographic_entanglement_2020,Apruzzi21,Kaidi22,Bhardwaj_PtII,ChatterjeeWen23,TH23,Huang_Chen_criticality,WenPotter_SPTgapless,bhardwaj2023club,motamarri2024symtftequilibriumtimecrystals}. The 1D system can be a quantum field theory or a spin system---we focus on the latter in this paper---and it resides on the boundary of the 2D TO. In this section, we describe the quantum double models for $G$-TO and we describe the algebra of boundary operators that will be placed in correspondence with the 1D $G$-symmetric system of interest.

\subsection{Quantum double models}\label{sec:quantum_doubles}

Consider a finite, Abelian
group $G$. We have an associated topologically ordered phase known as the quantum double model $\mathcal{D}(G)$, or $G$-gauge theory~\cite{kitaev2003fault,Cui_2020}. This theory will be defined on a directed lattice, $\mathcal{L}$, 
which is a tesselation of an orientable manifold, $\Sigma$, with local Hilbert spaces assigned to the edges of $\mathcal{L}$: $\mathcal{H}_\text{local} = \text{span}(  \lbrace \ket{g} \, | \, g\in G\rbrace  )$. Each of these local Hilbert spaces is a copy of the group algebra $\mathbb{C}[G]$ (i.e.,  the space of formal linear combinations of group elements with complex coefficients), and as such, following Kitaev~\cite{kitaev2003fault}, we define ``vertex'' and ``plaquette'' projection operators, $A_v$ and $B_p$, which act on the sites around $v$ and $p$, respectively, in such a way that $[A_v, B_p] = 0$, for all $v$ and $p$ (we define these in Appendix~\ref{app:QD_models}). The (fixed-point) Hamiltonian of the system is
\begin{align}\label{eq:H_qd}
    H_0 = -\sum_v A_v - \sum_p B_p.
\end{align}

The ground space is defined as the space of states obeying $A_v\ket{\psi} = B_p \ket{\psi} = \ket{\psi}$, for all $v$ and $p$.
The set of topologically distinct excitations in this Hamiltonian corresponds to an anyon model $\mathcal{A}_G$. Since we are dealing with finite, Abelian $G$, this set is  (see Appendix~\ref{app:QD_models}): $\mathcal{A}_G = \lbrace a = (g,\alpha) \, | \, g\in G, \, \alpha \in \text{Rep} \, G\rbrace$, where $\text{Rep}\, G$ is the set of irreducible representations (irreps), $\alpha \, : \, G \rightarrow U(1)$, of $G$. %
The fusion of two anyons, $a=(g,\alpha), b=(g',\alpha')\in\mathcal{A}_G$, is given by $(gg', \alpha \alpha')$, with the representation defined as $\alpha \alpha ' (g) = \alpha (g) \alpha' (g)$. The anti-anyon of $a$ is $\bar{a} = (g^{-1}, \alpha^{-1})\in\mathcal{A}_G$. The exchange phase of two $a = (g,\alpha)$ anyons is $e^{i\theta_a} = \alpha(g)$. The braiding phase of anyons $a = (g,\alpha)$ and $b = (g',\alpha')$, accrued by encircling $a$ clockwise around $b$, is $\mathcal{B}(a,b) = \alpha (g') \alpha' (g)$.
If the manifold, $\Sigma$, is a torus, the space of ground states of $H_0$ is degenerate, with a dimension given by the number of anyons in $\mathcal{A}_G$~\cite{Cui_2020}. We will be considering the system on a cylinder, where the ground state degeneracy will be determined by the top and bottom boundary conditions (see below).

Anyons
appear at the end points of \textit{string operators}, or Wilson and/or `t Hooft lines. (In models away from fixed points, these operators become ``dressed" and their support is predominantly on finite-width strips~\cite{Hastings_Wen_2005,Bravyi_2010_TQO_Stab,bauer_area_2013,Huse_localization_2013,Bridgeman_ribbon_ops_2016,Wahl_LIOMs_2020,Wahl_2024,MBL_Venn_2024}.)
We define one type of string operator on the edges of the lattice $\mathcal{L}$: given a path $\gamma_{ij}$ of edges, terminating at vertices $i$ and $j$, we have string operator $W^{(1,\alpha)}_{\gamma_{ij}}$,
directed from $i$ to $j$,
that creates an anyon $\bar{a} = (1,\alpha^{-1})$ on $i$ and an anyon $a = (1,\alpha)$ on $j$. This operator is defined to act with representation $\alpha$ on all edges along $\gamma$ that have a direction aligned with $\gamma$ (from $i$ to $j$) and with $\bar{\alpha}$ along edges that are anti-aligned. We have that $W^{(1,\alpha)}_{\gamma_{ij}}$ commutes with all vertex operators $A_v$ except for $A_i$ and $A_j$. %
We next consider paths on the dual lattice $\tilde{\mathcal{L}}$.
Consider a path $\tilde{\gamma}$ of edges in $\tilde{\mathcal{L}}$, terminating at $\mathcal{L}$-plaquettes $k$ and $l$. We define another type of string operator
$W_{\tilde{\gamma}_{kl}}^{(g,1)}$, directed from
$k$ to $l$, which acts with group element $g$ on all edges directed to the right (when facing along the direction of $\tilde{\gamma}$) and $g^{-1}$ on all edges directed to the left. This operator creates anyon $(g^{-1}, 1)$ on $k$ and anyon $(g,1)$ on $l$. $W_{kl}^{(g,1)}$ commutes with all plaquette operators $B_p$ except for $B_k$ and $B_l$.

As we will be dealing with finite Abelian $G$, we will make use of the notation for the (generalised) toric codes~\cite{kitaev2003fault}.
From the fundamental theorem of finite Abelian groups, we may decompose $G \simeq \mathbb{Z}_{k_1}\times \mathbb{Z}_{k_2}\times \ldots \times \mathbb{Z}_{k_p}$, for some number of factors, $p$. 
We will label the generator of each cyclic group $m_i$, such that $m_i^{k_i} = 1$. 
The irreducible representations of $G$ are products of the $k_i^{\text{th}}$-roots of unity, one of which we will label $e_i$ (chosen so that we may generate the other roots of unity via products of $e_i$). 
Hence, a general anyon in the model can be written as $(m_1^{i_1}m_2^{i_2}\ldots m_p^{i_p}, e_1^{j_1}e_2^{j_2}\ldots e_p^{j_p}) \equiv m_1^{i_1}m_2^{i_2}\ldots m_p^{i_p} e_1^{j_1}e_2^{j_2}\ldots e_p^{j_p}$, where we remove the brackets notation in this context as the meaning of the anyon is unambiguous.

\subsection{\texorpdfstring{Adding boundaries to the $\mathbf{G}$-TO}{Adding boundaries to the G-TO}}

As above, we will consider SymTFT on a cylinder, identifying one of the 1D boundaries as the physical boundary, $B_\text{phys}$,
which is the symmetric spin system of interest, and the rest of the system as belonging to the $G$-TO bulk. The bulk will be described by a Hamiltonian with an energy gap much larger than any of the energy scales of the dynamics at $B_\text{phys}$ (hence when describing the eigenvalues of the drive, we will always implicitly project into the low-energy subspace of the bulk Hamiltonian). We have one boundary, the reference boundary $B_\text{ref}$, which we will also require to possess a large energy gap, and hence can be thought of as being part of the bulk.

The gapped boundaries of an Abelian $G$-TO are in one-to-one correspondence with the \textit{Lagrangian subgroups} of its anyon model~\cite{Kapustin_2011,Levin_2013,Barkeshli13a,Cong_2017}.
A Lagrangian subgroup is a subset $\mathcal{M}\subset \mathcal{A}_G$ satisfying the following:
\begin{enumerate}
    \item $\mathcal{M}$ is closed under fusion,
    \item All $a\in \mathcal{M}$ are bosons: $e^{i\theta_a} = 1$,
    \item All $a, a' \in \mathcal{M}$ braid trivially with one another: $\mathcal{B}(a,a') = 1$, %
    \item The group is maximal: if $b\in\mathcal{A}_G$ braids trivially with all $a\in\mathcal{M}$, then $b\in\mathcal{M}$.
\end{enumerate}
One can show that any Lagrangian subgroup $\mathcal{M}\subset \mathcal{A}_G$ obeys $|\mathcal{M}| = |G|$. We can immediately find two examples of Lagrangian subgroups: $\mathcal{M}_G = \lbrace (g,1) \, | \, g\in G\rbrace$ and $\mathcal{M}_{\text{Rep}\, G} = \lbrace (1,\alpha) \, |\, \alpha \in \text{Rep}\, G\rbrace$.

The Hamiltonian of the $G$-TO with a boundary corresponding to Lagrangian subgroup $\mathcal{M}$ contains a term
$H_\text{bdry}  = \sum_{\langle i j\rangle , \, a\in \mathcal{M}} J_{ij} W_{ij}^{(a)}$, where we understand the string
operator to be running between adjacent boundary sites $i,j$
either in $\mathcal{L}$ or $\tilde{\mathcal{L}}$, or a combination thereof (if $\mathcal{M}$ contains $G$ and $\text{Rep}\, G$ components). Such a boundary is said to \textit{condense} anyons from $\mathcal{M}$, since a string $W^{(a)}$ corresponding to anyon $a\in \mathcal{M}$ terminating on the boundary creates no excitations at that boundary.

The couplings
$J_{ij}$ in $H_\text{bdry}$ will in general be taken to be disordered, and we will be considering the \textit{eigenstate} order of the static and dynamic phases exhibited on the boundary. While anyon condensation is usually defined in relation to the ground state of the system, we can define the boundary to condense anyons created by $W^{(a)}$ operators by noting that eigenstates display long-range correlations in these operators: $\lim_{|i-j|\rightarrow \infty}\lim_{L\rightarrow \infty}\bra{n}W^{(a)}_{ij}\ket{n} \neq 0$. We will discuss this further below. When dealing with Floquet phases in the following section, we will use the properties of many-body localised systems to explicitly replace the operators $W^{(a)}_{ij}$ with quasi-local integrals of motion that are related to the $W^{(a)}_{ij}$ by a quasi-local unitary operator $\mathcal{U}$. We discuss this further in Section~\ref{sec:MBL_bdrys}.

\subsection{The Boundary Algebra}\label{sec:Boundary_algebra}

To characterise $B_\text{phys}$, we identify the \textit{boundary algebra} (BA) as the algebra of string operators corresponding to the physical degrees of freedom~\cite{ChatterjeeWen23,motamarri2024symtftequilibriumtimecrystals}
--- these are identified with the operators acting on the physical (i.e., the 1D) system (even though they may be deformed to pass through the bulk). We take local operators $W^{(g,\alpha)}_{ij}$ (i.e., with $i,j$ near each other on $B_\text{phys}$)
as belonging to the BA. The Hamiltonian and (later) the Floquet unitary for $B_\text{phys}$ will be constructed from these BA operators.

The BA relations are obtained from the anyon braiding statistics. If we define intervals on $B_\text{phys}$ $[i,j]$ and $[k,l]$, such that one endpoint of $[i,j]$ lies inside the interval $[k,l]$ and vice versa, then we have:
\begin{align}
    W^{(a)}_{ij} W^{(b)}_{kl} = \mathcal{B}(a,b)W^{(b)}_{kl}W^{(a)}_{ij}.
\end{align}
These string operators 
commute with all the bulk vertex and plaquette operators of the model. Examples of these operators for the case of the $\mathbb{Z}_2$-toric code model are shown in Fig.~\ref{fig:classification_illustration}(a).

We will now consider the role of the group $G$ in the boundary algebra. We illustrate our discussion below in Fig.~\ref{fig:Static_phases}(a). We will take $G$ to be a symmetry group of the 1D system, with onsite, unitary action $V_k(g)$, belonging to the BA. We will fix this action to be given by the string operator $W_{k\, k+1}^{(g,1)}$ (where we here interpret $k$ and $k+1$ to refer to adjacent vertices in the dual lattice).
This creates anyons $(g^{-1},1)$ and $(g,1)$ either side of site $k$ (see operator $X_i$ in Fig.~\ref{fig:classification_illustration}(a)). The global symmetry operator is then given by $\overline{W}^{(g,1)} \equiv \prod_k W^{(g,1)}_{k\, k+1}$.

\begin{figure*}
\begin{tikzpicture}
    \node[inner sep=0pt] at (-4,0){\includegraphics[width=0.75\linewidth]{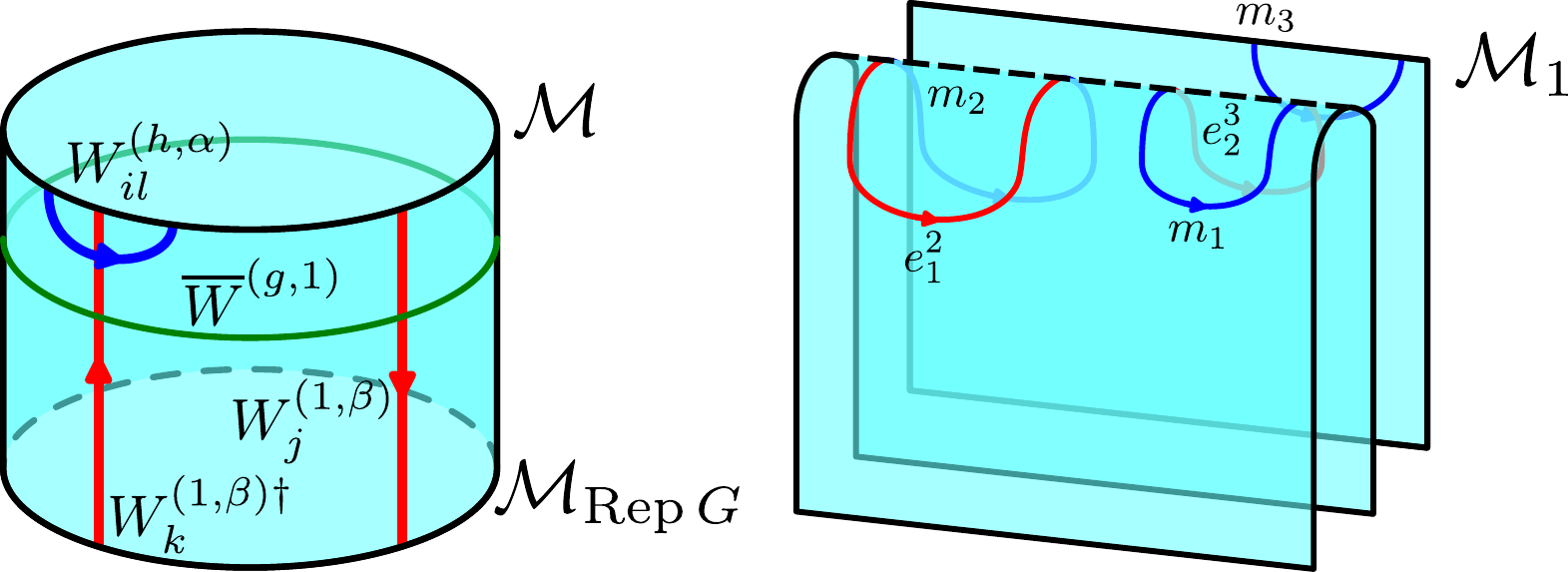}};
    \node[inner sep=0pt] at (-8.6,-3) {(a)};
    \node[inner sep=0pt] at (-1.55,-3)    {(b)};
\end{tikzpicture}
    \caption{Schematic 
    illustration of $G$-symmetric 1D static phase constructions in SymTFT, and of the boundary algebra (BA). (a) We illustrate general members of the BA. Local operator $W_{il}^{(h,\alpha)}$, shown in blue, corresponds to the process of anyons $(h,\alpha)$ and $(h^{-1},\alpha^{-1})$ being created close to the boundary and condensed at that boundary.
    Local order parameters $W_j^{(1,\beta)}$, correspond to strings running between top and bottom boundaries, and global symmetry/dual symmetry transformations $\overline{W}^{(g,\gamma)}$ (only a symmetry transformation shown in green here) correspond to strings wrapping the cylinder. The example shown is for a partially-SSB phase, with symmetry broken to subgroup $H<G$. The Lagrangian subgroup $\mathcal{M}$ of the physical boundary (top boundary) 
    includes anyons $(h,\alpha)$ for $h\in H$ and some $\alpha\in\text{Rep}\, G$, such that $\alpha(H) = 1$. 
    Local order parameters $W_k^{(1,\beta)\dagger}$ and $W_j^{(1,\beta)}$ (with $\beta$ satisfying $\beta(H) = 1$) signal long-range order on $B_\text{phys}$. These do not commute with symmetry operator $\overline{W}^{(g,1)}$ for $g\notin H$ shown in green. (b) An example of SPT order for $G = \mathbb{Z}_4\times \mathbb{Z}_6\times \mathbb{Z}_3$, with Lagrangian subgroup $\mathcal{M}_1 = \langle m_1 e_2^3, e_1^2m_2, m_3\rangle$. Only a short section of the $G$-TO and $B_\text{phys}$ is shown. The three layers represent $\mathbb{Z}_4$, $\mathbb{Z}_6$, and $\mathbb{Z}_3$ toric code layers, respectively. String operators corresponding to the three anyons generating $\mathcal{M}_1$ are shown. The $\mathbb{Z}_4$ and $\mathbb{Z}_6$ layers are folded together along the boundary (dashed line represents a domain wall) to indicate the multi-layer anyons that appear in $\mathcal{M}_1$.}\label{fig:Static_phases}
\end{figure*}

We also include the
local order parameters in the BA, 
which will signal phase transitions on $B_\text{phys}$. These order parameters must transform in some representation of $G$ under the action of the global symmetry: $(\prod_kV_k(g))^\dagger \mathcal{O}_i (\prod_kV_k(g)) = \alpha(g)\mathcal{O}_i$, for some irrep $\alpha$ (considering, as always, $G$ to be Abelian).
As such, a natural candidate for such an order parameter will be $W^{(1,\alpha)}_{k\rightarrow l}$, with $k$ a site on $B_\text{phys}$ and $l$ some site far separated from $B_\text{phys}$. 
To prevent this string operator from creating excitations in the bulk, we let $l$ be a site on $B_\text{ref}$, and we will choose the Lagrangian subgroup of $B_\text{ref}$ to be $\mathcal{M}_{\text{Rep}\, G}$ so that $W^{(1,\alpha)}_{k\rightarrow l}$ creates no anyons at $l$. Since the endpoint of this string operator can be deformed away from $l$ to elsewhere on $B_\text{ref}$ using suitable boundary vertex or plaquette operators, we drop this subscript and let the local order parameter at site $k$ be $W^{(1,\alpha)}_k$. [Locality here is understood in the sense of the 1D system---$W^{(1,\alpha)}_k$ corresponds to an operator at position $k$ of the 1D system---while of course $W^{(1,\alpha)}_k$ is non-local in terms of the 2D TO.] To see that this transforms appropriately, we use the braiding statistics between the associated anyons once again:
\begin{align}
    \overline{W}^{(g,1)\dagger} W^{(1,\alpha)}_k \overline{W}^{(g,1)} &= \mathcal{B}((1,\alpha),(g,1)) W^{(1,\alpha)}_k \nonumber\\
    &= \alpha(g)W^{(1,\alpha)}_k,
\end{align}
which is the appropriate transformation we defined above. %
We will use these local order parameters in what follows. Examples are schematically illustrated in Fig.~\ref{fig:Static_phases}(a).

Similarly to the above, we also define local \textit{dual} order parameters $W_k^{(g,1)}$ to belong to the BA. These run between link $k$ on $B_\text{phys}$ and some position on $B_\text{ref}$. For these to not introduce an excitation on $B_\text{ref}$, we introduce a short domain wall that interchanges $G$ and Rep$\, G$ anyons. We let the string operator $W_{k\rightarrow B_\text{ref}}^{(g,1)}$ run through this domain wall before terminating at $B_\text{ref}$, where it now produces no excitation. The $W_k^{(g,1)}$ are charged under the dual symmetries $\overline{W}^{(1,\alpha)}$. As we shall explain, $\overline{W}^{(1,\alpha)}$ detects $G$-twisted boundary conditions; the $W_k^{(g,1)}$ are thus boundary condition twisting operators. 
To summarise, the boundary algebra is generated by the following string operators from the bulk $G$-TO:
\begin{align}\label{tab:BAsummary}
\vcenter{\hbox{\begin{tabular}{@{}c@{}}
Local Symmetric \\
Operators
\end{tabular}}} &\; \leftrightarrow \; W^{(g,\alpha)}_{ij} \equiv \text{Anyon strings}\nonumber\\
\vcenter{\hbox{\begin{tabular}{@{}c@{}}
Local Symmetry \\
Action
\end{tabular}}} &\; \leftrightarrow \; W^{(g,1)}_{ij} \equiv \text{$m$-strings}\nonumber\\
\vcenter{\hbox{\begin{tabular}{@{}c@{}}
Local Dual \\
Symmetry Action
\end{tabular}}} &\; \leftrightarrow \; W^{(1,\alpha)}_{ij} \equiv \text{$e$-strings}\\
\vcenter{\hbox{\begin{tabular}{@{}c@{}}
Local Order \\
Parameters
\end{tabular}}} & \; \leftrightarrow \; W_k^{(1,\alpha)}\equiv \text{$e$-anyons on $B_\text{phys}$}\nonumber\\
\vcenter{\hbox{\begin{tabular}{@{}c@{}}
Local Dual \\
Order Parameters
\end{tabular}}} & \;\leftrightarrow\; W_k^{(g,1)}\equiv \text{$m$-anyons on $B_\text{phys}$}.\nonumber
\end{align}
The ``anyon strings" refer to those that are confined to some local region on $B_\text{phys}$, while ``anyons on $B_\text{phys}$" refer to single unpaired anyons located on $B_\text{phys}$, with the other string end-point located on $B_\text{ref}$. $e$ and $m$ anyons refer to general products of $e$-type/$m$-type anyons in the various $\mathbb{Z}_{k_i}$ toric code layers (see Section~\ref{sec:quantum_doubles}).

\section{SymTFT for static phases}\label{sec:Static_phase_SymTFT}

We will now use the above construction to present a SymTFT classification of static phases in spin chains with finite Abelian symmetry with on-site, unitary action. Our coverage recapitulates certain classification schemes and features present in the literature~\cite{lichtman2020bulk,Kong_holographic_entanglement_2020,ChatterjeeWen23,TH23}, while introducing novel aspects. We cover spontaneously symmetry-broken (SSB), partially-SSB, and symmetry-protected topological (SPT) phases. We will describe how one can include twisted boundary conditions in these 1D models and we will also discuss a simple interpretation of an SPT classification scheme, which, to our knowledge, is a new approach to SPT classifications.
We provide illustrative examples and summaries of the scheme discussed.

\subsection{Static SSB and partially-SSB phases}

We now introduce the static classification scheme for SSB and partially-SSB phases in terms of the Lagrangian subgroups of the $G$-TO. 
The physical system is in a (partially-)SSB phase, in which the symmetry group is broken to a subgroup $H < G$, if local order parameters transforming trivially under $H$ but non-trivially under $G$, and only those operators, display long-range correlations~\cite{Chaikin_Lubensky_1995,Beekman_2019,Levin_2020}. 
Note that, in general, we will consider correlation functions in any eigenstate of the drive $U_F$, to keep the discussion general enough to cover dynamic phases later.
In the more familiar case of $H=\lbrace 1 \rbrace$, all local order parameters transforming non-trivially under $G$ signal long-range order via these correlation functions. In the simple case of the transverse-field Ising model in the SSB phase we have that local operators $Z_j$ detect long-range order, since they transform non-trivially under the global symmetry $\prod_i X_i$. We have that
$C(j,k) \equiv \langle Z_j Z_k\rangle - \langle Z_j\rangle \langle Z_k\rangle \neq 0$ in the limit of $|j-k|\rightarrow \infty$.

Now we will show how this physics is captured by SymTFT (see Fig.~\ref{fig:Static_phases}(a) for an illustration). Suppose that the physical boundary of the $G$-TO is gapped and described by a Lagrangian subgroup $\mathcal{M} = \lbrace (h,\beta) \rbrace$. The $G$-components of $\mathcal{M}$ form a subgroup $H = \lbrace h\rbrace$ (since $\mathcal{M}$ is closed under fusion). We now show that this $H$ is precisely the 1D system's unbroken symmetry group.

The string operators $W_k^{(1,\beta)}$ transforming trivially under $H$ are precisely for those $(1,\beta)$ that are in $\mathcal{M}$, since $\mathcal{M}$ contains all bosons that braid trivially with all other anyons in $\mathcal{M}$. And since these anyons $(1,\beta)$ are also in $\mathcal{M}_\text{ref}$, we can terminate $W_k^{(1,\beta)}$ on $B_\text{ref}$ without creating excitations. Hence, considering a non-trivial $\beta\in \text{Rep}\, G$, these local order parameters display long-range order in the 1D model's drive eigenstates:
\begin{align}
    &\lim_{|k-j|\rightarrow \infty}\lim_{L\rightarrow \infty} \mathbb{E}(|\bra{n}W_k^{(1,\beta)\dagger}W_j^{(1,\beta)}\ket{n}| - \nonumber \\
    &\qquad \qquad\qquad |\bra{n}W_k^{(1,\beta)\dagger}\ket{n} \bra{n}W_j^{(1,\beta)}\ket{n}|)\nonumber\\
    &= \lim_{|k-j|\rightarrow \infty}\lim_{L\rightarrow \infty}\mathbb{E}(|\bra{n}W_k^{(1,\beta)\dagger}W_j^{(1,\beta)}\ket{n}|)\neq 0,\label{eq:Z2sig}
\end{align}
for $\ket{n}$ any eigenstate of the drive that is also an eigenstate of all symmetry operators.
In the above, we denote by $\mathbb{E}$ the expectation value over disorder realisations (pulling drive parameters from their respective distributions), and we take absolute values in order to account for the glassiness in the eigenstates of the drive (resulting from the arbitrary eigenvalues
of the LIOMs). We also use the fact that $W_k^{(1,\beta)\dagger}W_j^{(1,\beta)} = W_{jk}^{(1,\beta)} \equiv \prod_{i=j}^k W_{i\, i+1}^{(1,\beta)}$
fixes the eigenstate (up to a phase), while $W_k^{(1,\beta)}$ acts non-trivially on $\ket{n}$, since it does not commute with the symmetry group action under the assumption that $\beta$ is non-trivial.

For an order parameter transforming non-trivially under $H$, however, no such long-range correlator is formed. To see this, consider $(1,\alpha)\notin \mathcal{M}$. Then: 
\begin{align}
\bra{n}W_k^{(1,\alpha)\dagger}W_j^{(1,\alpha)}\ket{n} &= t_{j}^{-1}\bra{n}W_k^{(1,\alpha)\dagger}W_j^{(1,\alpha)} W^{(h,\beta)}_{j\, j+1}\ket{n}\nonumber\\
&= t_j^{-1} \alpha(h)\bra{n}W_k^{(1,\alpha)\dagger} W^{(h,\beta)}_{j\, j+1} W_j^{(1,\alpha)}\ket{n}\nonumber\\
&= \alpha(h)\bra{n}W_k^{(1,\alpha)\dagger} W_j^{(1,\alpha)}\ket{n} = 0.
\end{align}
In the above, $W^{(h,\beta)}_{j\, j+1}$ is an element of the BA with some $(h,\beta)\in \mathcal{M}$---hence, $\ket{n}$ is an eigenstate of the operator with eigenvalue $t_j$, related to the eigenvalues of LIOMs around $j$.
We suppose this string operator is localised to site $j$ but far-separated from $k$. Hence, commuting it past $W_j^{(1,\alpha)}$ results in the phase $\alpha(h)\neq 1$, while commuting it past $W_k^{(1,\alpha)\dagger}$ results in no phase.

We see that the above prescription captures the long-range order of the symmetry-broken phases on $B_\text{phys}$. A choice of Lagrangian subgroup on $B_\text{phys}$ uniquely specifies the (partially-)SSB phase on the boundary (up to a further SPT-grading we describe in the next section). This is because the choice of $\mathcal{M}$ fully specifies those local order parameters that display long-range order (they correspond %
to the anyons $(1,\beta)$ whose Rep$\, G$ component features in $\mathcal{M}$), and those that are only short-range correlated. The transformation of these operators under symmetry group $G$ is specified by the braiding rules of the associated anyons, $(g,1)$ and $(1,\beta)$. There are as many (partially-)SSB phases (up to an SPT-grading) as there are subgroups of $G$, since we can find an associated Lagrangian subgroup for each $H$ (see Appendix~\ref{app:SSB_phases_subgroup_H}).

\subsection{Static SPT phase order parameters}

With
respect to any unbroken symmetry $H$, the system may
have SPT order~\cite{Chen2010,Chen_classification_2011,schuch_classification_MPS_2011,Pollmann_Turner_2012_SPT_1D,chen2012SPTphases,Chen_2013_SPT_Group_cohomology,Verresen_2017_SPT}. In such a phase, we expect to be able to find ``string order parameters" displaying long-range correlation~\cite{String_order_param_spin_chains_1989,Perez_Garcia_2008,Haegeman_SPT_String_order_2012,Pollmann_Turner_2012_SPT_1D,Else_2013_hidden_sym_breaking,Duivenvoorden_2013}.
Such an operator is obtained by combining the onsite symmetry with the local order parameters: 
\begin{align}
    W^{(h,\alpha)}_{ij} = W^{(h,1)}_{ij} W^{(1,\alpha)}_{ij} = W^{(h,1)}_{ij} W^{(1,\alpha)}_i W^{(1,\alpha)\dagger}_j
\end{align}
where in the second equality we deform the string $W^{(1,\alpha)}_{ij}$ to two strings running between the physical and reference boundaries: each terminates on $B_\text{ref}$ without creating excitations. In order for this string order parameter to exhibit long-range correlations, we require that $(h,\alpha) \in \mathcal{M}$, the Lagrangian subgroup of $B_\text{phys}$, and also that $(1,\alpha)$ is charged under the symmetry: $\alpha (h') \neq 1$ for some $h'\in H$. This means $(1,\alpha)\notin \mathcal{M}$, since otherwise we could find long-range correlations detected by $W^{(1,\alpha)}_{ij}$, which would indicate that $H$ is spontaneously broken. Hence, $(1,\alpha)$ must braid non-trivially with some $(h',\alpha')\in \mathcal{M}$. Thus $\lim_{|i-j|\rightarrow \infty}\lim_{L\rightarrow\infty}\bra{n}W^{(h,\alpha)}_{ij} \ket{n} \neq 0$ while $\bra{n} W^{(1,\alpha)}_{ij} \ket{n} = 0$ for eigenstates $\ket{n}$.
Hence, we have long-range correlation signalled by string order parameters.
This signals that the phase has non-trivial SPT order---it cannot be connected to the trivial phase by a short-depth, symmetric circuit ~\cite{String_order_param_spin_chains_1989,Perez_Garcia_2008,Haegeman_SPT_String_order_2012,Pollmann_Turner_2012_SPT_1D,Else_2013_hidden_sym_breaking,Duivenvoorden_2013}. 

For an example, we can point to the Haldane phase as exemplified by a cluster state model, which has $\mathbb{Z}_2\times \mathbb{Z}_2$ SPT order~\cite{SPT_MBQC_2012,Else_2013_hidden_sym_breaking}:
\begin{align}
    H_\text{cluster} = - \sum_j Z_{j-1} X_j Z_{j+1}.
\end{align}
This can be understood, from SymTFT, as the boundary of two copies of the toric code~\cite{lichtman2020bulk}. The symmetries of the cluster Hamiltonian are generated by $P_1 = \prod_{j} X_{2j}$ and $P_2 = \prod_{j} X_{2j-1}$ (the product of $X$ operators on even or odd sites). Note that the phase involves no symmetry breaking. We find long-range correlations via the string order parameters: $\langle Z_{2j} X_{2j+1}X_{2j+3}\ldots X_{2k-1} Z_{2k}\rangle \neq 0$ and $\langle Z_{2j-1} X_{2j}X_{2j+2}\ldots X_{2k} Z_{2k+1}\rangle \neq 0$. 

Let us see what gapped boundary of the $\mathbb{Z}_2\times \mathbb{Z}_2$ bulk this phase is associated with. We will consider the first of the above string order parameters to begin with. The operators $Z_{2j}$ are charged under the symmetry $P_1$ and hence are associated with representation $\alpha \equiv (-1,1)$ of $\mathbb{Z}_2\times \mathbb{Z}_2$ (see Appendix~\ref{app:QD_models})---they become $e$ string operators
in one of the toric code copies, running between physical and reference boundaries. Suppose we label the $e$ anyon in this copy $e_1$. Meanwhile, the string of $X_{2j+1}$ operators is given by the symmetry operator $P_2$ restricted to a finite interval, which is a string operator associated with an $m$ anyon in one of the copies. Since $e_1$ braids trivially with this anyon (they are defined on non-overlapping qubits), $P_2$ must be associated with $m_2$ string operators. $e_1$ and $m_2$ are not themselves part of the Lagrangian subgroup, since their associated string operators do not commute with the Hamiltonian, but the anyon $e_1m_2$ is. Similarly, from observing the second of the two string order parameters, $m_1e_2$ is also part of the Lagrangian subgroup. Thus we find a gapped boundary of the two toric code copies with associated Lagrangian subgroup $\lbrace 1, e_1 m_2, m_1 e_2, f_1 f_2\rbrace$ (where $f = e\times m$ is a topological fermion)~\cite{lichtman2020bulk}.

\subsection{Classification of phases with an unbroken symmetry subgroup}\label{sec:Simple_SPT_Class}

Above, we have provided a means of describing
SSB and SPT phases through SymTFT. However, we have not yet enumerated the possible SPT phases appearing in 1D. More broadly, we can describe several distinct phases with the same unbroken symmetry subgroup $H\leq G$. SymTFT also provides a simple way of doing this, via Lagrangian subgroups and string order parameters. Our approach is similar to that of Ref.~\onlinecite{Verresen_PRX_2021}, however SymTFT provides a very natural interpretation and counting method in terms of Lagrangian subgroups. To our knowledge, this feature has not been pointed out elsewhere in the literature.

For unbroken symmetry group $H$, and a (possibly trivial) subgroup $M\leq H$, suppose that the Lagrangian subgroup $\mathcal{M}$ contains $(m,1)$ for all $m\in M$. We also assume that $M$ is the largest such subgroup.
In order for $\mathcal{M}$ to be maximal, we must have rk$(H/M)$ (where rk$(G)$ is the number of independent generators of $G$)
independent anyons of the form $(h,\alpha)$ with $h\notin M$ and $\alpha (m)=1$ for all $m\in M$ (so that $(h,\alpha)$ braids trivially with $(m,1)$), but with $\alpha$ not a trivial representation of $H$ (otherwise $(1,\alpha)$ would be a boson that braids trivially with all members of $\mathcal{M}$; it and its inverse would therefore be a member of $\mathcal{M}$ and hence $(h,1) = (h,\alpha)(1,\alpha^{-1}) \in \mathcal{M}$ and so $h\in M$).
Hence, we can label such anyons of the form $(h,\alpha)$ by a member of $H/M$ and a non-trivial irrep of $H/M$. 
This results in rk$(H/M)$ independent, non-trivial string order parameters which evidence an SPT phase, as we will see below. 

For each such $(h,\alpha)\in \mathcal{M}$, we have $\alpha (h') \neq 1$ for some $h'\in H$, since $\alpha$ is not a trivial representation of $H$. Hence, we can form an independent string order parameter of the form $W_{ij}^{(h,\alpha)}$ whose endpoints are charged under the symmetry group $H$. This string will display long-range correlation, owing to the fact that $(h,\alpha)\in \mathcal{M}$. %
Since $(h,1),(1,\alpha)\notin \mathcal{M}$,
as we have stated, the string order parameter $W_{ij}^{(h,\alpha)}$ does not reduce to a regular (i.e., local) order parameter, so we have genuine SPT order.

Note that, because the anyons forming the string order parameters are labelled by members and irreps of $H/M$, the SPT order belongs to a class in $H^2(H/M , U(1))$---we identify $H/M$ as the ``protecting" symmetry group, while the SPT phase is trivial with respect to the $M$ symmetry. This permits us to break the $M$ symmetry while retaining SPT order, by replacing $\{ (m,1)\}$ in $\mathcal{M}$ with $\{(1,\beta)\}$, where all anyons $(1,\beta)$ transform
non-trivially under $M$ (and only $M$): $\beta(m)\neq 1$ for some $m\in M$, $\beta(h) = 1$ for all $h\in H$, $h\notin M$. In this case, we maintain the long-range correlations in all existing string order parameters.
In Appendix~\ref{app:SPT_coset_cohomology}, we prove that in such cases, the SPT order is labelled by a cocycle in $H^2(H/M,U(1))$, and provide more details.

\subsubsection{Counting SPT Phases}
Having described how the structure of the Lagrangian subgroup can result in non-trivial string order parameters, let us use this to count the number of non-trivial SPT phases with symmetry group $H$. 
This reproduces
the classification of SPT phases
obtained through cohomology considerations~\cite{Chen2010,Chen_classification_2011,schuch_classification_MPS_2011,Pollmann_Turner_2012_SPT_1D,chen2012SPTphases,Chen_2013_SPT_Group_cohomology}, but does so in a conceptually simpler way. We will do so by starting with the trivial phase (for the unbroken symmetry), in which $H=M$, with Lagrangian subgroup $\mathcal{M}^H_\text{triv} = \lbrace (h,1) \, | \, h\in H\rbrace$, and performing a series of ``moves" on the anyons in $\mathcal{M}^H_\text{triv}$, each of which preserves the braiding statistics of the anyons. Note that the phase corresponding to $\mathcal{M}^H_\text{triv}$ is not SPT-ordered, since we cannot form any string order parameters whose end-points transform non-trivially under $G$. 

In what follows, we will use the generalised toric code notation for anyons, while ignoring for now any broken symmetry. We will assume that the unbroken symmetry group $H$ is of the form $\mathbb{Z}_{k_1}\times \ldots \times \mathbb{Z}_{k_n}$ for some integers $k_i$. %
We define $m_{k_i}$ and $e_{k_i}$ as generators and irreps of group $H$, respectively, in the way explained in Sec.~\ref{sec:quantum_doubles}. In this situation, we have $\mathcal{M}^H_\text{triv} = \langle m_{k_i}\, |\, i=1,\ldots , n\rangle$.

To obtain a phase with a string order parameter, we start attaching Rep$\, H$ anyons (i.e., those of the form $e_{i}^p$, $p\in\mathbb{Z}_{k_i}$) to the anyons of the Lagrangian subgroup, since these will allow for string order parameter endpoints to be charged under the symmetry. We must do this in a way that preserves the braiding statistics of $\mathcal{M}$. We cannot, for example, include $e_{i}m_{i}$ in $\mathcal{M}$, since this is not a boson. Thus we need to attach an anyon $e_{i}^p$ to $m_{j}$, for layer indices $i\neq j$: we thereby obtain products of anyons \textit{between} toric code layers in $\mathcal{M}$. To perform this attachment while preserving the braiding statistics with $m_{i}$, we must attach $e_{j}^q$ (for some $q\in \mathbb{Z}_{k_j}$) to $m_{i}$. $p$ and $q$ have to be chosen so that $e_{i}^p m_{j}$ and $m_{i}e_{j}^q$ braid trivially. This attachment is a ``move" that does not break any boundary symmetry (this can be seen by noting that we cannot find any long-range correlation in local order parameters $W^{(e_i)}_k$, since $e_i\notin \mathcal{M}$) %
and hence takes us to a new SPT phase. The new phase displays long-range order in $W_{ij}^{(e_i^p m_j)} = W_{ij}^{(m_{j})} (W_i^{(e_{i})})^p (W_j^{(e_{i})})^{p\dagger}$ (and similarly for $W_{ij}^{(m_ie_j^q)}$), whose endpoints transform non-trivially under the symmetry.

How many phases can we thus obtain? Consider two toric-code layers (indexed $i$ and $j$) from the full system, with $\mathbb{Z}_k$ and $\mathbb{Z}_m$ TO respectively. Let $\omega_k$ ($\omega_m$) be a simple $k^{\text{th}}$ ($m^\text{th}$) root of unity. Now consider anyons $a = m_i e_j^p$ and $b = e_i^q m_j$, for $p\in \mathbb{Z}_m$ and $q\in\mathbb{Z}_k$. We have:
\begin{align}
    \mathcal{B}(a,b) = \omega_k^q \omega_m^p.
\end{align}
This is equal to $1$ if and only if:
\begin{align}
    \frac{q}{k} + \frac{p}{m} \in \mathbb{Z}.
\end{align}
This can be shown to have gcd$(k,m)$ number of solutions for $p\in \mathbb{Z}_m$ and $q\in\mathbb{Z}_k$. We prove this in Lemma~\ref{lem:sum_of_fracs} of Appendix~\ref{app:SPT_Lemma}, with the solutions themselves presented there as well.

Hence, for layers $i$ and $j$, we have gcd$(k,m)$ possible anyon attachments that can lead to distinct SPT phases. Performing such moves between all pairs of layers in the TO system, we find the following classification of SPT phases:
\begin{align}\label{eqn:SPT_class}
    \prod_{i<j} \mathbb{Z}_{\text{gcd}(k_i,k_j)}
\end{align}
where each $\mathbb{Z}_{\text{gcd}(k_i,k_j)}$ factor determines which attachment has occurred between layers $i$ and $j$. This reproduces the classification obtained through group cohomology theory~\cite{chen2012SPTphases,potter2016prx,Chen2010,Chen_classification_2011,schuch_classification_MPS_2011,Pollmann_Turner_2012_SPT_1D,chen2012SPTphases,Chen_2013_SPT_Group_cohomology}. One can reconstruct the cohomology, or ``symmetry fractionalisation", via the string order parameters of the SPT~\cite{Verresen_PRX_2021}. Note that each of the phases we thus produce may have a different ``trivial" subgroup $M$. In Appendix~\ref{app:SPT_Lagrangian_subgroup_to_cohomology}, we show how to reconstruct the cohomology of the SPT phase from the corresponding Lagrangian subgroup.

\subsubsection{\texorpdfstring{Example: $G = \mathbb{Z}_4\times \mathbb{Z}_6\times \mathbb{Z}_3$}{Example: G = Z4 x Z6 x Z3}}

We have covered the example $G = \mathbb{Z}_2 \times \mathbb{Z}_2$ above, but let us note that in that case, setting $H=G$ (so the symmetry is unbroken), we have, according to Equation~\ref{eqn:SPT_class}, a $\mathbb{Z}_2$ classification of the SPT phases, which is indeed the case---previously, we found only one non-trivial SPT phase and one trivial phase (these were the only two phases with completely unbroken symmetry).
Here we examine the example of $G = \mathbb{Z}_4\times \mathbb{Z}_6\times \mathbb{Z}_3$, shown in Fig.~\ref{fig:Static_phases}(b). From Equation~\ref{eqn:SPT_class}, we expect to find a $\mathbb{Z}_2\times \mathbb{Z}_3$ classification of SPT phases with no broken symmetry. As usual, we start from the trivial phase with Lagrangian subgroup:
\begin{align}
    \mathcal{M}_\text{triv} = \langle m_1, m_2, m_3\rangle,
\end{align}
generated by $m$-type anyons of the three layers of $\mathbb{Z}_4$-, $\mathbb{Z}_6$-, and $\mathbb{Z}_3$-toric code.
Let us consider the possible $e$-anyon attachments that can be made. We see that the following is an allowed attachment between layers 1 and 2:
\begin{align}
    \mathcal{M}_1 = \langle m_1 e_2^3 , e_1^2 m_2, m_3\rangle .
\end{align}
Indeed, this is the only valid attachment that maintains trivial statistics between the anyons in layers 1 and 2 (see Appendix~\ref{app:SPT_Lemma}). 

There are no possible attachments between layers 1 and 3 since gcd$(4,3) = 1$. Between layers 2 and 3 we have two non-trivial attachments:
\begin{align}
    \mathcal{M}_2 &= \langle m_1 , m_2 e_3 , e_2^4 m_3\rangle , \\
    \mathcal{M}_3 &= \langle m_1 , m_2 e_3^2 , e_2^2 m_3 \rangle .
\end{align}
We can finally perform layer 1/2 attachments and layer 2/3 attachments sequentially:
\begin{align}
    \mathcal{M}_4 &= \langle m_1 e_2^3 , e_1^2 m_2 e_3 , e_2^4 m_3\rangle , \\
    \mathcal{M}_5 &= \langle m_1 e_2^3, e_1^2 m_2 e_3^2 , e_2^2 m_3 \rangle .
\end{align}
Hence, we end up with the expected 6 SPT phases and a $\mathbb{Z}_2\times \mathbb{Z}_3$ classification, the first factor corresponding to the layer 1/2 attachment and the second to the layer 2/3 attachments. 

We sketch out the different types of static phases we consider in Table~\ref{tab:static_phases} including their layered toric code picture.

\begin{table*}
    \centering
    \begin{tblr}{
        colspec = {Q[c,m,wd=2.2cm] X[c,m] X[c,m] X[c,m]},
        colsep  = 7pt,
        rowsep  = 8pt,
        hline{1,2} = {-}{0.75pt},
        vline{2} = {1-4}{0.75pt},
        row{1} = {font=\bfseries},
        row{2} = {abovesep=20pt},
        column{1} = {font=\bfseries},
        cell{3}{1} = {valign=m},
    }
            &
        {Symmetry-breaking\\ to $H\leq G$}
        &
        {$(M\leq H)$-trivial\\ phase}
        &
        {$(H/M)$-SPT\\ phase} \\

    {Lagrangian\\ Subgroup $\mathcal{M}$}
        &
        {$\displaystyle
        \bigl\{ (1,\beta)\;\big|\;
        \beta \in \mathrm{Rep}\,G,$\\
        $\beta(h)=1\;\forall h\in H
        \bigr\}
        \subseteq \mathcal{M}$}
        &
        $(m,1) \in \mathcal{M},\; \forall m\in M$
        &
        {$\displaystyle
        \bigl\{ (h,\alpha)\;\big|\;
        h\in H,\,
        \alpha\in \mathrm{Rep}\,H/M,$\\
        $\alpha(h)\neq 1 \text{ for some } h\in H
        \bigr\}
        \subseteq \mathcal{M}$} \\[8pt]

    \raisebox{3.5\height}{\parbox{2.2cm}{
\centering
Layered $\mathbb{Z}_k$-toric\\
code example
}}
        &
        \includegraphics[width=0.9\linewidth]{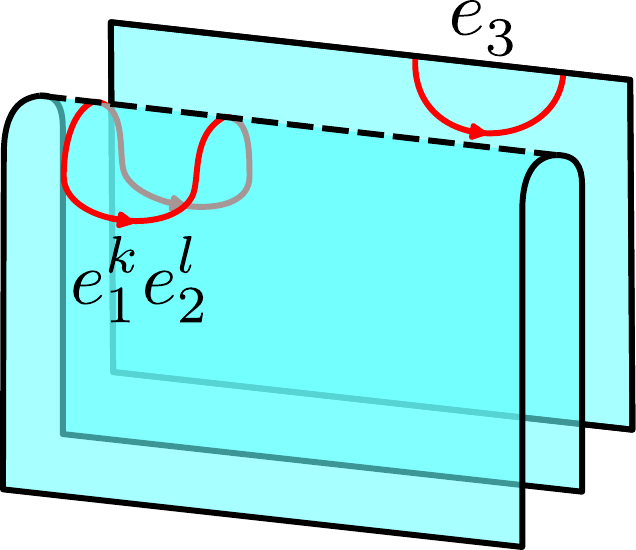}
        &
        \includegraphics[width=0.9\linewidth]{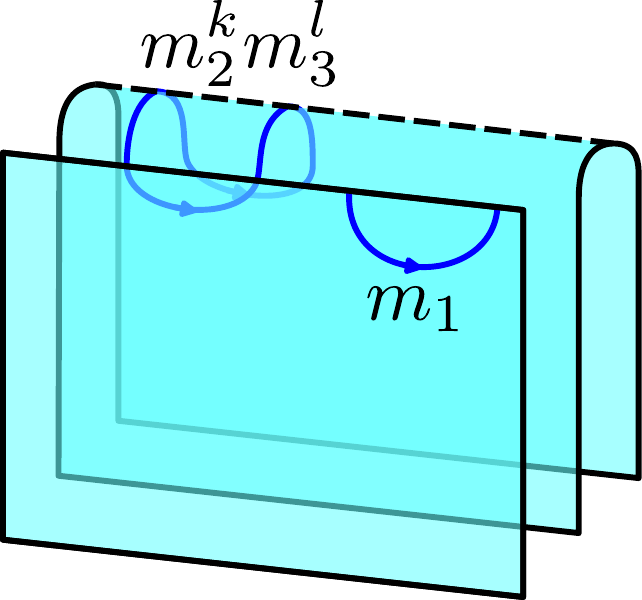}
        &
        \includegraphics[width=0.9\linewidth]{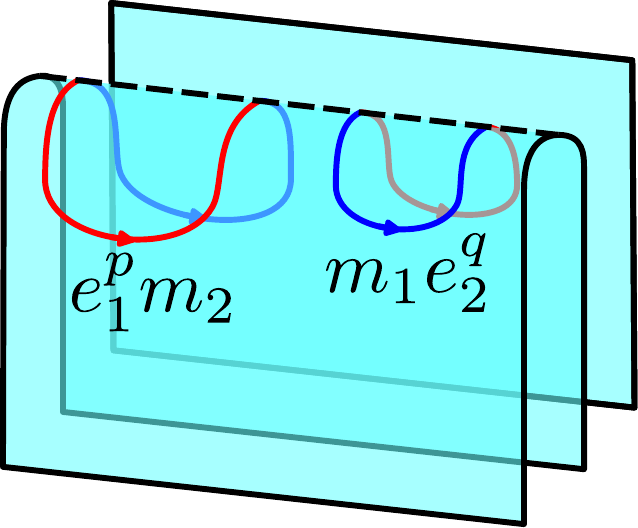} \\[12pt]

    \end{tblr}

    \caption{Summary of static phases understood through SymTFT.
    We provide a characterisation of the Lagrangian subgroup corresponding to each type of order: spontaneous symmetry-breaking to a subgroup $H\leq G$, trivial with respect to subgroup $M\leq H$ (the unbroken symmetry group), and SPT-ordered with respect to symmetry group $H/M$ (the phases being a subset of those that are trivial with respect to $M$).
    We provide a schematic illustration of an example of a layered $\mathbb{Z}_k$-toric code system whose boundary exemplifies each particular order [$B_{\text{phys}}$ ($B_{\text{ref}}$) is at the top (bottom)]. Examples of string operators condensed at $B_{\text{phys}}$ are shown (we do not show all condensed string operators).
    In the SSB phases, $e$-type anyons are included in $\mathcal{M}$. Any pure $m$-type anyons in $\mathcal{M}$ generate the subgroup $M$ with respect to which the phase is trivial. SPT order corresponds to the toric code layers folding along a domain wall that mixes $e$ and $m$ anyons, as shown.}
    \label{tab:static_phases}
\end{table*}

\subsection{Boundary conditions}\label{sec:static_BCs}

For the SymTFT on the cylinder, $B_\text{phys}$ has periodic boundary conditions to begin with. However, we can twist these boundary conditions or we can even introduce open boundaries  in natural ways. We will not discuss open boundaries at length in this paper (see Section~\ref{sec:Conclusion} for further discussion of these).
One can introduce these by simply adding a segment of a trivial phase along $B_\text{phys}$,
which acts as the vacuum to which the 1D system interfaces~\cite{motamarri2024symtftequilibriumtimecrystals}.
This may involve an interface between 1D boundaries with different Lagrangian subgroups, at which point one expects to find \textit{duality defects} or \textit{twist defects} of the $G$-TO~\cite{Bombin10,BarkeshliQi_PRX12,YouWen12,Barkeshli13c,BBCW19,motamarri2024symtftequilibriumtimecrystals}. 

In the context of closed boundary conditions, we can ``$g$-twist" boundary conditions with an operator $W^{(g,1)}_j$ running between $B_\text{phys}$ and $B_\text{ref}$ (see also Ref.~\cite{motamarri2024symtftequilibriumtimecrystals}). As discussed in Section~\ref{sec:G_TO}, in order to avoid creating excitations at $B_\text{ref}$, we introduce a short electric-magnetic duality domain wall. Any $(g,1)$ string operator that passes through this wall switches to $(1,\alpha)$-type and vice versa (note that $G$ and Rep$\, G$ are isomorphic and so we can define such a domain wall unambiguously). We pass the string operator $W^{(g,1)}_j$ through this domain wall, changing it to a $W^{(1,\alpha)}$ operator, then terminate this on $B_\text{ref}$ without creating excitations. To understand the effect of this inclusion, note that a translation of the 1D system by $L$ lattice spacings drags the $(g,1)$ anyon around the cylinder, resulting in the system returning to its original state up to a global symmetry operator, $\overline{W}^{(g,1)}$~\cite{aasen2016topological,Aasen_cat,motamarri2024symtftequilibriumtimecrystals}. This $g$-twisted boundary condition is detected by operators $\overline{W}^{(1,\alpha)}$ wrapping around $B_\text{phys}$: $\overline{W}^{(1,\alpha)} W^{(g,1)}_j \overline{W}^{(1,\alpha)\dagger} = \alpha(g) W^{(g,1)}_j$. For $G = \mathbb{Z}_2$ this simply gives us periodic (no twist) or anti-periodic ($(-1)$-twisted) boundary conditions.

\section{SymTFT for Floquet phases}\label{sec:SymTFT_Floquet_phases}

We now discuss the application of SymTFT to Floquet 1D systems. 
We will firstly discuss MBL from a SymTFT perspective, before defining the general form of the drive that we will consider. We will then provide a general classification of Floquet drives, demonstrate signatures of time-translation symmetry breaking (TTSB),
discuss boundary conditions, and the relation of our classification scheme to previous schemes. We also provide some illustrative examples.

\subsection{\texorpdfstring{MBL for boundaries of the $\mathbf{G}$-TO}{MBL for boundaries of the G-TO}}\label{sec:MBL_bdrys}

MBL is often defined for a general system as the existence of an extensive set of quasi-local conserved quantities, i.e., the LIOMs, whose eigenvalues label the eigenstates of the system~\cite{serbyn2013local,Huse_MBL_phenom_14,chandran2015constructing,ros2015integrals,Rademaker2016LIOM,Vasseur2016}. In our case, suitably dressed versions of a subset of the boundary algebra terms (see Section~\ref{sec:Boundary_algebra}) provide natural candidates for such LIOMs.

Since we are dealing with Abelian $G$, there are no obstructions to MBL in the boundary algebra, in either symmetry-broken or symmetry-preserving cases, unlike for cases with non-ablelian symmetry groups~\cite{KeyserlingkSondhi2016floquetSPT,KeyserlingkSondhi2016floquetSSB,Vasseur2016}. %
The bulk of the system is not supplied with any disorder and has a large energy gap relative to the disorder strength of $B_\text{phys}$, so all bulk stabiliser values are frozen out to $A_v, B_p = 1$. All perturbations that we add to the drives will commute with the bulk stabilisers, as the perturbations are elements of the boundary algebra, and so will not frustrate the bulk stabilisers. As such, the MBL is restricted to the 1D boundary $B_\text{phys}$.
Starting in the fixed-point limits,
we build drives from Hamiltonians (which include only boundary algebra terms) analogous to the $H_0$ and $H_1$ Hamiltonians in Section~\ref{sec:Illustrative_example} in the $G=\mathbb{Z}_2$ Ising case. Specifically, we will include terms in a Hamiltonian that correspond to anyon string operators corresponding to a particular Lagrangian subgroup. Eigenstates of these Hamiltonians are labelled by eigenvalues $\lbrace w_i\rbrace$ of those $W_{i\, i+1}^{(a)}$ operators for $a\in \mathcal{M}$, along with possibly up to rk$G$ other labels for symmetry-broken states: $\ket{\Psi} = \ket{\lbrace w_i\rbrace,p_1,\ldots,p_{\text{rk}G}}$. The string operators $W_{i\, i+1}^{(a)}$ form the LIOMs. Adding perturbations, we expect to be able to find a quasi-local unitary $\mathcal{U}$ (which acts trivially in the bulk)
that diagonalises these Hamiltonians, for strong-enough disorder, cf.~Sec.~\ref{sec:AFFPD}
(such systems have been extensively studied in 1D~\cite{MBL_spin_chains_2016,abanin2014theory,Huse_localization_2013,serbyn2013local,Huse_MBL_phenom_14,chandran2015constructing,ros2015integrals,Goihl2018,ImbrieLIOMreview2017}). $\mathcal{U}$ dresses the LIOMs with exponentially-decaying tails. The exact forms of the boundary algebra operators do not matter, since a finite-depth unitary can alter them to other local operators obeying the same algebraic relations. 

\subsection{Form of Floquet drives}

We begin with an ansatz for the Floquet drives. This involves
Hamiltonians
\begin{align}\label{eqn:HM_floq}
    H_\mathcal{M} = \sum_{a\in \mathcal{M}}\sum_{i=1}^L J_i^{(a)} \left(W_{i\, i+1}^{(a)} + W_{i\, i+1}^{(a)\dagger}\right)
\end{align}
built from the members of the boundary algebra, where $\mathcal{M}\subset \mathcal{A}_G$ is a Lagrangian subgroup of the anyon model of the $G$-TO, which ensures that all the string operators mutually commute.
We consider first these fixed-point Hamiltonians, and avoid longer-range interactions such as $W_{i\, i+2}^{(a)}$ for simplicity, before adding perturbations later.
More generally, we could consider Hamiltonians $H_a$ involving string operators only of a single anyon type $a\in \mathcal{A}_G$.
Using these $H_a$, we could form a large set of general Floquet drives of the form
\begin{align}
    U_F^{a,b,c,\ldots} = \ldots e^{-iH_{c}} e^{-i H_{b}} e^{-i H_{a}}, 
\end{align}
for $a,b,c,\ldots \in \mathcal{A}_G$. However, to construct fixed-point Hamiltonians, we require exactly local integrals of motion. Hamiltonians $H_\mathcal{M}$ for $\mathcal{M}$ a Lagrangian subgroup provide a complete set of these LIOMs. Hence we specialise our ansatz to:
\begin{align}\label{eqn:UF_Mabc}
    U_F^{\mathcal{M},a,b,c,\ldots} = \ldots e^{-iH_{c}} e^{-i H_{b}} e^{-i H_{a}} e^{-iH_\mathcal{M}}
\end{align}
for some Lagrangian subgroup $\mathcal{M}$.
To construct fixed-point drives, we require that $e^{-i H_a},e^{-iH_b},\ldots$ all commute exactly with the LIOMs. 
In this case, it is clear that either $H_a, H_b, \ldots$ will involve LIOMs themselves (i.e., $a,b,\ldots \in \mathcal{M}$), or one or more of $e^{-iH_a},e^{-iH_b},\ldots$ forms a logical operator:
\begin{align}
    e^{-iH_a} = \overline{W}^{(a)}.
\end{align}
No other possibilities allow for Equation~\ref{eqn:UF_Mabc} to describe a fixed-point drive. 
Hence, the general form of our fixed-point Floquet drives will be taken to be:
\begin{equation}
    U^{(b,\mathcal{M})}_F = \overline{W}^{(b)} e^{-i H_{\mathcal{M}}},
    \label{eqn:fp_drives}
\end{equation}
where
$H_{\mathcal{M}}$ is a Hamiltonian of the form of Equation~\ref{eqn:HM_floq}.
Hence, the fixed-point drives are associated with a Lagrangian subgroup $\mathcal{M}$ and an anyon $b\in \mathcal{A}_G$
(note that $\overline{W}^{(\mathbf{1})} = \mathds{1}$). 

The above simple form for Floquet drives allows us to consider phases with eigenstate order and non-trivial dynamics via TTSB. The prepending logical operator $\overline{W}^{(b)}$ is crucial for this latter purpose as we show below. We will also show that the above ansatz is general enough to recover previously classified
forms of SSB and SPT Floquet drives for appropriate choice of $\mathcal{M}$ and $b$. 

\subsection{Spatiotemporal order of Floquet eigenstates}\label{sec:ST_order_eigenstates}

We can examine the spatiotemporal order of drive eigenstates via the decay of disorder correlators, which can capture TTSB and long-range spatial order both in the presence and absence of $G$-spontaneous symmetry breaking.
It is also general enough to capture boundary or dual signatures of TTSB, such as those exemplified by the $0\pi$PM phase (cf. Sec.~\ref{sec:Illustrative_example}). Disorder correlators are built from finite-interval symmetry operator strings, which obey ``volume law" decay when the symmetry is broken, but ``area law" decay otherwise. We use the adaptation of this to TTSB as suggested in Appendix B of Ref.~\onlinecite{KhemaniKeyserlingkSondhi2017RepTh}. That is, we study how the expectation of $U^{(b,\mathcal{M})}_F(\ell,r)= W^{(b)}(\ell,r) e^{-i H_{\mathcal{M}} (\ell,r)}$ decays in the eigenstates $\ket{n}$ of $U_F^{(b,\mathcal{M})}$.
Here, $U_F^{(b,\mathcal{M})}(\ell,r)$ is defined as $U_F^{(b,\mathcal{M})}$ restricted to a finite-interval $[\ell,r]$.
We define the disorder correlator to be:
\begin{align}
    T_{\ell r, n} \coloneqq |\bra{n} U^{(b,\mathcal{M})}_F(\ell,r)\ket{n}|.
\end{align}

Let us consider the disorder correlator in two cases: (1) $b\in \mathcal{M}$, and (2) $b\notin \mathcal{M}$.
We will start with the fixed-point drive
first, before extending to perturbed drives. 
In case (1), the string operator $W^{(b)}(\ell, r)$ corresponds to an anyon $b$ in the Lagrangian subgroup, and so $W^{(b)}(\ell, r)$ fixes eigenstates $\ket{n}$ up to a phase. $e^{-iH_\mathcal{M}(\ell,r)}$ is also a function of the LIOMs (up to some finite-size correction at the boundaries).
In that case, $T_{\ell r, n} = |\bra{n} U^{(b,\mathcal{M})}_F(\ell,r)\ket{n}| = \text{constant}$. This holds away from the fixed point, as we explain in Appendix~\ref{app:extending_from_fixed_points}.
Meanwhile, in case (2), $W^{(b)}(\ell, r)$ flips LIOMs on $B_\text{phys}$ at locations $\ell$ and $r$, since $b$ is not a member of the Lagrangian subgroup. Hence in this case, $T_{\ell r, n} = |\bra{n} W^{(b)}(\ell,r)\ket{n}| = 0$. Away from the fixed point, this gets modified to an exponentially decaying function of $|\ell - r|$, resulting from the exponential tails in the LIOMs. We explain this in Appendix~\ref{app:extending_from_fixed_points}.

Therefore, we find, for a general drive:
\begin{align}
   |T_{\ell r,n}| &\coloneqq |\bra{n} U^{(b,\mathcal{M})}_{F}(\ell,r)  \ket{n}| \nonumber\\
    &= \begin{cases}
    \sim e^{-|\ell-r|/\xi}, & \text{if $b\notin\mathcal{M}$}\\
    \text{constant}, & \text{if $b\in \mathcal{M}$}
    \end{cases}\label{eqn:Disorder_corr_two_cases}
\end{align}
for any eigenstate $\ket{n}$ of the drive, where $\xi$ is the localisation length, which for the fixed-point drive is zero. Hence, we find that TTSB, as probed by the disorder correlator, is determined by the inclusion or exclusion of $b$ in $\mathcal{M}$ in the fixed-point drive $U_F^{(b,\mathcal{M})}$.

The disorder correlators have the same characteristic decay regardless of if there is any $G$-SSB (this is an advantage of the approach), or equally, regardless of if the TTSB occurs in the bulk or if, without the use of dual order parameters (as in the $0\pi$PM phase for the $\mathbb{Z}_2$-symmetric system), it is only detected on the boundaries
of the 1D system. The TTSB is completely determined by the choice of $b$ and $\mathcal{M}$. We can explicate this further by using the patch symmetries method from Ref.~\cite{JiWen2020categorical}. We consider both the Floquet unitary restricted to interval $[\ell, r]$ and a string operator $W^{(a)}_{kj}$ for some $a\in \mathcal{M}$,
with interval $[k,j]$ chosen such that $k$ lies deep in the bulk of $[\ell,r]$ and $j$ lies far away from $[\ell,r]$. We then have:
\begin{align}
    \bra{n} &U_F^{(b)}(\ell,r)^{t\,\dagger} W^{(a)}_{kj} U_F^{(b)}(\ell,r)^{t}\ket{n} = \nonumber\\
    & \mathcal{B}(a,b)^{t} \bra{n}W^{(a)}_{kj}\ket{n} \label{eqn:LR_ST_order}%
\end{align}
where $\bra{n}W^{(a)}_{kj}\ket{n} \neq 0$ in the limit of $|k-j|\rightarrow \infty$.
This displays the characteristic period-tupling of TTSB, resulting from the braiding phase between anyons $b$ and $a$. Once again, this diagnosis of TTSB in $\ket{n}$ eigenstates is the same regardless of whether there is any $G$-SSB. %
Equation~\ref{eqn:LR_ST_order} signals long-range \textit{spatiotemporal} order (it is an unequal-time correlator), which is a defining characteristic of TTSB.

The form of the anyon $a$ determines the type of time crystalline order we expect. As we have described in Section~\ref{sec:Illustrative_example}, one can think of TTSB occurring due to the inclusion of symmetry or dual-symmetry operators in $U_F$. When the anyon $b$ is of the form $(g,1)$, $\overline{W}^{(b)}$ is a symmetry operator, whereas when $b = (1,\alpha)$, it is a dual symmetry operator. Both phases can exhibit bulk features if we consider the (gauge) link degrees of freedom in the lattice (see Section~\ref{sec:breaking_dual_Ising_sym}). In full generality, $b=(g,\alpha)$ results in a breaking of both standard and dual symmetries. We provide an experimental protocol in Section~\ref{sec:Experimental_demo} for measuring correlators in all phases. We note that away from fixed points, we may only have emergent symmetries (or dual symmetries), if the perturbations we add to the drive do not commute with the symmetry operators~\cite{KeyserlingkKhemaniSondhi2016stability}.

In this subsection we have shown that the spatiotemporal order of eigenstates of the drive is determined by the choice of $b$ and $\mathcal{M}$ in $U_F^{(b,\mathcal{M})}$, with the qualitative form of all correlators holding away from the fixed point drives, owing to MBL (see Appendix~\ref{app:extending_from_fixed_points}). In short, if $b\notin \mathcal{M}$ ($b\in \mathcal{M}$), the disorder correlator decays (does not decay) to zero (Equation~\ref{eqn:Disorder_corr_two_cases}) and we observe (do not observe) period-tupling of observables (Equation~\ref{eqn:LR_ST_order}). Equivalently, we observe (do not observe) TTSB in eigenstates of the drive. In the next subsection, we leverage this to provide a classification of Floquet drives.

\subsection{Classification of Floquet Drives}
\label{sec:CFD}

We have seen that SymTFT allows us to label Floquet drives with a pair, $(b,\mathcal{M})$, where $\mathcal{M}$ determines the eigenstate order and $b$ characterises the time-crystalline/TTSB signatures. In particular, we observe non-trivial dynamic signatures whenever $b\notin \mathcal{M}$. For any anyon $a \in \mathcal{M}$, one can create equivalent signatures with drives $U_F^{(b,\mathcal{M})}$ and $U_F^{(b\times a,\mathcal{M})}$. Physically, there is an equivalence between $b$ and $b\times a$ at the boundary of the $G$-TO, since said boundary condenses $a$ anyons. This is essentially an identification of $a$ with the vacuum particle, and hence $b$ and $b\times a$ are entirely indistinguishable. From the perspective of the 1D system, we can write $U_F^{(b\times a,\mathcal{M})} = \overline{W}^{(b)}\overline{W}^{(a)}e^{-iH_\mathcal{M}}$, with $\overline{W}^{(a)}$ expressible in terms of the LIOMs. $\overline{W}^{(a)}$ therefore provides only an overall phase to eigenstates of $U_F^{(b,\mathcal{M})}$ and hence can be seen to provide no extra dynamics. We define the class of ``equivalent excitations" on the boundary defined by $\mathcal{M}$ as $[b] = \lbrace b\times a \, |\, a\in \mathcal{M}\rbrace$.

One can thus classify distinct possible Floquet drives for a $G-$symmetric system with a Lagrangian subgroup, $\mathcal{M}$, and an excitation class $[b]$. We summarise this as:
\begin{equation}
    \text{Cl}_{\text{F}}(G) =  \{ \mathcal{M} \times \mathcal{A}_G / \mathcal{M}  |  \mathcal{M} \in \text{Cl}_{\text{Lag}}(\mathcal{A}_G) \} 
\end{equation}
where $\mathcal{A}_G / \mathcal{M}$ denotes different equivalence classes of excitations for Lagrangian subgroup $\mathcal{M}$, and Cl$_\text{Lag}(\mathcal{A}_G)$ is the set of all Lagrangian subgroups of anyon model $\mathcal{A}_G$.

Further, we can look at the classification of drives given symmetry breaking $G \to H$. For an unbroken symmetry group $H$, the possible eigenstate orders $\mathcal{M}$ are given by the static/undriven SPT classification, which we have covered in Section~\ref{sec:Simple_SPT_Class}. We denote this static SPT classification as $\text{Cl}_0(H)$. Recall that the classification is given by a cocycle from $H^2(H,U(1))$. Given such an $\mathcal{M}$, one can show that its equivalence classes of excitations follow: $\mathcal{A}_G / \mathcal{M} \equiv G/H \times \text{Rep}(H)$ (see Appendix~\ref{app:excitation_classes}). Thus, we get the Floquet classification, when breaking to subgroup $H$: 
\begin{equation}\label{eqn:sym_broken_class_scheme}
    \text{Cl}_{\text{F}}(G,H) = \text{Cl}_0(H) \times G/H \times \text{Rep}(H).
\end{equation}
The drives producing TTSB %
are those in which the $\mathcal{A}_G/\mathcal{M}$ component is non-trivial.

\subsection{Observable Signatures and Twisted Boundary Conditions}\label{sec:sig}

We here discuss observable unequal time correlators that detect long-range spatiotemporal order in TTSB phases established in our classification scheme. We need not introduce open boundary conditions to observe TTSB in any of these phases (it is usually thought that Floquet SPT phases, such as the $0\pi$PM phase, have signatures that are only observable at open boundaries). Specifically, local order parameters or local dual order parameters display these correlations. We distinguish between time-crystalline phases, in which local order parameters display long-range spatiotemporal order, and dual time-crystalline phases, in which that is displayed by dual order parameters. $B_\text{phys}$ can display both of these types of time-crystalline order. Specifically, for any $a\in \mathcal{M}$ and Floquet eigenstate $\ket{n}$
we have [cf.~Eq.~\eqref{eq:Z2sig}]:
\begin{align}
    &C^{(a)}_{kj;\ket{n}}(t)=
    \bra{n} W_k^{(a)\dagger}(t)W_j^{(a)}\ket{n}  \label{eq:dyncorr}\\
    &\quad \quad =\bra{n}\left(U^{(b,\mathcal{M})\dagger}_F\right)^t W_k^{(a)\dagger}\left(U_F^{(b,\mathcal{M})}\right)^tW_j^{(a)} \ket{n}\\
    &\quad \quad= \mathcal{B} (a,b)^t \bra{n}W_{jk}^{(a)}\ket{n}
\end{align}
where $\bra{n}W_{jk}^{(a)}\ket{n}\neq 0$ in the limit of $|k-j|\rightarrow\infty$ (note that its value depends on $k$ and $j$, due to the values of the LIOMs between these sites). The correlator displays a characteristic period-tupling, dictated by the braiding between $a$ and $b$. This holds also away from the fixed point, provided the 1D system at $B_\text{phys}$ is many-body localised.

For $\mathcal{B} (a,b)\neq1$
with $a = (1,\alpha)\in \mathcal{M}$ and $b\notin \mathcal{M}$, we have regular time-crystalline order, whereas for $a=(g,1)\in\mathcal{M}$, $b\notin \mathcal{M}$, we have a dual time crystal. In the latter case,
the order parameters are operators $W^{(g,1)}_k$ that $g$-twist boundary conditions. That is, by making boundary conditions dynamic---Eq.~\eqref{eq:dyncorr} correlates twisting and later untwisting boundary conditions---one can observe signatures of the dual time crystals with closed boundary conditions. More generally, one can expect order parameters to correspond to anyons that feature both $G$ components and Rep$\, G$ components.

We now discuss a further feature of these dynamical boundary conditions. In addition to using $W^{(g,1)}_k$ as a dual order parameter, we can also make it part of the Floquet unitary:
$U_F^{(b,\mathcal{M}),g} \propto W^{(g,1)}_k \overline{W}^{(b)} e^{-iH_\mathcal{M}}$, with the proportionality constant simply a phase chosen so that the pre-factor is Hermitian. %
More generally, one can include $W^{(a)}_k $ with  $a=(g,\alpha)$ in the drive:
$U_F^{(b,\mathcal{M}),a}\propto W^{(a)}_k \overline{W}^{(b)} e^{-iH_\mathcal{M}}$.
We require that $a\in \mathcal{M}$,
otherwise $W_k^{(a)}$ will not commute with all of the LIOMs. 

To see the effect of this inclusion in the Floquet unitary, we can consider the dynamics of the $\overline{W}^{(b)}$ operator under the modified drive: $U_F^{(b,\mathcal{M}),a\dagger}\overline{W}^{(b)}U_F^{(b,\mathcal{M}),a} = \mathcal{B}(b,a) \overline{W}^{(b)}$. Hence, we find that $W^{(a)}_k$
is associated with a fixed-quasi-energy mode (generalising zero and $\pi$ modes~\cite{Floquet_Majoranas_PRL_2011,Kitagawa_top_char_2010,Thakurathi_Floquet_majorana_2013}). 
When we include $W_k^{(a)}$ in the Floquet unitary, we can dynamically cycle through eigenstates of $\overline{W}^{(b)}$. We can interpret this as dynamically pumping charge and/or twisting boundary conditions.
This can be made more general still %
by duality-twisting boundary conditions~\cite{motamarri2024symtftequilibriumtimecrystals}.

\subsection{Relation to previous classification schemes}

In this section, we describe how the SymTFT approach to classifying Floquet phases relates to other schemes appearing in the literature~\cite{KeyserlingkSondhi2016floquetSSB,else2016prb,potter2016prx,KeyserlingkSondhi2016floquetSPT}, and how our approach extends these classification schemes.
The schemes for classifying Floquet SPT phases~\cite{else2016prb,potter2016prx,KeyserlingkSondhi2016floquetSPT}
can be understood via projective representations of an enlarged symmetry group, $G\times \mathbb{Z}$, accounting for both the familiar $G$ symmetry along with discrete time-translation symmetry $\mathbb{Z}$.
This is analogous to how static SPT phases can be understood via projective representations of $G$ (see Appendix~\ref{app:SPT_class}). We can provide simple counting arguments to show that our classification reproduces these previously known results:
Refs.~\onlinecite{else2016prb,potter2016prx,KeyserlingkSondhi2016floquetSPT} identify Floquet SPT phases with elements of the second cohomology group $H^2(G\times \mathbb{Z}, U(1))$, leading to a number $|\text{Cl}_0 [G] \times G|$ of Floquet SPT phases, where $\text{Cl}_0[G]$ labels the static SPT order. This matches our scheme when the full symmetry is preserved ($H = G$), as one can see from Eq.~\ref{eqn:sym_broken_class_scheme}.

We now seek to better understand the structural
relationship between these classification schemes and ours. %
The classification of static SPT phases has been accounted for already -- it can be determined from the non-trivial braiding statistics of the anyon model. In Appendix~\ref{app:SPT_Lagrangian_subgroup_to_cohomology}, we provide a graphical route to explicitly calculating the cocycle in $H^2(G, U(1))$ that labels the SPT phase, based on the corresponding Lagrangian subgroup. This takes care of the Cl$_0[G]$ grading.
Meanwhile, the extra data afforded by the $\times G$ component of the classification tells about the action of the symmetry on the patch-restricted Floquet unitary, $U_F(\ell, r)$.

For concreteness, let us focus on the case of $G = \mathbb{Z}_2$, with the non-trivial Floquet topological phase corresponding to $U_F^{(e,\mathcal{M}_m)} = \overline{W}^{(e)}e^{-i t H_{\mathcal{M}_m}}$ with $\mathcal{M}_m = \lbrace 1, m \rbrace$ (recall that $\mathcal{M}_\text{ref} = \lbrace 1, e\rbrace$). We restrict the unitary to interval $[\ell, r]$, %
leading to: $U_F^{(e,\mathcal{M}_m)}(\ell,r) = W^{(e)}_{\ell r} e^{-i  H_{\mathcal{M}_m}(\ell,r)} = W_\ell^{(e)} W_r^{(e)} e^{-i t H_{\mathcal{M}_m}(\ell,r)}$. 
The classification of this phase, using the scheme of Refs.~\cite{else2016prb,potter2016prx,KeyserlingkSondhi2016floquetSPT}, arises from the non-commutation of the endpoint operators $W_\ell^{(e)}$ and $W_r^{(e)}$ with the symmetry $\overline{W}^{(m)}$: $\overline{W}^{(m)}W_\ell^{(e)}\overline{W}^{(m)\dagger}W_\ell^{(e)\dagger} = \mathcal{B}(e,m) = -1$. Here, $\mathcal{B}(e,\cdot): G\rightarrow U(1)$ is a representation of $G$ that
labels the Floquet SPT phase, in the framework of Ref.~\onlinecite{KeyserlingkSondhi2016floquetSPT}, and results in a projective representation of the symmetry group $G\times \mathbb{Z}$ near site $\ell$~\cite{potter2016prx,else2016prb}. That is, we define the symmetry group action restricted to the sites near $\ell$ as generated by $W_\ell^{(e)}$, $W^{(m)}_{\ell-1\, \ell+1}$. But this forms only a projective representation of $\mathbb{Z}_2\times \mathbb{Z}$ since the generators anti-commute. Such projective representations are labelled by elements of $H^2(G\times \mathbb{Z},U(1))$. Hence, we see that in our scheme, the topological Floquet classification arises via the braiding statistics of anyons in SymTFT. 

The number
of non-trivial Floquet SPT phases that supplement the static SPT phases are the number of excitation classes for the boundary, since in such cases the eigenstate expectation of the patch-restricted $U_F$ decays exponentially to zero, signaling TTSB (see Equation~\ref{eqn:Disorder_corr_two_cases}). There are $|G|$ such anyons, and they are of the form $(1,\alpha)$ for $\alpha \in \text{Rep}\, G$, since such anyons are not in $\mathcal{M}$ (otherwise there would be some symmetry-breaking) and they braid non-trivially with some $(g,1)$.
This is how SymTFT accounts for the Cl$_0[G]\times G$ classification of Floquet SPT phases of Refs.~\onlinecite{else2016prb,potter2016prx,KeyserlingkSondhi2016floquetSPT}.
The appearance of $W_\ell^{(a)}$ and $W_r^{(a)\dagger}$
at the patch-restricted $U_F$'s endpoints can be interpreted to provide a non-trivial ``pumped charge'', matching this feature in the SPT characterisation
of Ref.~\onlinecite{KeyserlingkSondhi2016floquetSPT}.

Ref.~\onlinecite{KeyserlingkSondhi2016floquetSSB} provides a classification of Floquet fully symmetry-broken phases. For Abelian $G$, they find $|G|$ Floquet SSB phases.
In our framework, a fully symmetry-broken phase corresponds to $H = \{1\}$, with Cl$_0[H]$ being trivial. From Equation~\ref{eqn:sym_broken_class_scheme}, we see that in such a case, we have $\text{Cl}_F(G,\{1\}) = G$. For these phases, we choose $B_\text{phys}$'s Lagrangian subgroup to be $\mathcal{M}_{\text{Rep}\, G}$ and indeed we can find $|G|$ distinct drives corresponding to the excitation classes $(g,1)$ for $g\in G$.

While our scheme captures these two regimes (no $G$-SSB and full $G$-SSB), it also captures additional phases. Specifically, Equation~\ref{eqn:sym_broken_class_scheme} captures the classification %
with partial $G$-symmetry breaking, both for static phases and dynamic phases with
(dual) time-crystalline order. These have not been previously classified in the literature.

\begin{table*}[!t]
    \centering
    \begin{tblr}{
        colspec = {X[c,m] X[c,m] X[c,m] X[c,m]},
        colsep  = 12pt,
        rowsep  = 6pt,
        hline{1,2} = {-}{0.75pt},
        row{1} = {font=\bfseries},
    }
    Static Ordering
        & Lagrangian Subgroups
        & Excitation Classes
        & {Non-trivial\\ spatiotemporal order} \\

    Full SSB
        & $\mathcal{M} = \langle e_1, e_2\rangle$
        & $b \in \langle m_1, m_2 \rangle$
        & $\mathbb{Z}_2\text{TC}\times \mathbb{Z}_2 \text{TC}$ \\[6pt]

    {Partial SSB\\ (to a $\mathbb{Z}_2$ subgroup)}
        &
        $\displaystyle
        \mathcal{M} =
        \begin{cases}
            \langle e_1, m_2\rangle ,\\
            \langle m_1, e_2\rangle ,\\
            \langle e_1 e_2, m_1 m_2\rangle
        \end{cases}$
        &
        $\displaystyle
        b \in
        \begin{cases}
            \langle m_1, e_2\rangle ,\\
            \langle e_1, m_2\rangle ,\\
            \langle e_1, m_1 \rangle
        \end{cases}$
        &
        $\mathbb{Z}_2\text{TC} \times \mathbb{Z}_2 \text{DTC}$ \\[10pt]

    SPT
        & $\mathcal{M} = \langle e_1 m_2, m_1 e_2\rangle$
        & $b \in \langle e_1, e_2\rangle$
        & $\mathbb{Z}_2 \text{DTC}\times \mathbb{Z}_2 \text{DTC}$ \\[6pt]

    Trivial
        & $\mathcal{M} = \langle m_1, m_2\rangle$
        & $b \in \langle e_1, e_2\rangle$
        & $\mathbb{Z}_2 \text{DTC}\times \mathbb{Z}_2 \text{DTC}$
    \end{tblr}

    \caption{The possible Floquet phases for the $G= \mathbb{Z}_2\times \mathbb{Z}_2$-symmetric spin chain, with the associated (via SymTFT) Lagrangian subgroups $\mathcal{M}$, defining the static ordering, and excitation classes $b$, defining the temporal ordering. We also provide the non-trivial spatiotemporal order possible in the associated phases -- either a time crystal (TC) or dual time crystal (DTC) ordering, or both.}
    \label{tab:Z2Z2_phases}
\end{table*}

\subsection{\texorpdfstring{$\mathbb{Z}_2\times \mathbb{Z}_2$ Examples}{Z2 × Z2 Examples}}
Here we investigate the phases of the $G=\mathbb{Z}_2\times \mathbb{Z}_2$-symmetric Floquet spin chain, using the classification scheme developed above. The possible phases are listed in Table~\ref{tab:Z2Z2_phases}.
We highlight known SPT and SSB phases for this system and also point out several partially-symmetry-broken phases that, to our knowledge, have not been examined in the literature. We will consider the $G$-TO as two copies of the toric code TO, with anyons $e_i$ and $m_i$ for $i=1,2$ labelling the layer in which the anyon resides. We begin by focusing on the physical 1D system, embedded at the boundary $B_\text{phys},$ with periodic (untwisted) boundary conditions. 

We have $|G| = 4$ fully symmetry-broken phases ($H=1$) (see Equation~\ref{eqn:sym_broken_class_scheme}). The Lagrangian subgroup for this phase will be $\mathcal{M}_{\text{Rep}\, G} = \langle e_1, e_2\rangle$ and the excitation classes correspond to anyons $b = 1, m_1, m_2, m_1 m_2$. Three of these phases ($b\neq 1$) have time-crystalline signatures given by:
\begin{align}
    C_{kj;\ket{n}}^{(e_i)}(t)
    &= \bra{n} U_F^{t\dagger} W_k^{e_i\dagger} U_F^t W_j^{e_i}\ket{n} %
\\
&= \mathcal{B}(e_i, b)^t \bra{n} W_{jk}^{e_i}\ket{n}.
\end{align}
The correlation function $|\bra{n} W_{jk}^{e_i}\ket{n}|\neq 0$, approaches a constant as $|k-j|\rightarrow \infty$ (taking an absolute value to account for the signs of the LIOMs in eigenstate $\ket{n}$) and so we find period doubling so long as $e_i$ and $b$ braid non-trivially.

We then have $|\text{Cl}_0 [G] \times G | = |\mathbb{Z}_2 \times \mathbb{Z}_2 \times \mathbb{Z}_2| = 8$ fully symmetry-preserved Floquet SPT phases. The static classification $\text{Cl}_0 [G]=\mathbb{Z}_2$ corresponds to choosing $\mathcal{M} = \mathcal{M}_G$ (trivial SPT) or $\mathcal{M} = \langle e_1 m_2, m_1 e_2\rangle$ (non-trivial SPT). For each of these, we can consider the four excitation classes given by elements of, e.g., $\langle e_1 , e_2\rangle$ (we consider these as representatives of equivalence classes, under fusion with anyons from $\mathcal{M}$, c.f.~Sec.~\ref{sec:CFD}). The excitations $\overline{W}^{(b)}$ in $U_F$ can be absorbed into $B_\text{ref}$, making all these phases seem equivalent when periodic boundary conditions are imposed. However the interval-restricted Floquet unitary $U_F^{(b,\mathcal{M})}(\ell,r)$ displays non-trivial features of these phases.
In this operator, we can deform $\overline{W}^{(b)} \mapsto W_\ell^{(b)} W_r^{(b)\dagger}$ (which indicates charge pumping between the boundaries, in the terminology of Ref.~\onlinecite{KeyserlingkSondhi2016floquetSPT}).
We have two independent symmetries, given by $P_1 = \overline{W}^{(m_1)}$ and $P_2 = \overline{W}^{(m_2)}$, which can be flipped by $W_\ell^{(b)}$ and $W_r^{(b)\dagger}$ ($b\in \lbrace e_1, e_2, e_1e_2\rbrace$), %
leading to four known possibilities: $W_\ell^{(b)}$ could anti-commute with $P_1$, $P_2$, both or neither. 

What is not previously well-known is that these non-trivial Floquet SPT phases are dual time-crystal phases. Let us introduce gauge degrees of freedom in the same way as we did for the $\mathbb{Z}_2$-symmetric system, noting that $\overline{W}^{(b)}$ has support on these gauge qubits (for $b = e_1,e_2,e_1e_2$). Then dual order parameters signal long-range spatiotemporal order, as explained in Sec.~\ref{sec:ST_order_eigenstates}. We will provide further details below.

We finally have 12 remaining, previously unexplored phases. These correspond to the partial breaking of $G\rightarrow H = \mathbb{Z}_2$. There are three ways to select a $\mathbb{Z}_2$ subgroup of $\mathbb{Z}_2 \times \mathbb{Z}_2 = \langle (1,0), (0,1)\rangle$. We have $(\mathbb{Z}_2)_1 = \langle (1,0)\rangle$,  $(\mathbb{Z}_2)_2 = \langle (0,1)\rangle$ and $(\mathbb{Z}_2)_{12} = \langle (1,1)\rangle$. These respectively correspond to Lagrangian subgroups $\langle m_1 , e_2 \rangle$, $\langle e_1 , m_2 \rangle$ and $\langle m_1 m_2 , e_1 e_2 \rangle$. Each of these will have $|G| = 4$ excitation classes and no static SPT classification ($\text{Cl}_0 [\mathbb{Z}_2] = 1$). 

Take the $H = (\mathbb{Z}_2)_1$ unbroken symmetry as an example: $\mathcal{M} = \langle m_1 , e_2\rangle$. The possible anyonic excitations are (taking a single representative from each equivalence class): $1$, $e_1$, $m_2$ and $e_1 m_2$. For $U_F = \overline{W}^{(m_2)} e^{-i H_\mathcal{M}}$, we can once again find time-crystalline order with period-doubling, as discussed in Section~\ref{sec:ST_order_eigenstates}. Meanwhile, for $U_F = \overline{W}^{(e_1)} e^{-i H_\mathcal{M}}$, we can once again find dual-time-crystalline order. %
For $U_F = \overline{W}^{(e_1 m_2)} e^{-i H_\mathcal{M}}$, we find both a regular and dual time crystal. The same conclusions hold for $H = (\mathbb{Z}_2)_2$ and $H = (\mathbb{Z}_2)_{12}$.

We can realise these phases in a 1D spin chain of $N$ sites (where $N$ is even), and $N$ gauge degrees of freedom, via the following Floquet unitaries:
\begin{align}\label{eq:PSSBTC}
    U_F &= e^{-i H_1}e^{-i H_0},\\
    H_0 &= \sum_{j=1}^{N/2} J_{j}^{(1)} Z_{2j-1}\tilde{Z}_{2j-1}Z_{2j+1} + J_j^{(2)}Z_{2j}\tilde{Z}_{2j}Z_{2j+2} + \nonumber\\
    & \qquad \qquad \qquad d_j^{(1)}Z_{2j-1}\tilde{Z}_{2j-1}Z_{2j}\tilde{Z}_{2j}Z_{2j+1}Z_{2j+2},
    \\
    H_1 &= \sum_{j=1}^{N/2} h_j^{(1)}X_{2j-1} + h_j^{(2)} X_{2j} + d_j^{(2)} X_{2j-1}X_{2j}\nonumber\\
    &\quad  + w_j (Z_{2j-1}\tilde{Z}_{2j-1}X_{2j}Z_{2j+1} + Z_{2j}\tilde{Z}_{2j}X_{2j+1}Z_{2j+2}).  
\end{align}
We can add perturbations to these Hamiltonians, but here we just consider fixed-point drives. The model has closed boundary conditions, so we set $Z_{N+1} \equiv Z_1,\, Z_{N+2}\equiv Z_2$ and similarly for $X$ operators. From a SymTFT picture, the odd (even) sites in this model correspond to sites belonging to the first (second) layer of $\mathbb{Z}_2$-toric code. 
To obtain fixed-point drives that exemplify the phases of interest, we choose either $J_j^{(i)}$, $h_j^{(i)}$, $d_j^{(i)}$ or $w_j^{(i)}$
to be strongly disordered (with large distribution widths compared to any non-zero perturbations, cf.~Sec.~\ref{sec:AFFPD}),
for $i=1,2$---the associated terms in the Hamiltonians will be the LIOMs for the phase. We then choose the remainder of the parameters to be equal to $0$ or $\pi/2$. 

We summarise how these parameter choices correspond to the phases listed in Table~\ref{tab:Z2Z2_phases}. We first set $d^{(i)}_j, w^{(i)}_j = 0$. Setting $J^{(i)}$, for $i=1,2$, to be non-zero and strongly disordered
corresponds to the fully SSB phase, while setting $h_j^{(1)}$ ($h_j^{(2)}$) equal to $\pi/2$ results in a factor of $\prod_j X_{2j-1} \equiv \overline{W}^{(m_1)}$ ($\prod_j X_{2j} \equiv \overline{W}^{(m_2)}$) in $U_F$. These choices account for the $\mathbb{Z}_2\text{TC}\times\mathbb{Z}_2\text{TC}$ phases (see Table~\ref{tab:Z2Z2_phases}). Unequal time correlators:
\begin{align}
    \bra{n}Z_{2j-1}(t)Z_{2k-1}\ket{n} &= \mathcal{B}(e_1,b)^t \bra{n}Z_{2j-1}Z_{2k-1}\ket{n}\\
    \bra{n}Z_{2j}(t)Z_{2k}\ket{n} &= \mathcal{B}(e_2,b)^t \bra{n}Z_{2j}Z_{2k}\ket{n}
\end{align}
reflect which of these phases we are in, where $b = m_1, m_2,m_1m_2$ is the anyon in $\overline{W}^{(b)}$ that appears in $U_F$.

Now consider one of the partially SSB phases from Table~\ref{tab:Z2Z2_phases}, e.g., with $\mathcal{M} = \langle e_1, m_2\rangle$. To obtain this Lagrangian subgroup, we set $J_j^{(1)}$ and $h_j^{(2)}$ non-zero and strongly disordered. The other terms are set according to the choice of $b$. To set $b = m_1$, we take $h_j^{(1)} = \pi/2$ and set all other terms to zero, obtaining $\overline{W}^{(m_1)} = \prod_{j=1}^{N/2} X_{2j-1}$, while taking $b = e_2$ requires setting $J_j^{(2)} = \pi/2$ and all others zero, obtaining $\overline{W}^{(e_2)} = \prod_{j=1}^{N/2}\tilde{Z}_{2j}$. Finally, $b = m_1e_2$ ($b=\mathbf{1}$) corresponds to setting both $J_j^{(2)}$ and $h_j^{(1)}$ to $\pi/2$ (zero), with all other terms set to zero. The $m_1$ component of $b$ corresponds to the $\mathbb{Z}_2\text{TC}$ phases, while the $e_2$ components correspond to the $\mathbb{Z}_2\text{DTC}$ phases. The two unequal time correlators through which we can observe these two time-crystalline signatures are:
\begin{align}
    \bra{n}Z_{2j-1}(t)Z_{2k-1}\ket{n} &= \mathcal{B}(e_1,b)^t \bra{n}Z_{2j-1}Z_{2k-1}\ket{n}\\
    \bra{n}\tilde{X}_{2j}(t)\tilde{X}_{2k}\ket{n} &= \mathcal{B}(m_2,b)^t \bra{n}\tilde{X}_{2j}\tilde{X}_{2k}\ket{n}.
\end{align}
The second of these is built from operators that are interpreted as twisting (i.e., introducing anti-periodicity in) boundary conditions in the even-site spin chain at time 0, then un-twisting those boundary conditions at time $t$.

The second (third) set of partially-SSB phases, corresponding to $\mathcal{M} = \langle m_1, e_2\rangle$ ($\mathcal{M} = \langle e_1e_2, m_1m_2\rangle$) are obtained by taking $J_j^{(2)}$ and $h_j^{(1)}$ ($d^{(1)}_j$ and $d^{(2)}_j$) strongly disordered. $e_1$-components of $b$ are obtained by setting $J_j^{(1)} = \pi/2$, while $m_i$-components are obtained by setting $h_j^{(i)} = \pi/2$. %
(Each of these choices correspond just to particular representatives of the excitation classes and are not unique.)
The static SPT ordering is obtained by taking only $w^{(i)}$ to be strongly disordered,
while the trivial phase results from taking only $h^{(i)}$ to be strongly disordered.
We obtain the $b=e_i$ excitations for these two static orderings by setting $J^{(i)}_j = \pi/2$.

\section{Floquet Phases from Twisted Quantum Doubles}\label{sec:TQD_phases}

We now explore the possibility for $G$-symmetric Floquet phases extending beyond the classification scheme introduced above, by considering MBL boundaries of 2D twisted quantum doubles (TQDs)~\cite{Dijkgraaf_Witten_1990,Hu_2013,Ellison_2022_TQDs}. We explicitly consider the $\mathbb{Z}_2$-symmetric boundary of the double semion model~\cite{Levin-Wen,Levin_gu_DS_model,Ellison_2022_TQDs} which, as we shall argue, already captures some key general features of TQD-inspired 1D Floquet drives.
We find SSB and TTSB phases on the boundary of the double semion model that differ non-trivially from those found in other $\mathbb{Z}_2$-symmetric systems, owing to the symmetry having a non-onsite action. One non-trivial feature is that the TTSB phase displays exact $\pi$-pairing even at finite sizes under symmetry-breaking perturbations, but \textit{only} in a system with open boundaries. 

The double semion model has anyon content that is described as follows. There exist two independent non-trivial anyons (semions), $s$ and $\bar{s}$, along with an anyon resulting from their fusion, $s\bar{s}$. These obey the following fusion rules and braiding statistics:
\begin{align}
    &s\times s = \bar{s}\times \bar{s} = \mathbf{1}, \; s\times \bar{s} = s\bar{s},\\
    &e^{i\theta_s} = e^{-i\theta_{\bar{s}}} = i,\; e^{i\theta_{s\bar{s}}} = 1,\\
    &\mathcal{B}(s,s) = \mathcal{B}(\bar{s},\bar{s}) = -1,\;
    \mathcal{B}(s,\bar{s}) = 1.
\end{align}
Since the model has a single non-trivial boson,
there is only one Lagrangian subgroup, $\mathcal{M} = \lbrace \mathbf{1}, s\bar{s}\rbrace$, and hence one gapped boundary. On this boundary, $s$ and $\bar{s}$ particles are identified, and hence, e.g., the dynamics of $s$ particles
captures the low-energy physics on the boundary~\cite{JiWen2020categorical,Wen_2013_GaugeAnomalies}. 

\begin{figure}
    \centering
    \includegraphics[width=0.9\linewidth]{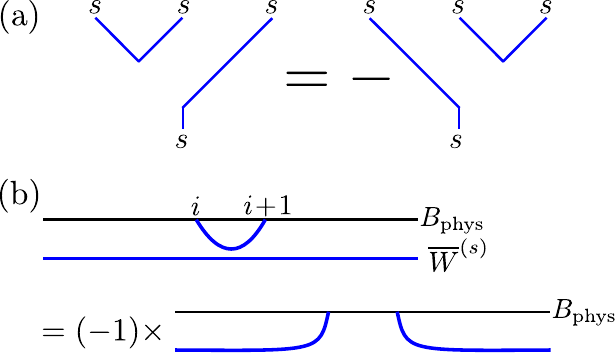}
    \caption{Reconnecting semion strings. The reconnection sign~\cite{Levin-Wen}, in panel (a), implies that starting from a $(1s1)$ configuration on three adjacent sites, hopping the $s$ to the right
    and then creating an $s$ pair to get $(sss)$ gives $(-1)$ times the state we get by moving the $s$ to the left and then creating the $s$ pair at the other two sites~\cite{JiWen2020categorical}. [The same holds for the reverse process starting from $(sss)$.] It also implies, as shown in (b), that the action of the $\mathbb{Z}_2$ symmetry $P=\overline{W}^{(s)}$ on a computational basis state $\ket{\{z_j\}}$ of $B_\text{phys}$ flips all spins and multiplies by $(-1)^{(n/2)}$ where $n$ is the number of domain walls~\cite{JiWen2020categorical}, created in pairs by the $X_i$ in $W^{(s)}_{i,i+1}$ (for various $i$) acting on a state with all spins aligned.
    }
    \label{fig:DS_model_strings}
\end{figure}

In what follows we focus on the physics of such a boundary with $\mathcal{M} = \lbrace \mathbf{1}, s\bar{s}\rbrace$, which we describe by identifying $W^{(s\bar{s})}_{j,j+1}\equiv Z_jZ_{j+1}$. In this picture, $s$ particles are analogous to the $m$ particles in the $\mathbb{Z}_2$ case in that they also correspond to domain walls on the links of the lattice. However, the operators used to create, annihilate and hop them must adhere to the properties implied by reconnecting $s$-anyon strings [see Fig.~\ref{fig:DS_model_strings}(a)]. Pair creation/annihilation and hopping operators with suitable properties are given by~\cite{JiWen2020categorical,Wen_2013_GaugeAnomalies}
\begin{align}
    a_j &= Z_{j\pm 1} X_j(1 + Z_{j-1}Z_{j+1}),\\
    t_j &= X_j(1 - Z_{j-1}Z_{j+1}),
\end{align}
where $t_j$ projects to the sector with a single domain wall on the links adjacent to site $j$ and then $X_j$ hops the $s$ particle between these links. Similarly, $a_j$ projects to the sector with zero (or two) domain walls on these links and then $X_j$ creates (or annihilates) an $s$ pair there while $Z_{j\pm 1}$ takes care of the sign in Fig.~\ref{fig:DS_model_strings} (we get the same $a_j$ for either sign in $Z_{j\pm 1}$). 

Using these ingredients, one can write a Hamiltonian for the boundary of the double semion model as~\cite{JiWen2020categorical,Wen_2013_GaugeAnomalies}:
\begin{align}
\begin{split}
    H_\text{DS,1D} = \sum_{j=1}^L -J_j Z_j Z_{j+1} - h_j^{(1)} (X_j - Z_{j-1}X_j Z_{j+1})\\
    - h_j^{(2)} Z_{j-1}(X_j + Z_{j-1}X_jZ_{j+1}),
\end{split}\label{eqn:HDS_1D}
\end{align}
where we impose periodic boundary conditions, setting $Z_{L+1} \equiv Z_1$, $Z_{0}\equiv Z_L$.

This Hamiltonian is again $\mathbb{Z}_2$ symmetric, but the symmetry operator $P$ is not the usual global spin-flip $\prod_j X_j$. Rather, it corresponds to $\overline{W}^{(s)}$ that, in terms of the 1D system, again involves $\prod_j X_j$ but also accounts for the reconnecting sign in Fig.~\ref{fig:DS_model_strings}(b). The latter is achieved if, say, each $\uparrow\to\downarrow$ domain wall in a state incurs a minus sign, as implemented by $Z_{j+1}CZ_{j,j+1}$ for each bond $j,j+1$ (here $CZ_{j,j+1}$ is the controlled-Z between the corresponding sites). Hence we now have the non-onsite (anomalous) $\mathbb{Z}_2$ symmetry~\cite{JiWen2020categorical,Wen_2013_GaugeAnomalies}
\begin{align}\label{eqn:Non-onsite_sym}
    P = \prod_{i=1}^L X_i \prod_{i=1}^L Z_{i+1} CZ_{i,i+1}.
\end{align}

We next study Floquet drives associated to TQD boundaries, i.e., for 1D systems with such anomalous $\mathbb{Z}_2$ symmetry. The drives we shall consider are analogous to the $\mathbb{Z}_2$ SG and $\pi$SG Floquet unitaries.
We choose the $J_j$ uniformly at random from an interval %
$[\frac{1}{2}\bar{J}, \frac{3}{2}\bar{J}]$,
and we similarly
choose the $h_j^{(i)}$ from interval $[\lambda \frac{1}{2}\bar{J}, \lambda \frac{3}{2}\bar{J}]$ for some perturbation strength $\lambda$. By choosing $\lambda \ll 1$, the resulting phase is many-body localised, since for the fixed-point case, the $Z_jZ_{j+1}$ provide an extensive set of LIOMs, which are transformed by quasi-local unitaries when the perturbation is turned on~\cite{Vasseur2016,MBL_spin_chains_2016}. This MBL
phase is a symmetry-breaking spin-glass (SG). %
Based on our ansatz from the previous Section, we can write two fixed-point drives based on this system:
\begin{align}
    U_\text{DS}^{\text{static}} &= e^{-iH_{\text{DS,1D}}^{(0)}},\\
    U_\text{DS}^{\text{TTSB}} &= Pe^{-iH_{\text{DS,1D}}^{(0)}},
\end{align}
where $P$ is from Equation~\ref{eqn:Non-onsite_sym} and $H_\text{DS,1D}^{(0)}$ is $H_\text{DS,1D}$ with $h^{(i)}_j=0$. 
Both have SSB order, but $U_\text{DS}^\text{TTSB}$ additionally has time-translation symmetry breaking:
\begin{align}
\begin{split}
    \bra{\lbrace n_\alpha\rbrace, p} (U_\text{DS}^{\text{TTSB}\dagger})^t Z_j (U_\text{DS}^\text{TTSB})^t Z_k \ket{\lbrace n_\alpha\rbrace, p} = \\
    (-1)^t \bra{\lbrace n_\alpha\rbrace, p} Z_j Z_k \ket{\lbrace n_\alpha\rbrace, p} = (-1)^t \cdot c(j,k,\{n_\alpha\})%
\end{split}
\end{align}
where $n_\alpha$ is the eigenvalue of the LIOM $Z_\alpha Z_{\alpha+1}$,
and $p$ is the eigenvalue of $P$, while $c(j,k,\{n_\alpha\})=\pm 1$ is a constant (in $t$) that depends on the values of the LIOMs between $j$ and $k$. Hence, we find period-doubling for this drive, which we expect to also persist
for weak perturbations, 
$h_j^{(i)}\neq 0$.
This manifests as $\pi$-splitting in the spectrum, as we are now familiar with. Meanwhile, we find eigenstate degeneracy in the spectrum of $U_\text{DS}^{\text{static}}$; both $\pi$ pairing and the degeneracy receive corrections exponentially small in the system size $L$ for small nonzero $h_j$. 

\begin{figure*}
\begin{tikzpicture}
\node[inner sep=0pt] at (-6,0) {\includegraphics[width=0.32\linewidth]{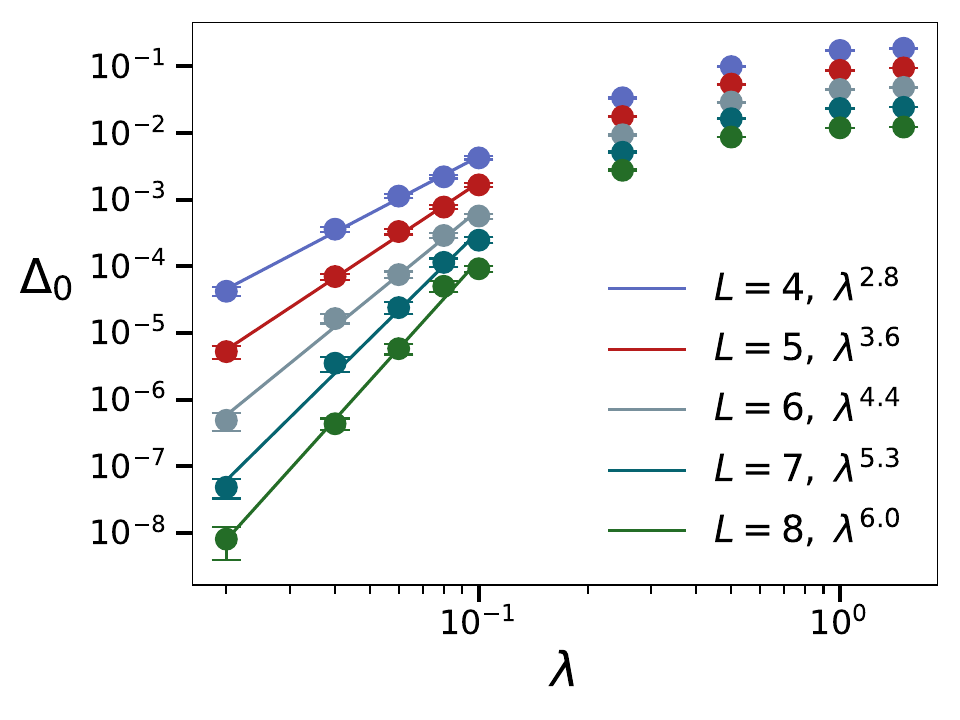}};
\node[inner sep=0pt] at (0,0) {\includegraphics[width=0.32\linewidth]{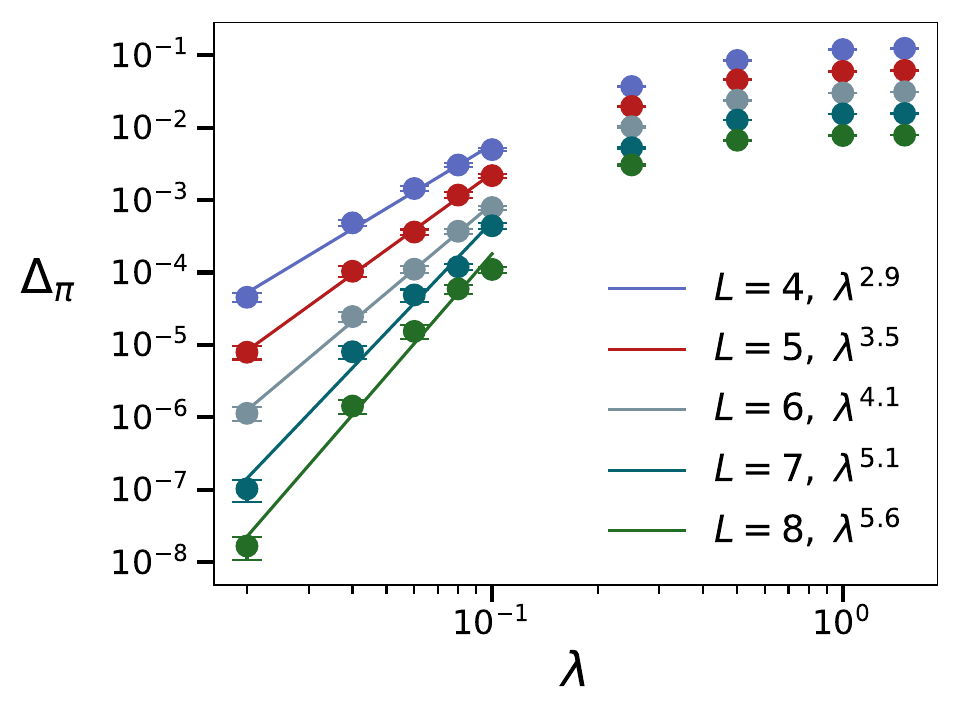}};
\node[inner sep=0pt] at (6,0) {\includegraphics[width=0.32\linewidth]{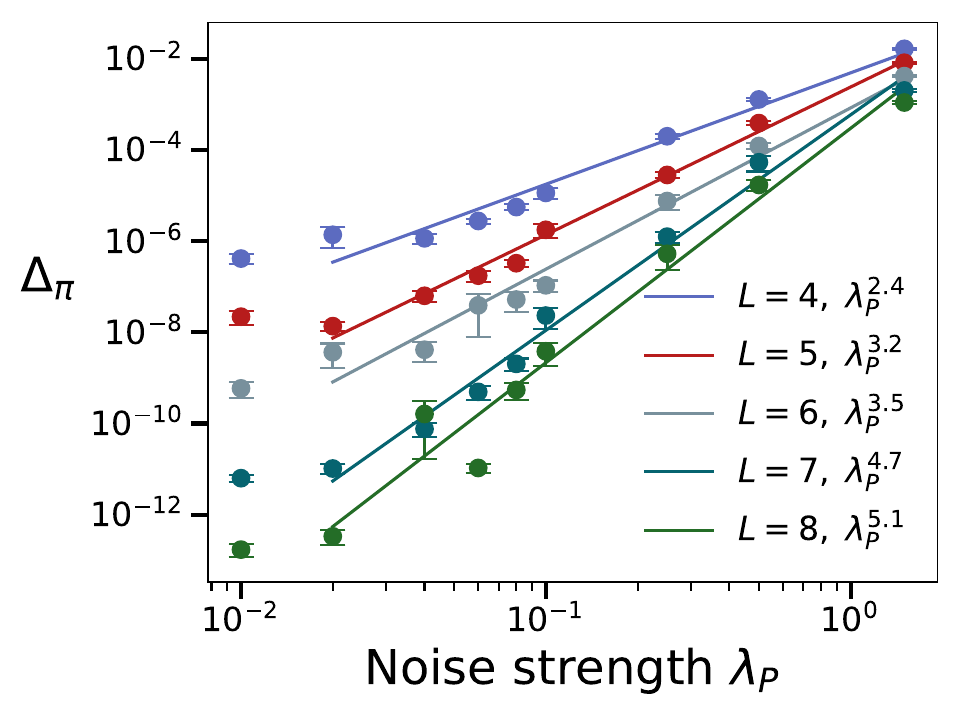}};

\node[inner sep=0pt] at (-5.5,-2.3) {(a)};
\node[inner sep=0pt] at (0.5,-2.3) {(b)};
\node[inner sep=0pt] at (6.5,-2.3) {(c)};
\end{tikzpicture}
\caption{Average eigenvalue splitting/$\pi$-splitting in the spectrum of (a) $U_\text{DS}^\text{static}$, (b) $U_\text{DS}^\text{TTSB}$, and (c) $U_\text{DS}^\text{TTSB}$ with a noisy implementation of $P$. $L$ denotes the length of the spin chain. $J_j$ were chosen uniformly at random from range $[\frac{1}{2}\bar{J}, \frac{3}{2}\bar{J}]$ for $\bar{J} = 3\pi/8$. Averages were taken over 1000 disorder realisations. (a) $\Delta_0$ is the average splitting between minimally-separated eigenvalues of the Floquet unitary.
Perturbation couplings $h_j^{(i)}$ are chosen uniformly at random from $[-\lambda \bar{J}, \lambda \bar{J}]$. 
(b) $\Delta_\pi$ is the average $\pi$-splitting in the spectrum.
Both perturbation couplings $h_j^{(i)}$ and local field strengths for $X$ and $Z$ fields are chosen at random from the same distribution as above.
The results demonstrate absolute stability in the presence of non-symmetric perturbations. 
(c) Perturbation couplings $h_j^{(i)}$ are chosen as above, for $\lambda = 0.005$. Noisy $P$ is chosen at random with noise strength $\lambda_P$,
as explained in the main text.}
    \label{fig:DS_numerics}
\end{figure*}
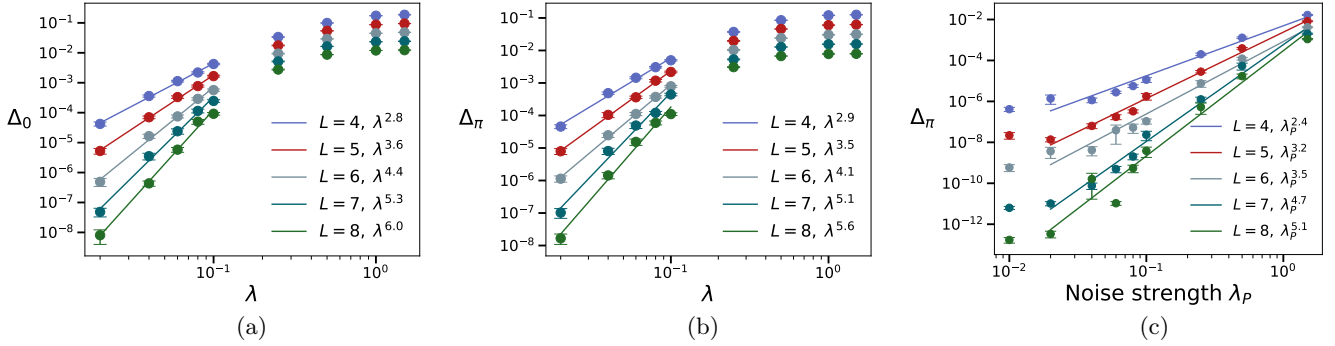

We confirm these predictions numerically in Fig.~\ref{fig:DS_numerics}(a) and (b). For Fig.~\ref{fig:DS_numerics}(a), we use the perturbed (and disordered) Hamiltonian $H_\text{DS,1D}$ from Eq.~\ref{eqn:HDS_1D}, with perturbation strength parametrised by $\lambda$.
We expect to find that, as we increase the system size $L$, the order of perturbation theory at which splitting between degenerate eigenstates occurs increases linearly, and hence we expect that the average splitting across the whole spectrum, $\Delta_0$, decreases as $\lambda^{\alpha L}$, with $\alpha > 0$ a constant. 
We do indeed numerically find an exponentially
small spectral splitting, $\Delta_0$. We calculated $\Delta_0$ via the minimum angular separation between eigenvalues averaged over the entire spectrum, and also averaged over 1000 disorder realisations.
The best-fit lines show a close to linear-in-$L$ exponent, up to some minor deviations that we attribute to finite-size effects and the effect of a relatively small number of samples.

We expect $U^{\text{TTSB}}_{\text{DS}}$ to show $\pi$-spectral pairing up to exponentially small splitting, similar to the effect observed with $\Delta_0$ above. In Fig.~\ref{fig:DS_numerics}(b), we present the $\pi$-spectral splitting
of $U_\text{DS}^\text{TTSB}$, $\Delta_\pi$, calculated in a similar way to $\Delta_0$, but taking the minimum angular separation between eigenvalues $e^{i\varepsilon_1}$ and $-e^{i\varepsilon_2}$, averaged again over the entire spectrum and over 1000 disorder realisations. Unlike the spin-glass case, however, we expect that the $\pi$-pairing affords the TTSB system an absolute stability to symmetry-breaking perturbations. Hence, we included in $H_\text{DS,1D}$ both a symmetric perturbation (from Eq.~\ref{eqn:HDS_1D}), and non-symmetric local field terms: $\sum_i w_i X_i + \sum_i m_i Z_i$. The strengths $w_i$ and $m_i$ were chosen from the same distribution as $h_j^{(i)}$. We found that the $\pi$-splitting still showed an exponential suppression with increasing $L$, despite this non-symmetric perturbation. This is to be expected for weakly symmetry-breaking perturbations in systems %
with $\pi$-paired spectra~\cite{KeyserlingkKhemaniSondhi2016stability}.

Finally, we consider the effect of implementing the symmetry $P$ in $U_\text{DS}^\text{TTSB}$ via a finite-depth circuit of noisy gates, as would be done in an experiment. Specifically, we consider the following:
\begin{align}\label{eqn:P_epsilon}
    P_\varepsilon = \prod_j \exp\left(i\frac{\pi + \varepsilon^{(1)}_j}{2} X_j\right) \prod_j \exp \left( i\frac{\pi/2 + \varepsilon^{(2)}_j}{2} Z_jZ_{j+1}\right)
\end{align}
for $\varepsilon_j^{(i)}$ randomly chosen uniformly from $[-\lambda_P \bar{J}, \lambda_P \bar{J}]$ for noise strength $\lambda_P$ and $\bar{J}$ the strength of the $J_j$ disorder in the Hamiltonian. It can be shown that
$P_{\varepsilon \rightarrow 0}$ is equal to $P$, up to a phase:
we have that $\exp(i\frac{\pi}{4}Z_jZ_{j+1}) \equiv S_jS_{j+1}CZ_{j,j+1}$ (up to a factor of $e^{i\pi/4}$) and hence all $ZZ$-exponential factors combine to give $\prod_j Z_{j+1}CZ_{j,j+1}$. 
As we see in Fig.~\ref{fig:DS_numerics}(c), for small values of $\lambda_P$, we also find an exponential suppression in $\Delta_\pi$ with $L$. Note that for $\lambda_P \lesssim \lambda$ (the strength of the symmetric perturbation in $H_\text{DS,1D}$), $\Delta_\pi$ is dominated by $\lambda$ and we find no further suppression. We also note that the noise in the $ZZ$ rotation gate did not alter the $\pi$-splitting. This can be understood from the form of Equation~\ref{eqn:P_epsilon}: the $ZZ$ part is expressible in terms of the LIOMs.

\subsection{Non-onsite symmetry and boundary conditions}

Many of these features, including absolute stability against symmetry-breaking perturbations, are similar to those found in the $\mathbb{Z}_2$ time crystals examined in Section~\ref{sec:Illustrative_example}. However, the drive $U_\text{DS}^\text{TTSB}$ is distinct from its non-anomalous $\mathbb{Z}_2$ counterpart, $U_F^{(m,\mathcal{M}_e)} = \overline{W}^{(m)} e^{-iH_e}$,
on the boundary of the toric code phase. This is because the non-onsite symmetry cannot be gauged~\cite{Wen_2013_GaugeAnomalies}, which affords the system unique properties when the symmetry is enforced, particularly for a system with open boundary conditions. %
The anomalous nature of the symmetry (which we will explain in this section) also prevents the drive $U_\text{DS}^{\text{TTSB}}$ from being able to be generated by local, symmetric Hamiltonian terms.

Indeed, one of the consequences of this non-onsite symmetry is that the symmetry fails to be a representation of $\mathbb{Z}_2$ (even a projective one) on a system with open boundary conditions. Indeed, to consider such a system, we set $J_L = h_1^{(i)} = h_L^{(i)} = 0$ in Eq.~\ref{eqn:HDS_1D} and define the following symmetry operator for the system with open boundary conditions (OBC):
\begin{align}
    P_\text{OBC} = \prod_{i=1}^L X_i \prod_{i=1}^{L-1} Z_{i+1}CZ_{i,i+1}.
\end{align}
This is equivalent to an $s$-string operator with endpoints in the double semion model. The symmetry operator does not square to the identity: $P_\text{OBC}^2 = Z_1 Z_L$. 

We now use $P_\text{OBC}$ to study the TTSB drive with open boundary conditions,
\begin{align}
    U_\text{DS,OBC}^{\text{TTSB}} = P_\text{OBC} e^{-iH_{\text{DS,OBC}}} .
\end{align}
Over two periods, the drive becomes
\begin{align}
    (U_\text{DS,OBC}^{\text{TTSB}})^2 &= Z_1 Z_L e^{-2iH_\text{DS,OBC}}.
\end{align}
Note that the value of $Z_1Z_L$ does not change over time, even away from the fixed-point (under only symmetric perturbations), since it is an exact symmetry (in fact, $Z_1Z_L$ is a product of LIOMs even away from the fixed-point).
Given a fixed value for $Z_1 Z_L = \pm 1$, we find $\pi$-pairing in the eigenstates of $U_\text{DS,OBC}^\text{TTSB}$,
since there are two distinct eigenstates of $P_\text{OBC}$ with eigenvalues differing by $-1$ (the eigenvalues are either $\pm 1$ for $Z_1Z_L = +1$ or $\pm i$ for $Z_1 Z_L = -1$). Hence each of the
$Z_1Z_L=\pm 1$ sectors, individually, exhibit $\pi$-pairing,
without necessarily any consistent pairing between the two sectors (this is because, while the eigenvalues of $P_\text{OBC}$ differ by a factor of $i$ between the sectors, the LIOM eigenvalues will also change, since the two sectors differ in the number of domain walls). 

When the symmetry is enforced, the $\pi$-pairing for the system with OBCs becomes exact. 
This is because, even at $\sim L^{\text{th}}$ order of perturbation theory, we cannot construct the spin-flip transformation with local terms, since the terms $X_1$ and $X_L$ cannot appear in the perturbation (they do not commute with $Z_1$ and $Z_L$ respectively, and so do not commute with $P_\text{OBC}^2$). Note that since $P_\text{OBC} = \pm \prod_{j=1}^{L} X_j$ in any sector with a fixed number of domain walls, a symmetric perturbation can produce a $\pi$-splitting if terms of the perturbation can multiply together to $\prod_{j=1}^{L}X_j$, which is not possible in this case but is possible in the case of periodic boundary conditions. Hence, the anomalous symmetry means that the system has exact $\pi$-pairing under symmetric perturbations, \textit{only} in the case of open boundary conditions.

We may observe $(\pi/2)$-pairing in the spectrum too. This can be obtained by omitting a bond $Z_{\lfloor L/2 \rfloor} Z_{\lfloor L/2 \rfloor + 1}$, along with any perturbation that straddles this bond, in the Hamiltonian. Doing so, all four symmetry sectors ($P_\text{OBC} = \pm 1, \pm i$) become degenerate in the spectrum of $H_\text{DS,1D}$ (since a domain wall can exist with zero energy penalty in the interval $[\lfloor L/2\rfloor,\lfloor L/2\rfloor+1]$) -- hence, we see (in the absence of perturbations) exact $\pi/2$-pairing in the spectrum of $U^{\text{TTSB}}_\text{DS,OBC}$. This pairing is, however, split at roughly linear order (independent of the system size) by non-symmetric perturbations, or symmetric perturbations that straddle the missing bond, which can cause fluctuations between domain wall sectors. In Appendix~\ref{app:DS_model_OBCs}, we demonstrate this numerically.

\section{Experimental Considerations}\label{sec:Experimental_demo}

\subsection{Requirements}

We suggest
ways in which these phases can be obtained in experiment on a quantum processor. Experimental demonstration of the TTSB phases we described requires the ability to probe connected correlations of $W_k^{(a)\dagger}(t)W_j^{(a)}$ in suitable initial states [this, essentially, amounts to probing a suitable spectral average of $C^{(a)}_{ij;\ket{n}}$ in Eq.~\eqref{eq:dyncorr}]~\cite{ippoliti2021manybody,mi2022time,Wahl_2024}. This capability was demonstrated, for the $\mathbb{Z}_2$ time crystal ($\pi$SG) phase, in Ref.~\onlinecite{mi2022time}. The Floquet unitary is implemented by precisely controlled two-qubit gates of the form $\exp(-i\frac{\phi}{4}Z_i Z_j)$ for phase $\phi$, along with single-qubit rotations. The challenge to probing phases with larger symmetry groups is that qubits may
need to be replaced by qudits, and multi-qudit gates need to be engineered.

A generic fixed-point unitary for a $G$-Floquet phase would involve terms built out of combinations of generalised Pauli operators $X^{(j)}_i, Z^{(j)}_i$ for qudits; these correspond to $W_{i,i+1}^{(m_j)}$ and $W_{i}^{(e_j)}$, respectively, in terms of the $\mathbb{Z}_k$ toric code layer decomposition via $G \simeq \times_{j=1}^p \mathbb{Z}_{k_j}$ from Sec.~\ref{sec:quantum_doubles}; see also Equations~\ref{tab:BAsummary}. 
Realising the Floquet unitary requires the ability to implement multi-qudit interactions (already on-site $G$-symmetric terms may feature  up to $p$-layer couplings).
Various experimental platforms have recently demonstrated qudit-based quantum computation \cite{Ringbauer_2022_universal, Tripathi_2025, Chi_2022}. Trapped-ion based setups in particular can realise universal gate sets for qudits and hence can also achieve multi-site interactions \cite{Ringbauer_2022_universal}.
Moreover, programmable $N$-body interactions have already been demonstrated for qubits in trapped-ion systems~\cite{Katz_PRXQ, Katz_2022_N-body}, and multi-qubit gates (involving $>2$ qubits) demonstrated in neutral atom arrays~\cite{Multi_qubit_gates_neutral_atoms_2019,Pelegrí_2022} and superconducting architectures~\cite{Superconducting_three_qubit_gate_2020,3_qubit_gate_superconducting_2025}.

Superconducting platforms have also realised qudits with $d=3,4$~\cite{Tripathi_2025}.
We expect that cross-resonance gates, along with single-qudit gates, could be used to generate the required terms in the Floquet unitaries~\cite{Qudit_gates_PRXQ_2023}. For groups with multiple cyclic factors, such as 
$\mathbb{Z}_{k_1}\times \mathbb{Z}_{k_2}$, $G$-symmetric phases can be achieved via parallel chains of qudits. %

Measuring the correlator requires controlled operations; for $G$ based on $\mathbb{Z}_2$ factors, this can directly adapt the protocols used to detect the $\pi$SG~\cite{mi2022time} or anyon interferometry in the toric code~\cite{Satzinger2021}. 
More generally, since one needs at least the target of the controlled-operation to be a qudit, or possibly a tensor product of multiple $Z_{k_j}$ factors, 
one needs the ability to implement controlled gates with multiple target qudits.

The correlator Eq.~\eqref{eq:dyncorr}, in general, involves operators with $a=(g,\alpha)$, $g\neq1$ which symmetry twist boundary conditions. This symmetry twisting is implemented via a gauge qudit [a link degree of freedom, generalising our discussion leading to Eq.~\eqref{eq:Xtilde}].

\subsection{Experimental Protocol}

We now describe an interferometric protocol to measure the expectation of $W_k^{(a)\dagger}(t)W_j^{(a)}$ in initial state $\ket{\psi}_\text{sys}$. Our protocol adapts the schemes from time-crystal and toric code experiments~\cite{mi2022time,Satzinger2021}.  
For ease of presentation, we consider $G=\mathbb{Z}_{N_1} \times \mathbb{Z}_{N_2}$ and $W_j^{(a)}\equiv Z_j^\alpha$ ($\alpha\in \text{Rep}G$). $Z_j^\alpha$ is defined to act on basis states as $Z_j^\alpha \ket{g} = \alpha(g)\ket{g}$.
We use basis states $\ket{(n_1,n_2)}$ with $n_i \in \mathbb{Z}_{N_i} \equiv \{ 0,1,2,..,N_i -1\}$ on each site. We also introduce an ancilla qubit which will be used for measurement, initialised in the state 
$(\ket{0} + \ket{1})/\sqrt{2}$. (Equivalently, we can use an $N$-dimensional qudit in state $(\ket{N-1}+\ket{0})/\sqrt{2}$ in a scheme where controlled operations are triggered when the control is in $\ket{N-1}$.) 
First, we perform a qubit-to-qudit controlled Pauli $Z^{\alpha}$ that enacts $Z^{\alpha}$ on the target qudit if and only if the control qubit is in state $\ket{1}$ (and similarly for $\ket{N-1}$ for an ancilla qudit).
We use the system qudits on site $j$ as target and the ancilla as control. The system  + ancilla hence branches into the superposition
\begin{equation}
    \frac{1}{\sqrt{2}}\left[\ket{\psi}_{\text{sys}}\ket{0}_{a} + Z^{\alpha}_j \ket{\psi}_{\text{sys}}\ket{1}_{a}\right].
\end{equation}
Next, we evolve the system + ancilla for $t$ periods of the Floquet drive $U^{(b,\mathcal{M})}_F$,
followed by a controlled $Z^{-\alpha}$ now at site $k$ (note that $Z^{-\alpha} = (Z^\alpha)^\dagger$). The system + ancilla state is now
\begin{align}
     \frac{1}{\sqrt{2}}\left[\left(U^{(b,\mathcal{M})}_F\right)^t\ket{\psi}_{\text{sys}}\right. &\ket{0}_{a} +  \nonumber\\
     &\left. Z^{-\alpha}_k\left(U^{(b,\mathcal{M})}_F\right)^tZ^{\alpha}_j \ket{\psi}_{\text{sys}}\ket{1}_{a}\right].
\end{align}
We now consider the expectation value of the ancilla operator $X_a$ (or $X_a + X_a^\dagger$ if the ancilla is a qudit of dimension $d>2$): %
\begin{equation}
\label{corr_meas}
\langle X_a\rangle = \text{Re}\left[\bra{\psi}_\text{sys}  \left(U^{(b,\mathcal{M})\dagger}_F\right)^{ t}Z^{-\alpha}_k\left(U^{(b,\mathcal{M})}_F\right)^tZ^{\alpha}_j \ket{\psi}_\text{sys}\right].
\end{equation}
This is (the real part of) the correlator we are after. Depending on the (ensemble of) $\ket{\psi}$ and $k$, $j$, one can generate different types of average correlation functions~\cite{ippoliti2021manybody,mi2022time,Wahl_2024}: for a generic highly-entangled state (preparable via a suitable quantum circuit from product states~\cite{mi2022time}) and $k=j$,
one can generate a good approximation to the uniform eigenstate average of $C^{(a)}_{kj;\ket{n}}(t)$ in Eq.~\eqref{eq:dyncorr} (the $k=j$ requirement is to avoid the random, eigenstate-dependent, signs in $\bra{n}W_{kj}^{(a)}\ket{n}$). If instead $\ket{\psi}=\left(U^{(b,\mathcal{M})}_F\right)^{t_0}\ket{n}_\text{FP}$, with $\ket{n}_\text{FP}$ an eigenstate of the fixed-point representative of the phase $U^{(b,\mathcal{M})}_F$ realises, then, for sufficiently large $t_0$, Eq.~\eqref{corr_meas} approximates the (real part of) $C^{(a)}_{kj;\ket{n}}(t)$ for $\ket{n}$ with the same LIOM eigenvalues as in $\ket{n}_\text{FP}$~\cite{Wahl_2024}.

\section{Conclusions and 
Outlook}\label{sec:Conclusion}

We have shown how SymTFT can classify 1D $G$-symmetric Floquet systems and how the unified the perspective it provides on symmetries and boundary conditions enriches the set of observable bulk signatures. 
We found that distinct 1D Floquet phases are labeled by $(b,\mathcal{M})$, where $\mathcal{M}$ is a Lagrangian subgroup of the anyons in the fictitious 2D bulk topological order, and $b$ labels an equivalence class of excitations relative to $\mathcal{M}$. 

Our general classification unifies symmetry-broken and SPT phases and allows for partially broken symmetries. Concretely, for a system that exhibits SSB from $G$ to a subgroup $H$, we showed that the Floquet phases are classified according to 
\begin{equation}
    \text{Cl}_{\text{F}}(G,H) = \text{Cl}_0(H) \times G/H \times \text{Rep}(H), 
\end{equation}
cf.~Eq.~\eqref{eqn:sym_broken_class_scheme}, where $\text{Cl}_0(H)$ is the classification of static (undriven) SPT phases with symmetry group $H$ and the $G/H \times \text{Rep}(H)$ factor component classifies the excitation equivalence classes characterising time-translation symmetry breaking.  We found that in such phases, the system could demonstrate both regular and dual time-crystalline behaviour.

We showed that SymTFT provides novel insights already for the first, $\text{Cl}_0(H)$, factor:
we described a simple and natural interpretation for the group cohomology classification of static SPT phases, based on a counting method using Lagrangian subgroups of the quantum double model. In the picture of the quantum double as multiple copies of $\mathbb{Z}_k$-toric code, these Lagrangian subgroups were described by a series of attachments of $e$-type anyons to $m$-type anyons in other layers, performed such that the anyons in the subgroup maintained their braiding properties.
This allowed for an interpretation of 1D SPT phases in terms of simple characteristics of the anyons in the corresponding Lagrangian subgroup, and surprisingly reproduces the classification of SPT phases without recourse to group cohomology arguments. %

Turning to dynamical signatures, we studied spatiotemporal correlators and showed how any phase with non-trivial excitation class exhibits spontaneous time-translation symmetry breaking, including bulk signatures even in phases (such as the $0\pi$PM) previously considered to be only boundary time crystals. These bulk signatures arise from considering order parameters (charges) and boundary conditions on equal footing---in other words, from considering order parameters both for standard and for dual symmetry breaking~\cite{JiWen2020categorical}.

Finally, we considered extensions beyond our classification scheme, to  novel phases with non-onsite symmetries; in SymTFT these correspond to boundaries of twisted quantum doubles. We performed an in-depth analysis of the Floquet drives possible at the boundary of the double semion model (twisted $G=\mathbb{Z}_2$ quantum double), which resulted in a time-translation-symmetry-breaking phase with a drive that cannot be obtained through local symmetric
Hamiltonian terms. In the model with open boundary conditions, the non-onsite symmetry failed to be a representation of $\mathbb{Z}_2$, which led to an unusual $(\pi/2)$-pairing in the spectrum in certain cases (instead of the usual $\pi$-pairing associated with $\mathbb{Z}_2$-symmetric time-crystals). We also observed complete stability of $\pi$-pairing in this system to symmetric perturbations, where this does not exist in the case with periodic boundary conditions. (We also showed evidence for absolute stability of the TTSB against non-symmetric perturbations.) 

We also discussed how the spatiotemporal correlators detecting time-translation symmetry breaking can be measured in quantum computing experiments. We described an interferometric protocol based in qudit systems and commented on the corresponding hardware requirements. %

In this work we have not focused on open boundary conditions to a great extent (outside of Sec.~\ref{sec:TQD_phases}). As stated in Sec.~\ref{sec:static_BCs}, in a static context one can introduce these by introducing a segment of a trivial phase along $B_\text{phys}$, which acts as the vacuum to which the 1D system interfaces~\cite{motamarri2024symtftequilibriumtimecrystals}.
In a driven context, this works similarly, as open boundary conditions become interfaces between the dynamical phase on $B_\text{phys}$ with the trivial phase (without TTSB). The Floquet unitary of interest will be the restriction of $U_F$ to the dynamical region of $B_\text{phys}$. 
Hence, depending on the Lagrangian subgroup in the dynamical region of $B_\text{phys}$, we may find twist defects located at the interfaces between regions~\cite{Bombin10,BarkeshliQi_PRX12,YouWen12,Barkeshli13c,BBCW19,motamarri2024symtftequilibriumtimecrystals}, since there is a correspondence between the twist defects of Abelian topological phases and the interfaces between different gapped boundaries of that phase~\cite{Kesselring_2018,Cong_2017,Cong_2017_defects}. 
These defects also play a role in dualities between phases~\cite{frohlich2004duality,lichtman2020bulk,WenWei15}. 
Future work could explore such directions, understanding how defects and domain walls provide further insight into dynamical phases via SymTFT.

Avenues for future work also include, for example, 
the study of Floquet systems with non-onsite (anomalous) symmetries beyond $G=\mathbb{Z}_2$. We believe that such systems open a rich unexplored direction. For example, we have not attempted to classify the phases that arise at the boundaries of twisted quantum doubles, where for each $G$ there is a suite of possibilities, each associated with a distinct 3-cocycle from $H^3(G,U(1))$~\cite{Ellison_2022_TQDs}. While we have considered open boundary conditions for the anomalous drive, additional features of these phases could be studied in the presence of closed but twisted boundary conditions, as we have done for the non-anomalous drives in this paper. It is possible that twisting boundary conditions may not be possible in the same way as it is for the non-anomalous drives, owing to the impossibility of gauging the symmetry, and the absence of a non-trivial dual symmetry or Kramers-Wannier duality.
Future work could also consider complete symmetry breaking in 1D models corresponding to the boundaries of non-Abelian twisted quantum doubles.
Further directions include the study of embedding in SymTFT other forms of time-crystallinity \cite{RMP_DTCs} such as Stark MBL time crystals \cite{Stark_TCs}, discrete time quasicrystals \cite{DTQCs}, or topologically ordered time crystals in higher spatial dimensions~\cite{Wahl_LIOMs_2020,Wahl_2024}.
SymTFT also applies to systems with beyond group-like symmetries. Another interesting direction would be thus to study driven phases in systems with categorical symmetries~\cite{bhardwaj2024fermionicnoninvertiblesymmetries11d, Bhardwaj_2025_gapped_non_inv_1d, bhardwaj2023categoricallandauparadigmgapped} (see also Ref.~\cite{ZTM} for forthcoming work).

\begin{acknowledgments}
CM thanks David Long for valuable discussions and acknowledges support from the Intelligence Advanced Research Projects Activity (IARPA), under the Entangled Logical Qubits program through Cooperative Agreement Number W911NF-23-2-0223. VM acknowledges support from the University of Cambridge Harding Distinguished Postgraduate Scholars Programme and EPSRC.  This work was supported by EPSRC grant EP/V062654/1.
\end{acknowledgments}

\appendix

\section{Quantum double models}\label{app:QD_models}

In this Appendix, we provide some more details on quantum double model for group $G$~\cite{kitaev2003fault,Cui_2020}. Let $\mathcal{L}$ be a lattice. As in the main text, we assign local Hilbert spaces $\mathbb{C}[G]$ to the edges of $\mathcal{L}$. We also assign an orientation to each edge, indicated by an arrow. We then define the following operator for each vertex $v$:
\begin{align}
    A_v(g)\Bigg|\fourvalentvertex{g_1}{g_2}{g_3}{g_4}{v} \Bigg\rangle &= \Bigg|  \fourvalentvertex{g_1g^{-1}}{gg_2}{gg_3}{gg_4}{v} \Bigg\rangle.
\end{align}
Note that the operator shifts the local algebra elements by multiplication by $g$ ($g^{-1}$) on the left (right) if the arrow proceeds out of (into) the vertex $v$. We similarly define the plaquette operators:
\begin{align}
\begin{split}
    B_{p}(h) \Bigg|\plaquette{g_1}{g_2}{g_3}{g_4}{p}{v}\Bigg\rangle = \\
    \delta_{h,(g_1g_2g_3g_4^{-1})} \Bigg|\plaquette{g_1}{g_2}{g_3}{g_4}{p}{v}\Bigg\rangle.
\end{split}
\end{align}
The group element appearing in the Kronecker delta is defined by proceeding counter-clockwise from the bottom-left vertex and multiplying by the edge group element or its inverse, if the arrow is aligned or anti-aligned with the path, respectively. We finally define:
\begin{align}
    A_v &= \frac{1}{|G|}\sum_{g\in G}A_v(g)\\
    B_p &= B_{p}(1),
\end{align}
with $1$ the identity element of $G$.
The Hamiltonian is then Eq.~\eqref{eq:H_qd} of the main text.

For general $G$, the anyonic excitations of this model are given by irreducible representations (irreps) of the Drinfeld (or quantum) double of the group algebra, $\mathbb{C}[G]$~\cite{kitaev2003fault}. This is the algebra generated by operators $A_v(g)$ and $B_{p}(h)$. Such irreps can be labelled by a conjugacy class $C$ of $G$ along with a representation of the centraliser $Z_G(h)=\lbrace g\in G \, | \, gh = hg\rbrace$ of an (arbitrarily chosen) representative $h$ of $C$. 

In the case of finite Abelian groups (our interest here), the set of conjugacy classes is the group $G$ itself, and the centraliser of any element in $G$ is also the entire group, $G$. Hence, the excitations are labelled by $(g,\alpha)$, for $g\in G$ and $\alpha \in \text{Rep}\, G$, the set of irreps of $G$. The vertex operators $A_{v}(g)$ can be viewed as transporting $(g,1)$ around a vertex (i.e., dual-lattice) plaquette, while the plaquette operator $B_p$ can be viewed being constructed from analogous operators $B_p(\alpha)$ transporting $(1,\alpha)$ around a plaquette  [i.e., $B_p(\alpha)$ applies representation $\alpha$ to group elements on the links around $p$; it is thus the electric-magnetic and lattice-dual-lattice dual of $A_{v}(g)$]
\begin{equation}\label{eq:B_p_rep}
B_p=\frac{1}{|\text{Rep}\, G|}\sum_\alpha B_p(\alpha).
\end{equation}
To obtain Eq.~\eqref{eq:B_p_rep}, we used the completeness relation of representation characters to express $\delta_{g,1}$. 

Let us examine the $G=\mathbb{Z}_2$ example more closely. Here, $A_v(0)$ acts as the identity while $A_v(1)$ flips all spins around vertex $v$. Hence, $A_v=(1+\prod_{j\in v} X_j)/2$ (where we take the product over edges whose endpoints include $v$). Meanwhile, $B_p$ enforces all spin values around a plaquette to sum to $0$ mod $2$. Hence, $B_p=(1+\prod_{j\in p} Z_j)/2$. This can also be understood as being $B_p = \frac{1}{2}(B_p(1) + B_p(-1))$ [with $1$ ($-1$) representing the trivial (non-trivial) irrep of $\mathbb{Z}_2$]. $B_p(-1)$ applies phase $+1$ ($-1$) to all states $\ket{0}$ ($\ket{1}$) on edges around $p$. Note that in the conventions of this appendix, $A_v$ and $B_p$ are projectors rather than Pauli operators, differing from the conventions of Sec.~\ref{sec:Illustrative_example}.
The $m$ anyon corresponds to the non-trivial $(g,1)$-type anyon---this corresponds to the end-point of a string (in the dual lattice) of operators sending $\ket{g}\mapsto \ket{g\oplus 1}$, i.e., a string of $X$ operators. The $e$ anyon corresponds to the end-point of a string of $Z$ operators, which apply the non-trivial $\mathbb{Z}_2$ irrep to qubit states.

\section{\texorpdfstring{(Partially-)SSB phases with symmetry breaking to any subgroup of $G$}{(Partially-)SSB phases with symmetry breaking to any subgroup of G}}\label{app:SSB_phases_subgroup_H}

We here discuss spontaneous symmetry breaking to a subgroup $H$ of $G$, for which we need to find a Lagrangian subgroup $\mathcal{M}_H = \lbrace(h,\beta)\rbrace$, whose $G$-components form subgroup $H$. To do so, we need to identify the irreps $\beta$ that act trivially on $H$: $\beta(h) = 1$ for all $h\in H$.

Because $G$ is Abelian, the cosets form a group, $G/H$, and we can find the irreps of that group. The irreps of $G/H$ lift to irreps of $G$ in a natural way. If $\tilde{\beta}$ is an irrep of $G/H$, then we define $\beta(g) = \tilde{\beta}(gH)$ to act as $\tilde{\beta}$ on the corresponding coset to which $g$ belongs. Thus $\beta$ is an irrep of $G$. It is a representation because $\beta(g_1g_2) = \tilde{\beta}([g_1g_2]H) = \tilde{\beta}([g_1H][g_2H]) = \tilde{\beta}(g_1H)\tilde{\beta}(g_2H) = \beta(g_1)\beta(g_2)$, and it is irreducible because it is 1-dimensional. In the second equality we use the multiplication definition of cosets, which is well defined since $G$ is Abelian, while in the third equality we use the fact that $\tilde{\beta}$ is a representation. $\beta$ clearly acts trivially on all $h\in H$. 

Hence, we can form a Lagrangian subgroup $\mathcal{M}_H = \overline{M}_H\times  \lbrace (1, \beta)\,|\, \tilde{\beta} \in \text{Rep}\, G/H\rbrace$, where $\overline{M}_H$ is a subgroup of $\mathcal{M}_H$ (which is closed under fusion and has a maximal number [i.e., $|H|$] of elements) and the Rep$\, G$ components of $\overline{M}_H$ act trivially on all $g\notin H$. To check this obeys the properties of a Lagrangian subgroup, first note that all anyons are bosons: we assume the anyons of $\overline{M}_H$ are bosons and braid trivially with one another, and we know that $(1,\beta)$ are bosons for all $\beta\in \text{Rep}\, G$, so it suffices to check that all anyons in $\overline{M}_H$ braid trivially with all $(1,\beta)$ for $\tilde{\beta}\in \text{Rep}\, G/H$, which is trivially true because $\beta(h)=1$ by definition. It is closed under fusion, since $H$ forms a subgroup of $G$, and similarly any product of two irreps of $G/H$ remains an irrep of $G/H$. And finally its size is maximal, since $|\mathcal{M}_H| = |H||\text{Rep}\, G/H| = |G|$.

\section{SPT classification}\label{app:SPT_class}

In this Appendix, we complete the classification of static SPT phases using Lagrangian subgroups.

\subsection{Roles of symmetry subgroups in SPT phases}\label{app:SPT_coset_cohomology}

Here, we show that, if there exists a subgroup $M\leq H$ of the unbroken symmetry group $H$, such that $(m,1)\in \mathcal{M}$ for all $m\in M$, then the SPT phase is labelled by a cocyle in $H^2(H/M,U(1))$. Let us recall how the group cohomology classification of SPT phases comes about. Consider elements of the symmetry group $h,g\in G$ and suppose the SPT phase has left and right boundaries, labelled $L$ and $R$. Then, since $H$ is Abelian, we have that the actions $U^g$ and $U^h$ commute: $U^gU^h(U^g)^{-1}(U^h)^{-1} = 1$. However, restricting $U^g$ and $U^h$ to the left-hand boundary (equiv. the right-hand boundary) need not produce commuting operators. Indeed, we can have $U^g_LU^h_L(U^g_L)^{-1}(U^h_L)^{-1} = e^{i\gamma_{gh}}$, so long as $U^g_RU^h_R(U^g_R)^{-1}(U^h_R)^{-1} = e^{-i\gamma_{gh}}$. This means that $U_L$ and $U_R$ form projective representations of the group $H$. In other words, we have that:
\begin{align}
    U_L^g U_L^h = \omega (g,h)U_L^{gh}
\end{align}
for some $\omega(g,h)\in U(1)$, with $g,h\in H$. This 2-cocycle $\omega$ determines the commutator via:
\begin{align}
    U_L^gU_L^h(U_L^g)^{-1} (U^h_L)^{-1} &= U_L^gU_L^h U_L^{(g^{-1})}U_L^{(h^{-1})}\\
    &= \omega(g,h)\omega(gh,g^{-1}) = e^{i\gamma_{gh}},
\end{align}
and, since $G$ (and therefore $H$) is Abelian, this relation can be inverted~\cite{Verresen_PRX_2021}.
In the above, we have set $U^{(g^{-1})}_L = (U_L^g)^{-1}$. To see why this is possible, note that we can always perform a transformation $U^g_L \mapsto \beta(g)U^g_L$ for some $\beta(g)\in U(1)$. This transforms the cocycle as:
\begin{align}
    \omega (g,h) \mapsto \tilde{\omega}(g,h) = \frac{\beta(g)\beta(h)}{\beta(gh)}\omega(g,h) \equiv \omega(g,h)d\beta(g,h).
\end{align}
Namely, the cocycles differ by a coboundary--these are considered as equivalent. Set $U_L^e = \mathds{1}$, so that $\omega(g,g^{-1}) = U_gU_{g^{-1}}$ and $\omega(g^{-1},g) = U_{g^{-1}}U_g$, so $U_{g^{-1}} = \omega (g,g^{-1})U_g^\dagger = \omega(g^{-1},g)U_g^\dagger$ and hence $\omega(g,g^{-1})= \omega(g^{-1},g)$. Therefore, take $\beta(g) = 1/\sqrt{\omega(g,g^{-1})} = 1/\sqrt{\omega(g^{-1},g)}$ so that $\tilde{\omega}(g,g^{-1}) = 1$ and hence $(U_L^g)^{-1} = U_L^{(g^{-1})}$. Furthermore, we have $U^g_L U^e_L = \omega(g,e)U^g_L$ and, since we have set $U^e_L = \mathds{1}$, $\omega(g,e) = 1$ for all $g\in H$.

If $(m,1)\in \mathcal{M}$ for all $m\in M$, then $U^m$ acts as the identity on the ground state of the system. Indeed, it may be transformed (via stabilisers of the system) away from the boundaries, to act only in the bulk of the 1D system. Hence, $U^m_L$ and $U^m_R$ must commute with $U^g_L$ and $U^g_R$, respectively. Therefore, $\omega (g,h)$ must only depend on the cosets of $M$ in $G$, since $\omega(g,h)U^{gh}_L = U^g_LU^h_L \sim U_L^gU_L^h U^m = \omega(g,m)\omega(gm,h)U^{ghm}_L= \omega(m,h)\omega(g,hm)U^{ghm}_L$, where $\sim$ represents equality within the ground space of the system. Meanwhile, $U^{ghm} = U^{gh}U^m \sim U^{gh}$, so $U^{ghm}_L \sim U^{gh}_L$. Hence, we have that $\omega(g,m)\omega(gm,h) = \omega(g,h)$ is independent of $m\in M$. By choosing $h=e$ (similarly, $g=e$), we see that $\omega(g,m) = 1$ ($\omega(m,h)=1$) for all $m\in M$. Therefore, we see that $\omega(g,h) = \omega(gm,h) = \omega(g,hm)$, for all $m\in M$. So we can define $\bar{\omega}\in H^2(H/M,U(1))$ in the obvious way, which labels the SPT order.

\subsection{Reproducing the SPT classification using Lagrangian subgroups}\label{app:SPT_Lemma}

Here we expand on the argument that reproduces the classification of SPT phases in Section~\ref{sec:Simple_SPT_Class}. Specifically, we will prove the following:
\begin{lem}\label{lem:sum_of_fracs}
    The number of solutions to $\frac{q}{k} + \frac{p}{m}\in \mathbb{Z}$ for $q\in\mathbb{Z}_k$ and $p\in\mathbb{Z}_m$ is gcd$(k,m)$.
\end{lem}
\begin{proof}
    Let $d = \text{gcd}(k,m)$ and let $m = m'd$ and $k = k'd$. Finding $q,p$ such that $\frac{q}{k} + \frac{p}{m}$ is an integer is equivalent to finding solutions to:
    \begin{align}
        mq + kp &\equiv 0 \quad (\text{mod}\, mk)\\
        \implies d(m' q + k' p) &\equiv 0 \quad (\text{mod}\, mk)\\
        \implies m' q + k' p &\equiv 0 \quad (\text{mod}\, m' k' d).
    \end{align}
    Ignoring the trivial solution for now ($q=p=0$), it must be the case that:
    \begin{align}\label{eqn:app_c_integers}
        m' q + k' p = m' k' d,
    \end{align} 
    since $q<k$ and $p<m$ means that $m'q + k'p < 2m' k' d$. This implies that $m'$ divides $(m' k' d - k'p)$ evenly, and so $p = m' l$ for some integer $l$ (since $k'$ and $m'$ are coprime). Therefore, substituting this value for $p$ into Eq.~\ref{eqn:app_c_integers}, we find $q = k' (d-l)$. Since $l$ is an integer, we see that there are $d-1$ solutions for $0<q<k$ and $0<p<m$. Including the trivial one, there are therefore $d$ solutions in total. 
\end{proof}

Note that we can produce the solutions by taking $0\leq l \leq d$, $p= m'l$ and $q= k'(d-l)$. This is useful as it determines the Lagrangian subgroups corresponding to the SPT phases.

\subsection{From Lagrangian subgroup to cohomology}\label{app:SPT_Lagrangian_subgroup_to_cohomology}

Here, we demonstrate how to obtain the projective representation associated with a static SPT phase from the associated Lagrangian subgroup in SymTFT. Recall that the Lagrangian subgroup associated with an SPT phase can be written as $\mathcal{M} = \langle m_1 e_{j_1}^{p_1}, m_2 e_{j_2}^{p_2},\ldots \rangle$, where $j_i \neq i$ (see Section~\ref{sec:Simple_SPT_Class}). The $m$ anyons are associated with the generators of $G$ and hence let us consider the on-site action of the symmetry, $U^{m_i}$. Then consider the following graphical series of transformations, where $\sim$ denotes equality up to application of elements in $\mathcal{M}$ and deformation of string operators in spacetime:
\begin{align}
    U^{m_j}U^{m_i}&(U^{m_j})^{-1}(U^{m_i})^{-1}\\[0.5em]
    &= \adjincludegraphics[scale=0.52,valign=c]{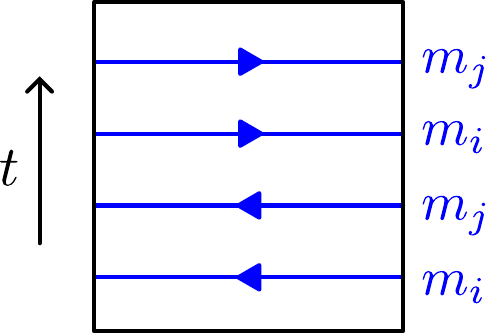}\\[0.5em]
    &\sim \adjincludegraphics[scale=0.52,valign=c]{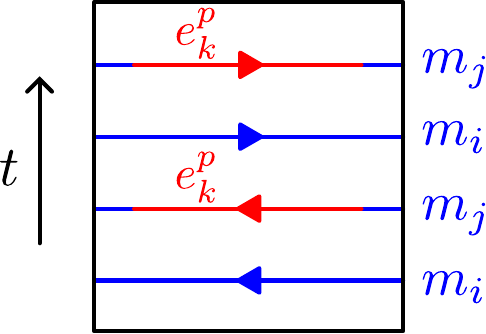}\\[0.5em]
    &\sim \adjincludegraphics[scale=0.52,valign=c]{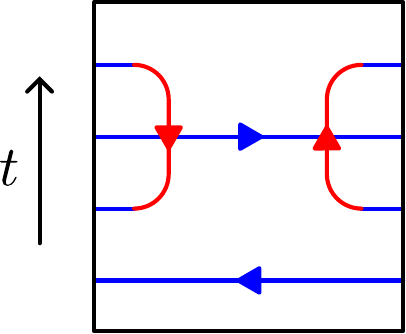}\\[0.5em]
    &\sim \adjincludegraphics[scale=0.52,valign=c]{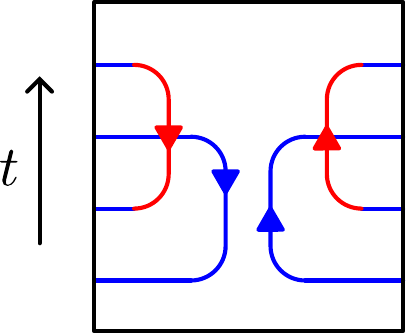}.\label{eqn:graph_projective_rep_final}
\end{align}
In the above, the 1D system runs horizontally, while time runs vertically. We can deform certain $m$ anyon strings (in the bulk) into $e$ anyon strings via $\mathcal{M}$ elements, and we have supposed that $m_j e_k^p\in \mathcal{M}$. In Equation~\ref{eqn:graph_projective_rep_final}, we can readily identify four intersecting loops, two confined to the left of the system and two to the right of the system. Each of these is equivalent to the vacuum up to a phase, since the overall operator is equivalent to the application of symmetry operators to the (non-degenerate) ground state. 

In these loops, we can observe an $m$ anyon braiding with an $e$ anyon. We can identify the left/right loops with the $U_L$/$U_R$ actions, respectively, obtaining:
\begin{align}
    U_L^{m_j}U_L^{m_i}(U_L^{m_j})^{-1}(U_L^{m_i})^{-1} &= \exp (i\theta_{e_k^p, m_i}),\\
    U_R^{m_j}U_R^{m_i}(U_R^{m_j})^{-1}(U_R^{m_i})^{-1} &= \exp (i\theta_{\overline{e_k^p}, m_i})\\
    &= \exp (-i\theta_{e_k^p, m_i}).
\end{align}
Hence, the projective representation $U_L$ (equivalently, $U_R$) is obtained from the braiding relation between $e_k^p$, which is attached to anyon $m_j$ in $\mathcal{M}$, and $m_i$.

\section{Extending Disorder Correlators Away from Fixed Points}\label{app:extending_from_fixed_points}

Here we see how the disorder correlators decay once we have perturbed the drive away from fixed points [cf. Sec.~\ref{sec:ST_order_eigenstates}]. For such a drive,
the LIOMs will in general have exponentially decaying tails. We will refer to the LIOMs at this point as $D_j$, related to the fixed-point LIOMs by a quasi-local unitary $\mathcal{U}$. Now suppose we are perturbing away from a unitary $U_F^{(b,\mathcal{M})}$ with $b\notin \mathcal{M}$ (see Eq.~\ref{eqn:fp_drives}). We will imagine the perturbation may, in general, not be symmetric or dual symmetric. In such cases we will still be able to find emergent (dual) symmetries, related to the original symmetry by the same unitary transforming the LIOMs, $\mathcal{U}$. In particular, let us define the emergent operator $\widetilde{W}^{(b)} = \mathcal{U}^\dagger \overline{W}^{(b)}\mathcal{U}$. Then, since the $D_j$ operators retain the same commutation relations away from the fixed point, the perturbed Floquet unitary can be written as:
\begin{align}
    U_{F,\text{pert}} &= \widetilde{W}^{(b)}\exp[-i\mathcal{H}({D_j})]\\
    &= \left( \prod_k \tau^x_k\right) \exp[-i\mathcal{H}(D_j)]
\end{align}
with $\mathcal{H}$ some function of the LIOMs,
and we write $\widetilde{W}^{(b)}$ as a product of operators, $\tau^x_j$ ($\mathcal{U}$-conjugated versions of local operators making up $\overline{W}^{(b)}$), that do not individually commute with the LIOMs.
Suppose that the interval $I = [\ell,r]$ is much larger than the localisation length for the system, $\xi$, but much smaller than the full size of the system. Truncating $U_{F,\text{pert}}$ to this interval therefore approximately does not affect its form, in terms of the $D_j$ and $\tau^x_j$ operators, deep in the bulk. We can therefore write $U_{F,\text{pert}}(\ell,r)$ in the following way:
\begin{align}\label{eq:UFtrunc}
    U_{F,\text{pert}}(\ell,r) &= V_\ell V_r \left(\prod_{k\in I} \tau^x_k\right) \exp[-i\mathcal{H}(D_{j\in I})] 
\end{align}
where the middle terms are written explicitly in terms of the $D_j$ and $\tau^x_j$ operators, and $V_\ell$ and $V_r$ are unitary operators exponentially localised to the boundaries of $I$. Since the $D_j$ and $\tau^x_j$ operators have only quasi-local support, truncating $U_{F,\text{pert}}$ to $I$ is not equivalent to selecting only those $D_j$ and $\tau^x_j$ with $j\in I$. Doing so results in an operator that differs from $U_{F,\text{pert}}(\ell,r)$ by some amount exponentially localised to $\ell$ and $r$---the operators $V_\ell$ and $V_r$ ``correct" this error.

Now let us consider the disorder correlator associated to TTSB. We can label the eigenstates of $U_{F,\text{pert}}$ with the eigenvalues of the LIOMs, $d_j$,
and the eigenvalues of some (potentially emergent) symmetry and dual symmetry operators, which we collectively denote $\tilde{w}$: $\ket{n} \coloneqq \ket{\{ d_j\}, \tilde{w}}$. The disorder correlator then becomes:
\begin{align}
\begin{split}
    &\bra{n}U_{F,\text{pert}}(\ell,r)\ket{n} \\
    &= \bra{\{ d_j\}, \tilde{w}} V_\ell V_r \left(\prod_{k\in I} \tau^x_k\right) \exp[-i\mathcal{H}(D_{j\in I})] \ket{\{ d^z_j\}, \tilde{w}}
\end{split}\\
&= \exp[-i\mathcal{H}(d^z_{j\in I})]\bra{\{ d^z_j\}, \tilde{w}}V_\ell V_r \left(\prod_{k\in I} \tau^x_k\right)\ket{\{ d^z_j\}, \tilde{w}}.\label{eqn:disorder_corr_calc}
\end{align}
Now, since the $\tau_k^x$ flip some of the eigenvalues $d_j$ near the boundaries of $I$ (owing to $b\notin \mathcal{M}$), $\bra{\{d_j\}, \tilde{w}} \prod_{k\in I}\tau_k^x \ket{\{ d^z_j\}, \tilde{w}} = 0$, and hence the only way for Equation~\ref{eqn:disorder_corr_calc} to be non-zero is for there to be some term in $V_\ell V_r$ proportional to $\prod_{k\in I}(\tau_k^x)^\dagger$ (up to a product of LIOMs and symmetry/dual symmetry operators), cancelling the effect of the off-diagonal product of $\tau^x_k$. However, due to the exponential localisation of $V_\ell V_r$ to the boundaries of $I$, any such term will be exponentially small, i.e., will have an operator norm that is $\sim e^{-|\ell-r|/\xi}$. This explains the exponential decay of disorder correlators in the TTSB case (whenever $b\notin \mathcal{M}$), as seen in Equation~\ref{eqn:Disorder_corr_two_cases}.

\section{Structure of Lagrangian subgroups \& their excitation classes}\label{app:excitation_classes}
We consider a Lagrangian subgroup $\mathcal{M}$, with corresponding unbroken symmetry group $H$. From Appendix~\ref{app:SSB_phases_subgroup_H}, we can identify a decomposition of the lagrangian subgroup $\mathcal{M}_H = \overline{M}_H \times \mathcal{M}_{\text{Rep}\, G/H}$, where $\mathcal{M}_{\text{Rep}\, G/H} = \lbrace (1,\beta)\, |\, \beta\in \text{Rep}\, G,\, \beta(h)=1\, \forall \, h\in H\rbrace$ and $\overline{M}_H$ has $|H|$ elements whose $G$-components generate $H$.  %

With a view on the structure of $\mathcal{M}$ we can now classify excitations. We have an excitation for any $b \notin \mathcal{M}$. However, if $a \in \mathcal{M}$, $b$ and $b \times a$ are equivalent excitations. Thus excitations fall into equivalence classes; distinct excitations correspond to distinct classes in $\mathcal{A}_G/\mathcal{M}$. The number of these is $|\mathcal{A}_G/\mathcal{M}| = |\mathcal{A}_G|/|\mathcal{M}|=|G|^2/|G|= |G|$.

It is useful to consider two particular variants of excitations. First, we have $b=(1,\beta) \in \mathcal{M}_{\text{Rep}\, G}$ such that $\beta(h) \neq 1$ for some $h \in H$. Since $b\times (1,\alpha)$ [with $(1,\alpha) \in \mathcal{M}_{\text{Rep}\, G/H}$] corresponds to an equivalent excitation, we have equivalence classes $b \mathcal{M}_{\text{Rep}\, G/H}$. Denoting the group these form by $\mathcal{M}^{\text{exc}}_{\text{Rep}\, G}$, we have $\mathcal{M}^{\text{exc}}_{\text{Rep}\, G} = \mathcal{M}_{\text{Rep}\, G} / \mathcal{M}_{\text{Rep}\, G/H}$. Since $\mathcal{M}_{\text{Rep}\, G/H} \simeq \text{Rep} (G/H)$, we have $\mathcal{M}^{\text{exc}}_{\text{Rep}\, G} \simeq \text{Rep}(G)/\text{Rep} (G/H)\simeq \text{Rep}(H)$. %

Next, for any $g \notin H$ we have $a' = (g,1) \notin \mathcal{M}$. Now $a' \overline{M}_H$ forms an equivalence class of excitations. Thus, if we let $\mathcal{M}_{G}= \lbrace(g,1) \,| \, g \in G\rbrace$, the group $\mathcal{M}^{\text{exc}}_{G}$ of inequivalent excitations is $\mathcal{M}^{\text{exc}}_{G} = \mathcal{M}_{G} / \overline{M}_{H} \simeq G/H$. %

The above two classes of excitations can be fused and together yield $|\mathcal{M}^{\text{exc}}_{G}| |\mathcal{M}^{\text{exc}}_{\text{Rep}\, G}| = |G|$ distinct excitations, thus exhausting $\mathcal{A}_G / \mathcal{M}$. 
That is, we obtained the decomposition $\mathcal{A}_G / \mathcal{M} = \mathcal{M}^{\text{exc}}_{G} \times \mathcal{M}^{\text{exc}}_{\text{Rep}\, G}$ in the sense of a fusion group. For the purpose of the classification of Floquet drives it is useful to note that $\mathcal{A}_G / \mathcal{M} \simeq G/H \times \text{Rep}(H)$.

\section{Double Semion Boundary Model with OBCs}\label{app:DS_model_OBCs}

\begin{figure}
\centering
\includegraphics[width=\linewidth]{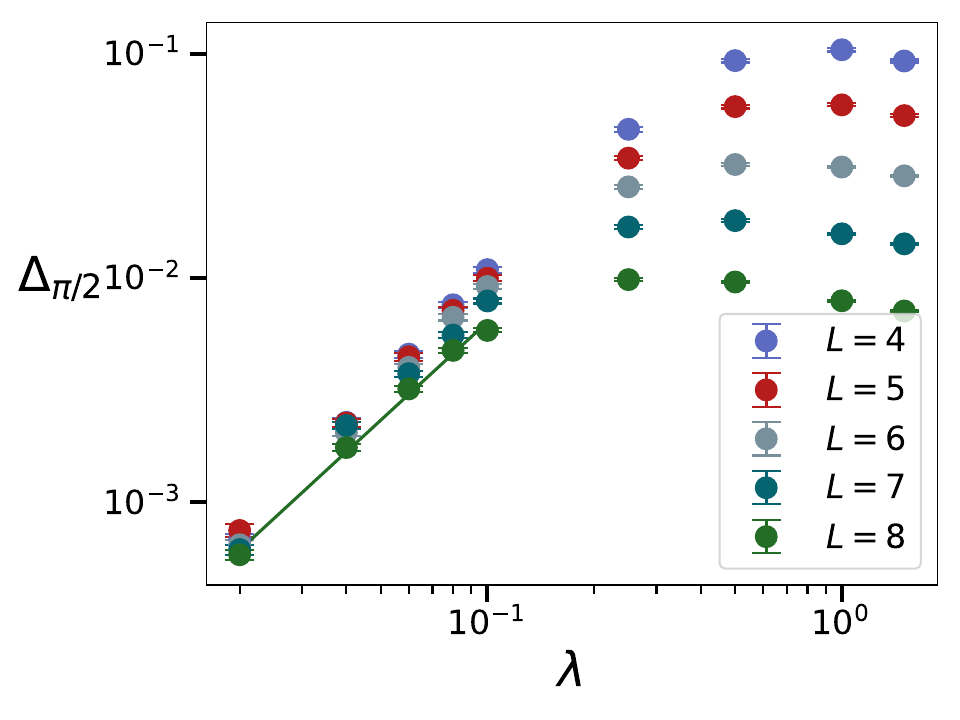}
\caption{Average splitting between ($\pi/2$)-paired eigenvalues of $U^{\text{TTSB}}_{\text{DS,OBC}}$ with all terms in $H_\text{DS,OBC}$ that straddle bond $[\lfloor L/2\rfloor, \lfloor L/2\rfloor + 1]$ being omitted. $\lambda$ is the strength of both symmetric perturbations (not including those that straddle the bond) and non-symmetric local $X$ and $Z$ fields. The $(\pi/2)$-splitting is averaged over 1000 disorder realisations. The $L=8$ data points
with $\lambda \leq 0.1$ are used to produce the fit shown, which is $\Delta_{\pi/2} \sim \lambda^{1.5}$.\label{fig:pi_2_splitting_missing_bond}}
\end{figure}

Here, we provide data for the model considered in Section~\ref{sec:TQD_phases}, with open boundary conditions and a missing bond, to allow for $(\pi/2)$-pairing in the spectrum. We first observed that all $\pi$-splitting is smaller than $10^{-15}$ in all cases tested, whenever only a symmetric perturbation is included, and the splitting did not depend on $\lambda$ -- this suggests no splitting induced by this symmetric perturbation (we take the finite splitting to be resulting from machine imprecision). Similarly, when we only include symmetric perturbations and omit all terms in the Hamiltonian that straddle bond $[\lfloor L/2\rfloor, \lfloor L/2\rfloor + 1]$, we observe the same stability in $(\pi/2)$-pairing---it is not split at all by the perturbation.

Meanwhile, if we include symmetric perturbations that straddle the bond, and/or if we include non-symmetric perturbations, the $(\pi/2)$-pairing is split at an order in $\lambda$ independent of the system size. This is shown in Fig.~\ref{fig:pi_2_splitting_missing_bond}, where for small $\lambda$, the decay in $\Delta_{\pi/2}$ with decreasing $\lambda$ is roughly independent of $L$, and we observe that $\Delta_{\pi/2} \sim \lambda^{1.5}$. This suggests the $(\pi/2)$-pairing is not stable.

\end{document}